\documentclass[11pt,letterpaper,english]{article}
\usepackage[letterpaper,margin=1in]{geometry}
\usepackage[utf8]{inputenc}

\usepackage{titlesec}
\titleformat{\subparagraph}[runin]
{\normalfont\bfseries}
{}{}{}[.]

\titlespacing*{\subparagraph}{\parindent}{0.3\topsep}{1em}

\usepackage{float}

\usepackage{ifthen}
\usepackage{shortversion}
\newappendixcollect[1]{proof}{\proof}{\endproof}{\proof[Proof of #1]}{\endproof}
\newappendixcollect[1]{figure}{\figure[#1]}{\endfigure}{\figure[H]}{\endfigure}
\newappendixcollect{example}{\example}{\endexample}{\example}{\endexample}

\usepackage{framed}
\usepackage{colortbl}

\usepackage{paralist} 
\usepackage{subcaption} 
\usepackage[font=small,indention=0pt,margin=0pt]{caption}

\usepackage{varioref} 
\usepackage{hyperref}
\usepackage{todonotes}

\usepackage{algorithm2e}

\usepackage{amsmath}
\usepackage{amssymb}
\usepackage{amsthm}

\usepackage{array}

\usepackage{nicefrac}
\usepackage{bm}

\newtheoremstyle{mystyle}
  {0.5\topsep} 
  {0.5\topsep} 
  {} 
  {} 
  {\bfseries} 
  {.} 
  {.5em} 
  {} 

\theoremstyle{mystyle}
\newtheorem{theorem}{Theorem}
\newtheorem{lemma}[theorem]{Lemma}
\newtheorem{corollary}[theorem]{Corollary}
\theoremstyle{definition}
\newtheorem{definition}[theorem]{Definition}
\newtheorem{example}{Example}
\newtheorem*{remark}{Remark}

\newcommand{\R}{\mathbb{R}}
\newcommand{\N}{\mathbb{N}}

\renewcommand{\vec}[1]{\boldsymbol{\mathbf{#1}}}

\newcommand{\sgn}{\mathop{\mathrm{sgn}}}			
\newcommand{\diag}{\mathop{\mathrm{diag}}}			

\newcommand{\argmin}{\mathop{\mathrm{argmin}}}		
\newcommand{\argmax}{\mathop{\mathrm{argmax}}}		

\newcommand{\vspan}{\mathop{\mathrm{span}}}			

\newcommand{\lexle}{\prec_{\textsc{lex}}}			
\newcommand{\lexleq}{\preceq_{\textsc{lex}}}		

\newcommand{\sCurveP}{\lambda}
\newcommand{\sCurveI}{i}
\newcommand{\sCurveB}{\lambda}

\newcommand{\incidenceMatrix}{\vec \Gamma}
\newcommand{\incidenceM}{\incidenceMatrix}
\newcommand{\incidenceV}{\vec{\gamma}}
\newcommand{\incidenceVs}{\hat{\vec{\gamma}}}

\newcommand{\commodity}{j}
\newcommand{\numCommodities}{k}
\newcommand{\setCommodities}{K}

\newcommand{\source}{s}
\newcommand{\sink}{t}

\newcommand{\vertex}{v}
\newcommand{\numVertices}{n}

\newcommand{\edge}{e}
\newcommand{\altEdge}{{e'}}

\newcommand{\pathVar}{P}
\newcommand{\altPath}{Q}
\newcommand{\setPaths}{\mathcal{P}}

\newcommand{\mcFlowVec}{\vec X}
\newcommand{\mcFlow}[1][\edge, \commodity]{x_{#1}}

\newcommand{\cFlowVec}[1][\commodity]{\vec x^{#1}}

\newcommand{\totalFlow}{z}
\newcommand{\totalFlowVec}{\vec z}

\newcommand{\idUnused}{\text{un}}
\newcommand{\idInactive}{\text{in}}

\newcommand{\indicUnused}{\vec{I}_{\commodity, \idUnused}}
\newcommand{\indicInactive}{\vec{I}_{\commodity, \idInactive}}

\newcommand{\regionVec}{\vec{t}}
\newcommand{\regionV}{\regionVec}
\newcommand{\regionI}{t}

\newcommand{\neighboringR}{\mathfrak{N}}

\newcommand{\aM}{\vec{A}}
\newcommand{\cM}{\vec{C}}

\newcommand{\potentialV}{\vec{\pi}}

\newcommand{\flowV}{\vec{x}}
\newcommand{\flow}{x}

\newcommand{\excessV}{\vec{y}}
\newcommand{\excessVs}{\hat{\vec{y}}}

\newcommand{\laplaceM}{\vec{L}}

\newcommand{\laplaceMsI}{\hat{\vec{L}}^{-1}}

\newcommand{\unitVec}{\vec u}

\newcommand{\identityM}{\vec{I}}

\newcommand{\shermanOp}{\boxtimes}

\newcommand{\perturbed}[1]{\breve{#1}}
\newcommand{\potentialBPd}{\perturbed{\sigma}}
\newcommand{\flowBPd}{\perturbed{\tau}}
\newcommand{\potentialVp}[1][]{
	\ifthenelse { \equal{#1}{} } 
	{	\perturbed{\potentialV}		}
	{	\perturbed{\potentialV}^{#1} (\delta)	}
}
\newcommand{\potentialDp}[1][]{
	\ifthenelse { \equal{#1}{} } 
	{	\Delta \perturbed{\potentialV}		}
	{	\Delta \perturbed{\potentialV}^{#1}	}
}
\newcommand{\flowVp}[1][]{
	\ifthenelse { \equal{#1}{} } 
	{	\perturbed{\flowV}		}
	{	\perturbed{\flowV}^{#1} (\delta)	}
}
\newcommand{\flowDp}[1][]{
	\ifthenelse { \equal{#1}{} } 
	{	\Delta \perturbed{\flowV}		}
	{	\Delta \perturbed{\flowV}^{#1}	}
}
\newcommand{\flowdp}[1][]{
	\ifthenelse { \equal{#1}{} } 
	{	\Delta \perturbed{\flow}		}
	{	\Delta \perturbed{\flow}^{#1}_{\edge}	}
}
\newcommand{\flowp}[1][]{
	\ifthenelse { \equal{#1}{} } 
	{	\perturbed{\flow}		}
	{	\perturbed{\flow}^{#1}_{\edge} (\delta)	}
}

\newcommand{\regionp}[1][]{
	\ifthenelse { \equal{#1}{} } 
	{	\perturbed{R}	}
	{	\perturbed{R}^{#1} (\delta)	}
}
\newcommand{\epsilonp}[1][]{
	\ifthenelse { \equal{#1}{} } 
	{	\perturbed{\epsilon}	}
	{	\perturbed{\epsilon}^{#1} (\edge, \delta)	}
}

\usepackage{tikz}
\usetikzlibrary{shapes.geometric, arrows, positioning, angles, quotes}
\usetikzlibrary{patterns}
\usetikzlibrary{decorations.pathmorphing}
\usetikzlibrary{fit,calc}

\tikzset{
arrow/.style={-latex},
>=stealth'}

\tikzstyle{vertex} = [circle, draw=black, thick, minimum height=1.5em, text centered]
\tikzstyle{solid} = =[circle, fill,inner sep=1.5pt,outer sep=0pt]

\newcommand*{\tikzmk}[1]{\tikz[remember picture,overlay,] \node (#1) {};\ignorespaces}

\newcommand{\nameboxit}[2]{\tikz[remember picture,overlay]{\node[fill=#1,fill opacity=0.25,text opacity=1,fit={($(A)+(0,0.8\baselineskip)$)($(B)+(.95\hsize,0.8\baselineskip)$)},align=right,text height=0.8\baselineskip,inner sep=0pt] {\colorbox{#1}{\bf\sc #2}};}\ignorespaces}

\colorlet{pink}{red!50}
\colorlet{blue}{cyan!60}
\colorlet{green}{green!60}
\colorlet{yellow}{yellow!60}

\definecolor{color1}{RGB}{141,211,199}
\definecolor{color2}{RGB}{255,255,179}
\definecolor{color3}{RGB}{190,186,218}
\definecolor{color4}{RGB}{251,128,114}
\definecolor{color5}{RGB}{128,177,211}
\definecolor{color6}{RGB}{253,180,98}
\definecolor{color7}{RGB}{179,222,105}
\definecolor{color8}{RGB}{252,205,229}
\definecolor{color9}{RGB}{217,217,217}

\definecolor{darkgreen}{RGB}{0,150,42}

\title{Computing all Wardrop Equilibria parametrized by the Flow Demand\footnote{This research was carried out in the framework of {\sc Matheon} supported by Einstein Foundation Berlin.}}

\author{
	Max Klimm\footnote{Humboldt-Universit\"at zu Berlin, Germany, \texttt{max.klimm@hu-berlin.de}}
	\and
	Philipp Warode\footnote{Humboldt-Universit\"at zu Berlin, Germany, \texttt{philipp.warode@hu-berlin.de}}
}

\date{}

\begin{document}
\thispagestyle{plain}
\maketitle

\begin{abstract}
We develop an algorithm that computes for a given undirected or directed network with flow-dependent piece-wise linear edge cost functions \emph{all} Wardrop equilibria as a function of the flow demand. Our algorithm is based on Katzenelson's homotopy method for electrical networks. The algorithm uses a bijection between vertex potentials and flow excess vectors that is piecewise linear in the potential space and where each linear segment can be interpreted as an augmenting flow in a residual network. The algorithm iteratively increases the excess of one or more vertex pairs until the bijection reaches a point of non-differentiability. Then, the next linear region is chosen in a simplex-like pivot step and the algorithm proceeds. We first show that this algorithm correctly computes all Wardrop equilibria in undirected single-commodity networks along the chosen path of excess vectors. We then adapt our algorithm to also work for discontinuous cost functions which allows to model directed edges and/or edge capacities. Our algorithm is output-polynomial in non-degenerate instances where the solution curve never hits a point where the cost function of more than one edge becomes non-differentiable. For degenerate instances we still obtain an output-polynomial algorithm computing the linear segments of the bijection by a convex program. The latter technique also allows to handle multiple commodities.
\end{abstract}

\newpage
\setcounter{page}{1}

\section{Introduction}

In Wardrop's equilibrium model \cite{wardrop1952}, we are given a network with flow-dependent cost functions and a set of commodities $j=1,\dots,k$ with fixed travel demand $q_j > 0$  and source-sink pairs $s_j, t_j$. A multi-commodity flow is at equilibrium when for each commodity all used paths have the same cost and no unused path has smaller cost. For fixed travel demands, a Wardrop equilibrium can be computed efficiently via convex optimization techniques under mild conditions on the cost function~\cite{beckmann1956,dafermos1969}.
In this paper we are interested in computing \emph{all} Wardrop equilibria in a network as a function of the travel demand, i.e., we want to compute a  function $f : \R_{\geq 0} \to \R_{\geq 0}^{m \times k}$ such that the multi-commodity flow $\vec f(\lambda) = (f_{e,k} (\lambda) )_{e \in E, k=1,\dots,k}$ is a Wardop equilibrium for the demand vector $\vec q = \lambda\vec{w} \in \R^k$ for all $\lambda \geq 0$ where $\vec w \in \R^k_{\geq 0}$ is arbitrary. This is useful, e.g., when analyzing traffic networks with unknown demands as the knowledge of the functions $f_{e,k}(\lambda)$ allows to answer relevant questions to a network designer such as the maximal congestion of an edge over all possibly load factors $\max_{\lambda \geq 0} l_e(\sum_{k=1}^j f_{e,k}(\lambda))$. The paradoxes given by Breass~\cite{breass1986} and Fisk~\cite{fisk1979} show the flow of an edge may be non-monotonic in the demands of one of the commodities, so that the maximum congestion on an edge needs not be attained for the maximum demand.
Since already for networks of parallel edges, the evolution of the Wardrop equilibrium as a function of the demand is governed by a rather intricate system of differential equations~\cite{harks2015}), we cannot expect to represent the functions $f_{e,k}(\lambda)$ succinctly for general edge cost function $l_e$.
To circumvent this representation issue, in this paper we restrict ourselves to piece-wise linear cost functions $l_e$. Since the set of piece-wise linear functions is dense in the set of continuous functions, continuous linear cost functions can be approximated with arbitrarily small error. Moreover, piece-wise linear functions allow to be fitted easily to given measurements of travel times in a network. However, the complexity of the problem naturally increases with the number of breakpoints of the piece-wise linear functions.

\subparagraph{Our results} 
We develop an algorithm that computes all Wardrop equilibria for a given undirected or directed network.
Our algorithm builds upon and extends a homotopy methods for computing electrical flows in electrical networks with piece-wise linear resistances developed by Katzenelson~\cite{katzenelson1965}. To establish the connection between Wardrop equilibria and electrical flows, we use characterization of Wardrop flows by vertex potentials. For undirected networks (such as electrical networks) this characterization allows to construct a bijection between vertex potentials and the excess flows of the corresponding Wardrop equilibrium. For piece-wise linear cost functions, this bijection is piece-wise linear and can be expressed by the Laplacian of the network where the conductances of an edge correspond to the reciprocal of the slope of the linear segment.
The different piece-wise linear parts of the bijection impose a division of the space of vertex potentials into different regions such that within each region the bijection is linear. The algorithm starts at an arbitrary excess flow vector, say the all-zero vector, and determines the corresponding bijection. The bijection can be interpreted as a flow in a residual network of the graph.
The flow is then linearly increased by the flow in the residual network until the region corresponding to that bijection is left. In the non-degenerate case, the bijection for the next region can be obtained by a simplex-like pivot step. The degenerate case where multiple edge cost functions at the same time enter a point of non-differentiability is harder to handle. For the degenerate case, Katzenelsen~\cite{katzenelson1965} proposed to jump to a different starting excess vector and recurse. This, however, is not suitable for our purposes since this would not allow to capture the evolution of the Wardrop equilibria for increasing demand. To circumvent this issue, we develop a lexicographic rule similar to the lexicographic rule for the Simplex algorithm that allows to navigate to the next region, even in the degenerate case.

We then generalize the algorithm so that it can also handle discontinuities in the cost functions as long as the cost functions stay non-decreasing. Discontinuous cost functions are important for four different reasons: ({\it i}) while Katzenelson's algorithm inherently only works for undirected networks (as electrical flow can go in any direction), discontinuous cost functions allow to model directed edges by setting the cost of negative flow to minus infinity; ({\it ii}) by letting the cost function jump to infinity, hard edge capacities (see, e.g., \cite{correa2004}) can be modeled; ({\it iii}) discontinuities in cost functions may be used to model the transition between different regimes such as free flow and congested operation; ({\it iv}) for convex piece-wise edge cost functions $c_e(x)$, the marginal cost function $(xc_e(x))'$ is piece-wise linear but may be discontinuous. As system optimal flows are Wardop equilibria w.r.t.\ the marginal cost, being able to handle discontinuities allows to compute system optimal flows as well.
From a technical point of view, discontinuities are challenging as they render the Laplace matrix singular so that there is not a bijection between vertex potentials and vertex excess flows anymore. We show that this issue can only arise when source and sink are separated by a cut of edges at their discontinuity and that the singularity of the Laplace matrix can be repaired by increasing the vertex potentials on one side of the partition induced by this cut. This procedure can be naturally described as a jump procedure that avoids regions in potential space all-together and directly jumps to another potential vector with the same Wardrop flow. Using this method, we show that our algorithm still works for discontinuous cost functions with or without capacities and on both directed and undirected networks.

For the case that no degenerate point is hit, our algorithm is output polynomial, i.e., its complexity is polynomial in the output size. More precisely, the algorithm runs in $\mathcal{O}(n^{2.4} + n^2 K)$ time where $K$ is the number of breakpoints of the output flow functions. We show, however, that there are networks where the output is exponential in the input size so that an exponential output cannot be avoided. For our lexicographic method, a proof of a polynomial running time seems out of reach (similar to a proof of the (non-)existence of a polynomial pivot rule for the simplex algorithm). However, we can show that we can recompute the bijection after a degenerate point by a convex quadratic program which can be solved efficiently so that we obtain an output-polynomial algorithm. While the pivoting procedure is limited to the single-commodity case, we show that by computing the bijections with convex programs, also the multi-commodity case can be solved in output-polynomial time.

\subparagraph{Further related work}
Potential-based flows in networks have a long history in traffic, electric, water, and gas networks \cite{birkhoff1956}. In an electrical network, edges correspond to resistors connected at junctions. The network is undirected and there is a single commodity with possible different sources and sinks. The flow on an edge is uniquely determined as a function---the characteristic function of the resistor---of the difference of the potentials of its vertices. For monotonic characteristic functions, the potentials uniquely define the flow \cite{birkhoff1956,duffin1947,kirchhoff1847}. 
Katzenelson~\cite{katzenelson1965} proposed an algorithm for (approximately) computing the current for linear and non-linear resistors based on piece-wise linear approximation of non-linear characteristics. Minty~\cite{minty1961} gave a similar algorithm based on approximations by step functions.
The equilibrium principle for traffic flows was postulated by Wardrop~\cite{wardrop1952}.
Beckmann et al.~\cite{beckmann1956} characterized equilibria in terms of the maximum of a convex program. Wardrop equilibria can be computed by the Frank-Wolfe method~\cite{frank1956} which solves linearized versions of the convex program and eventually converges to a solution. For further developments, see, e.g.,~\cite{arezki1990,bar-gera2002,correa2010,leblanc1985}.

In mathematical terms, our algorithms are homotopy methods. The most prominent homotopy method in game theory is the Lemke-Howson algorithm for computing equilibria in 2-player games \cite{lemke1964}. For $n$-player games, the problem becomes non-linear. Herings and von den Elzen~\cite{herings2002} propose an algorithm that approximates an equilibrium by solving piece-wise linear approximations of the underlying complementarity problems. See \cite{herings2010} for a survey of further homotopy methods in game theory and \cite{goldberg2002} for a discussion of their complexity. A Wardrop equilibrium for a fixed flow demand in a networks with affine cost functions can also be computed with Lemke's algorithm by transforming the problem into a linear complementarity problem \cite{asmuth1979computing}, see also Kojima et al.~\cite{kojima1976extension} for a variant of Lemke's algorithm for piece-wise linear complementarity problems. This transformation adds an artificial variable that is reduced to zero during the course of the algorithm. This artificial variable turns any intermediate solution infeasible for the Wardrop equilibrium problem. On the other hand, our approach does not require an artificial variable and instead uses the flow demand for the path following which allows to obtain Wardrop equilibria with lower flow demand as intermediate solutions. 

Parametric optimization problems have been studied in related contexts, most prominently linear programs \cite{murty1980}, paths \cite{carstensen1983,karp1981,nikolova2006} and flows and cuts in capacitated networks~\cite{aissi2015,carstensen1983,gallo1989,granot2012}. 
Madry \cite{madry2016computing} presents an algorithm that solve the maximum flow problem by finding augmenting flows that are obtained by computing electrical flows in the residual network.
For flows in networks with integer capacities and linear costs, the sucessive shortest path algorithm computes a minimum cost flow for all (integer) flow values. Zadeh~\cite{zadeh1973bad} showed that this algorithm has exponential running time in the worst case. Disser and Skutella~\cite{disser2015simplex} proved that the algorithm is in fact NP-mighty, meaning that it can be used to solve NP-hard problems while being executed. Our results for discontinuous and constant cost functions allow to treat the minimum cost flow problem with linear costs and capacities as a special case of our the parametrized Wardrop equilibrium problem studied in this paper. Specifically, for this special case, our algorithm is similar to the sucessive shortest path algorithm with the difference that all successive shortest paths are augmented at the same time. Since the result of Disser and Skutella holds for arbitrary tie-breaking rules, their result implies that our algorithm is NP-mighty as well.

\subparagraph{Organization of the paper} In \S~\ref{sec:prelim}, we introduce the basic notions and useful properties of graph Laplacians, in \S~\ref{sec:basic}, we present the basic version of our algorithm for undirected graphs, and in \S~\ref{sec:discontinuous} we extend our analysis to discountinuous cost functions and directed graphs. 
\shortversion{Due to space constraints, all proofs, figures and examples are deferred to the appendix. The reader is invited to also consult the full version of this paper \cite{klimm2018computing}. Besides all results presented here, we there additionally cover the case of constant cost functions.}{\S~\ref{sec:constantCost} covers constant costs and the extensions to the basic algorithm that are necessary in this case.}
\section{Preliminaries}
\label{sec:prelim}

We consider both directed and undirected graphs $G = (V,E)$ with vertex set $V := \{1,\dots,n\}$ and edge set $E := \{e_1,\dots,e_m\} \subseteq V \times V$. This does not allow for parallel edges but such networks can easily be handled by introducing a dummy vertex in the middle of one of the edges. Each edge $e = (v,w) \in E$ has a unique orientation from $v$ to $w$. We assume that $G$ is (strongly) connected and encode $G$ by its incidence matrix $\vec \Gamma = (\gamma_{v,e}) \in \R^{n \times m}$ defined as $\gamma_{v,e} = 1$, if edge $e$ enters vertex $v$, $\gamma_{v,e} = -1$, if edge $e$ leaves vertex $v$, and $\gamma_{v,e} = 0$, otherwise. For an edge $e = e_i$ for some $i \in \{1,\dots,m\}$, we denote by $\vec \gamma_e$ the $i$-th column of $\vec \Gamma$. Throughout this work, all vectors are column vectors and will be denoted in bold face.

Let $K := \{1,\dots,k\}$ be a set of $k$ commodities. Each commodity~$j$ is specified by a start vertex $s_j$, a target vertex $t_j$ and a flow value $q_j$.
We conveniently express the commodity information in a matrix $\vec Y = (y_{v,j}) \in \R^{n \times k}$ where $y_{v,j} = -q_j$ if $v = s_j$, $y_{v,j} = q_j$ if $v= t_j$, and $y_{v,j} = 0$ otherwise.
A multi-commodity flow is a matrix $\vec X  = (x_{e,j}) \in \R^{m \times k}$ where $x_{e,j}$ denotes the flow of commodity $j$ on edge $e$.
For undirected networks, we allow both positive and negative flows $x_{e,j} \in \R$ where we use the convention that positive flow goes in the direction of the orientation of the edge while negative flow takes the opposite direction. For directed networks, we require $x_{e,j} \geq 0$, i.e., we only allow flow to go in the direction of the orientation of the edge.  A flow $\vec X$ is feasible if $\vec \Gamma \vec X = \vec Y$. 
When only a single commodity is considered $\vec X$ is a $m \times 1$ matrix and $\vec Y$ is a $n \times 1$ matrix and we write $\vec x$ and $\vec y$ instead. Moreover, we then write $s$ and $t$ instead of $s_1$ and $t_1$.

The cost of the flow on an edge is determined by an edge-specific cost function $l_e$ depending on the total flow $z_e$ on the edge defined as $z_e := \sum_{j = 1}^{k} x_{e, j}$.
For directed graphs, we have $l_e : \R_{\geq 0} \to \R_{\geq 0}$ and for undirected networks, we have $l_e : \R \to \R$. In the latter case, we assume that $l_e(z_e) \geq 0$ for $z_e > 0$ and $l_e(z_e) \leq 0$ for $z_e < 0$ so that $l_e(z_e)z_e$ is always non-negative. If the cost function $l_e$ of edge~$e$ of an undirected graph is anti-symmetric, i.e., $l_e(z_e) = -l_e(-z_e)$ for all $z_e$, then the cost of a flow does not depend on the direction. We allow both for anti-symmetric and general cost functions.
 We assume that the cost functions are piece-wise linear and continuous. Formally, for all edges~$e$ there are $\bar{t}_e \in \N$ breakpoints $\tau_{e,t} \in \R \cup \{-\infty, \infty\}$, $t \in T_e := \{1,\dots,\bar{t}_e\}$ and $\tau_{e,1} < \tau_{e,2} < \dots < \tau_{e,\bar{t}_e} = \infty$ where for undirected graphs we have $\tau_{e,1} = -\infty$ and for directed graphs we have $\tau_{e,1} = 0$. Further, there are constants $a_{e,t}, b_{e,t} \in \R$ such that $l_e(z_e) = a_{e,t} \,z_e + b_{e,t}$ for all $x \in (\tau_{e,t}, \tau_{e,t+1})$.

\subparagraph{Laplacians}
It will be convenient to consider the inverse cost functions $l_e^{-1}$. Since $l_e$ is piecewise linear, so is $l_e^{-1}$ and for $c_{e,t} := 1/a_{e,t}$ and $d_{e,t} := b_{e,t}\,/\, a_{e,t}$ we obtain the equality $l_e^{-1}(v_e) = c_{e,t}\, v_e - d_{e,t}$ for all $v_e \in (\sigma_{e,t}, \sigma_{e,t+1})$ 
where $\sigma_{e,t} = l_e(\tau_{e,t})$ for all $e \in E$ and $t \in T_e$. We call $c_{e,t}$ the \emph{conductivity} of edge~$e$ in segment $t$.
Let $\bar{T} = T_{e_1} \times \dots \times T_{e_m}$ and let $\vec t = (t_{e_1},\dots,t_{e_m})$. Intuitively, a vector $\vec t$ specifies a linear segment $(\tau_{e,t}, \tau_{e,t+1})$ for all edges~$e$ such that for $z_e \in (\tau_{e,t}, \tau_{e,t+1})$, we have $l_e(z_e) = a_{e,t}\, z_e + b_{e,t}$. We  set
$\smash{\vec C_{\vec t} = \diag(c_{e_1,t_{e_1}}, \dots, c_{e_m,t_{e_m}})}$.
%
For $\vec t \in \bar{T}$, the matrix $\vec L_{\vec t} := \vec \Gamma \vec C_{\vec t} \vec \Gamma^{\top}$ corresponds to the Laplace matrix of a weighted graph with edge weights $c_{e, t_e}$. \longversiononly{A short introduction to Laplacian matrices can be found in Appendix~\ref{app:laplacian}.}
The following lemma states well-known properties of Laplacians of weighted graphs, see \cite{grone1991geometry, merris1994laplacian, mohar1991laplacian} for a reference.

\begin{lemma} \label{lem:LaplacianMatrix}
For all $\vec t \in \bar{T}$, the Laplacian $\vec L_{\vec t} = \vec \Gamma  \vec C_{\vec t} \vec \Gamma^{\top}$ has the following properties:
\begin{compactenum}[(i)]
\item \label{lem:LaplacianMatrix:rank}
	The rank of $\vec L_{\vec t}$ is $n- n_C$ where $n_C$ is the number of connected components and two vertices are considered connected if they are connected by an edge with $c_{e,t_e} \neq 0$.
\item \label{lem:LaplacianMatrix:semiDefinite}
	The matrix $\vec L_{\vec t}$ is positive semi-definite.
\item \label{lem:LaplacianMatrix:rowSum}
	The row sum and column sum of $\vec L_{\vec t}$ is zero for every row or column.
\item \label{lem:LaplacianMatrix:subMatrix}
	If $G$ is connected then the reduced Laplacian matrix $\hat{\vec L}_{\vec t}$ obtained from the matrix $\vec L_{\vec t}$ by deleting the first row and column is positive definite and inverse-positive, i.e. $\vec{L}_{\vec t}^{-1} \geq 0$.
\end{compactenum}
\end{lemma}
\longversiononly{
\begin{proof}
See the proof of Lemma~\ref{lem:appendixLaplacianMatrix} in Appendix~\ref{app:laplacian}.
\end{proof}
}

\subparagraph{Wardrop equilibria}

Feasible edge flows $\mcFlowVec$ have a natural representation in terms of simple (i.e., cycle-free) path flows and cycle flows, see, e.g. \cite[Thm.~3.4]{ahuja1993} for a reference.
For a commodity $j \in \setCommodities$ let $\setPaths_j$ denote all simple $(\source_j$,$\sink_j)$-paths in $G$ (undirected paths if we consider the undirected case and directed paths if we consider the directed case) and let $\setPaths = \bigcup_{j = 1}^{k} \mathcal{P}_{j}$. For every $P \in \mathcal{P}$ we define the indicator vector of this path $\vec{\chi}_P$ by setting $\chi_{P,e} = 1$ if $e \in P$ in forward direction, $\chi_{P,e} = -1$ if $e \in P$ in backward direction, and $\chi_{P,e} = 1$ otherwise.
The length of a path $\pathVar \in \setPaths$ with respect to the flow $\mcFlowVec$ is defined as 
$l_\pathVar (\mcFlowVec) := \sum_{e \in \pathVar} l_e( z_e ) \,\chi_{P,e}$.

We say a path $\pathVar \in \setPaths_j$ is a flow-carrying for commodity $j \in \setCommodities$ if $\mcFlow[e, j] > 0$ for all forward-edges $e$ and $\mcFlow[e, j] < 0$ for all backward edges in the path, i.e. $x_{e,j} \, \chi_{P,e} > 0$ for all edges. The following notion of an equilibrium flow is due to Wardrop~\cite{wardrop1952}. Intuitively, it requires that no (infinitesimal small) flow particle of any commodity traveling along a flow-carrying path can decrease the length of its path by deviating to another path.
\longversiononly{
\begin{definition}[Wardrop equilibrium]
A feasible flow $\mcFlowVec$ is a \emph{Wardrop equilibrium} if we have for all commodities $j \in \setCommodities$ that no cycle is flow-carrying and
$l_{\pathVar} (\mcFlowVec) \leq l_{\altPath} (\mcFlowVec)$ for all $\pathVar, \altPath \in \setPaths_j$ where $\pathVar$ is flow-carrying.
\end{definition}
}
\shortversiononly{%
Formally, a feasible flow $\mcFlowVec$ is a \emph{Wardrop equilibrium} if we have for all commodities $j \in \setCommodities$ that no cycle is flow-carrying and
$l_{\pathVar} (\mcFlowVec) \leq l_{\altPath} (\mcFlowVec)$ for all $\pathVar, \altPath \in \setPaths_j$ where $\pathVar$ is flow-carrying.}%
To ensure the existence of Wardrop equilibria, we only admit continuous cost functions for the first part of the paper. Note, that in the undirected case the general assumption $l_e(z_e) \geq 0$ for $z_e > 0$ and $l_e (z_e) \leq 0$ for $z_e < 0$ implies $l_e(0) = 0$ in the continuous case. In section \ref{sec:discontinuous} we will discuss how we can extend our results to discontinuous cost functions.
Since we are interested in edge flows rather than in paths flows, we give the following characterization of Wardrop equilibrium flows via edge flows and potential functions. A similar characterization has been given by Roughgarden~\cite[Proposition 2.7.7]{roughgarden2005book} for directed single-commodity networks. 
\longversion{
\begin{lemma} \label{lem:characterisationWEPotentials}
Let $\mcFlowVec$ be a feasible flow in a directed graph $G$ with continuous cost functions. Then the following statements are equivalent.
\begin{compactenum}[(i)]
\item \label{it:lem:charactierizationEWPotentials1}
	The flow $\mcFlowVec$ is a Wardrop equilibrium flow.
\item \label{it:lem:charactierizationEWPotentials2}
	There exist potential vectors $\vec \pi_j \in \R^{\numVertices}$ for every commodity $j \in \setCommodities$ such that
	\begin{align} 
	\vec \gamma_e^{\top} \vec \pi_j = \pi_{j,w} - \pi_{j,v} &\leq l_e(z_e)  
	\qquad \text{for all } e = (v,w) \in E , \label{eq:lem:characterisationWEPotentials1}\\
	\vec \gamma_e^{\top} \vec \pi_j = \pi_{j,w} - \pi_{j,v} &= l_e(z_e) 
	\qquad \text{for all } e = (v,w) \in E \text{ with } \mcFlow > 0 .\label{eq:lem:characterisationWEPotentials2}
	\end{align}
\end{compactenum}
	For undirected graphs $G$ the equivalence holds when \eqref{eq:lem:characterisationWEPotentials2} is satisfied for all $e \in E$.
\end{lemma}
}{
\begin{lemma} \label{lem:characterisationWEPotentials}
A feasible flow $\mcFlowVec$ in a directed graph $G$ with continuous cost functions is a Wardrop equilibrium if and only if there exist potential vectors $\vec \pi_j \in \R^{\numVertices}$, $j \in \setCommodities$ such that
	\begin{align} 
	\vec \gamma_e^{\top} \vec \pi_j = \pi_{j,w} - \pi_{j,v} &\leq l_e(z_e)  
	\qquad \text{for all } e = (v,w) \in E , \label{eq:lem:characterisationWEPotentials1}\\
	\vec \gamma_e^{\top} \vec \pi_j = \pi_{j,w} - \pi_{j,v} &= l_e(z_e) 
	\qquad \text{for all } e = (v,w) \in E \text{ with } \mcFlow > 0 .\label{eq:lem:characterisationWEPotentials2}
	\end{align}
	For undirected graphs $G$ the equivalence holds when \eqref{eq:lem:characterisationWEPotentials2} is satisfied for all $e \in E$.
\end{lemma}
}
\begin{appendixproof}{Lemma~\ref{lem:characterisationWEPotentials}}
We first show the result for directed networks. Note that for potential vectors $\vec \pi \in \R^n$ satisfying \eqref{eq:lem:characterisationWEPotentials1} we have
\begin{align} \label{eq:lem:characterisationWEPotentials:2}
l_\pathVar(\mcFlowVec) &= \sum_{e =(v,w) \in \pathVar} l_e( \totalFlow_e ) \geq \sum_{e=(v,w) \in \pathVar} \pi_{\commodity,w} - \pi_{\commodity,v} 
	= \pi_{\commodity, \sink_\commodity} - \pi_{\commodity, \source_\commodity} 
\end{align}
for every path $\pathVar \in \setPaths_\commodity$ and every commodity $\commodity
 \in \setCommodities$.

Let $\vec \pi \in \R^n$ be a potential satisfying \eqref{eq:lem:characterisationWEPotentials1} and \eqref{eq:lem:characterisationWEPotentials2}. Then, $\vec \pi$ satisfies \eqref{eq:lem:characterisationWEPotentials:2} and the inequality in \eqref{eq:lem:characterisationWEPotentials:2} is satisfied with equality if $P$ is flow-carrying for commodity~$\commodity$. Thus,
 \begin{align*}
 l_P(\vec X) = \pi_{\commodity,\sink_\commodity} - \pi_{\commodity,\source_\commodity} \leq l_Q(\vec X)	
 \end{align*}
for every path $Q \in \setPaths_\commodity$ and every flow-carrying path $P \in \setPaths_\commodity$ and $\vec{X}$ is a Wardrop equilibrium.
 
Conversely, let $\vec X$ be a Wardrop equilibrium and consider the vector $\vec \pi = (\pi_v)_{v \in V}$ defined as
\begin{align*}
\pi_{\commodity,\vertex} = \text{length of a shortest $(\source_\commodity,\vertex)$-path w.r.t.\ $l_e(z_e)$}.
\end{align*}
The potentials $\vec \pi_j$ satisfy \eqref{eq:lem:characterisationWEPotentials1} by the Bellmann equations for shortest paths and, hence, \eqref{eq:lem:characterisationWEPotentials:2} is satisfied, i.e., $l_P(\vec X) \geq \pi_{\commodity,\sink_\commodity} - \pi_{\commodity,\source_\commodity}$ for all commodities~$\commodity \in \setCommodities$ and $\pathVar \in \setPaths_\commodity$. As $\vec X$ is a Wardrop equilibrium, we further have $l_P(\vec X) = \pi_{\commodity,\source_\commodity} - \pi_{\commodity,\sink_\commodity}$ if $P$ is flow-carrying.  This implies $\vec \gamma_e^{\top} \vec \pi_\commodity = l_e(\totalFlow_e) $ for all edges contained in a flow-carrying path. Since $\mcFlow[e, \commodity] > 0$ implies that $e$ is contained in a flow-carrying path the potentials $\vec{\pi}_j$ satisfy \eqref{eq:lem:characterisationWEPotentials1} and \eqref{eq:lem:characterisationWEPotentials2}.

For undirected networks, the same arguments as before apply and the only thing left to show is that also $\vec\gamma^\top_e \vec \pi_\commodity = l_e(z_e)$ for edges with $x_{e,j} = 0$. Since we obtain the inequality \eqref{eq:lem:characterisationWEPotentials1} in both directions of the edge flow, the claimed result follows.
\end{appendixproof}

\subparagraph{Problem Statement and Variants}
We consider the following problem.\\[-1.5\baselineskip]
\begin{framed}
\noindent\textsc{Parametrized Wardrop Equilibria}\\
\begin{tabular}{ll}
\textsc{Given: } & graph $G = (V,E)$, source $\source_j$, sink $\sink_j \in V$, $j = 1,\dots,k$;\\
\textsc{Find:} & map $\vec f : \R \to \vec \R^{m \times k}$\\
&s.t.\ $\vec f(\lambda) = (f_{e,k})_{e \in E, j=1,\dots,k}$ is a WE for demands $\vec q = \lambda\vec 1$ for all $\lambda\geq 0$.
\end{tabular}
\end{framed}
~\\[-1.5\baselineskip]
Some comments on the problem statement are in order. First, one may wonder about the representation of the function $f$. As we will see, under the given assumptions (i.e., piece-wise linear cost functions), the function $f$ is piece-wise linear in $\lambda$, i.e., there are breakpoints $0 = \lambda^1 < \lambda^2 < \dots$ and $\vec M^1, \vec M^2,\dots, \vec N^1, \vec N^2,\dots \in \R^{m \times k}$ such that $f(\lambda) = \lambda \vec M^i + \vec N^i$ for all $\lambda \in [\lambda^i, \lambda^{i+1}]$, $i = 1,2,\dots$. Second, one may wonder about natural generalizations of the problem. For instance, each commodity $j$ may have weight $w_j$ and the task is to specify $f$ such that $f(\lambda)$ is a Wardrop equilibrium for the demands $q_j = \lambda w_j$. More generally, the target demand may be given by an arbitrary piece-wise linear function $h_j(\lambda)$. Our algorithms can be easily adapted to solve both generalizations. For ease of exposition, we specify our algorithms only for the parametrized Wardrop equilibrium above. Since socially optimal flows are Wardrop equilibria with respect to the marginal edge costs, our algorithms may be applied to compute socially optimal flows for all demands as well.

\section{The Basic Algorithm}
\label{sec:basic}

In this section, we give an algorithm that computes all Wardrop equilibria parametrized by the flow demand in an undirected single-commodity network with continuous and strictly increasing cost functions $l_e$ with $l_e(0) = 0$.  Let $s = 1$ and $t = n$. 
The main algorithmic idea is due to Katzenelson~\cite{katzenelson1965} who devised such algorithm for electrical networks. However, in Katzenelson's work the focus lies on finding the flow vector for a given potential vector. Katzenelson's algorithm starts in an arbitrary given solution for a different vector of potentials and traces the evolution of the flow when changing the potential vector towards the goal vector. Along this path, the trace may hit a degenerate point which are subtle to handle. When such point is reached, Katzenelson suggests to jump to another non-degenerated point and proceed with the algorithm. This approach is incompatible with our objective of finding all Wardrop equilibria and we will significantly change the algorithm at these points. Specifically, Katzenelson's algorithm is unable to start with the zero flow as this point is usually highly degenerated. To overcome the issue of degenerated points, we devise a new pivoting rule for the degenerated solutions similar to existing rules for the Lemke-Howson-Algorithm and the simplex algorithm.

Since we only deal with the single-commodity case in this section, we will use the commodity flow $\vec x$ instead of the total flow $\vec z$ whenever the total flow is needed. Lemma~\ref{lem:characterisationWEPotentials} states that in undirected networks a flow $\vec{x}$ is a Wardrop equilibrium if and only if $l_e(x_e) = \vec \gamma_e^{\top} \vec \pi = \pi_w - \pi_v$ for all edges $e = (v,w) \in E$ and some potential vector $\vec \pi$.
A given vector $\vec \pi = (\pi_v)_{v \in V} \in \R^n$ of vertex potentials induces a flow vector $\vec x = (x_e)_{e \in E} \in \R^m$ that satisfies the conditions of Lemma~\ref{lem:characterisationWEPotentials} according to the equality 
\begin{flexequation}
\label{eq:x_e}
x_e = l_e^{-1}(\pi_w - \pi_v) \longversiononly{\quad} \text{ for edge $e = (v,w)$ }.
\end{flexequation}
Since $l_e$ is piecewise linear so is $l_e^{-1}$. To obtain a closed form expression for $l_e^{-1}$, we proceed to subdivide the space of potential vectors $\vec \pi \in \R^n$ into finitely many convex regions where the cost functions are affine. Let $\vec t \in \bar T := T_1 \times T_2 \times \dots \times T_m$. We have that 
\begin{equation} \label{eq:inverseCostFunction}
l_e^{-1}(\pi_w - \pi_v) = c_{e,t_e} ( \pi_w - \pi_v ) - d_{e,t_e}
\end{equation}
for all $\pi_w - \pi_v \in [\sigma_{e,t_e},\sigma_{e,t_e+1}]$ and $e = (v,w) \in E$. The  region $R_{\vec t}$ is described  by the inequalities
\begin{align} \label{eq:regionInequalities}
\pi_w - \pi_v &\leq \sigma_{e,t_e+1} & &\text{ and } & \pi_w - \pi_v &\geq \sigma_{e,t_e}
\end{align}
for all $e = (v,w) \in E$. 
So all regions are bounded by hyperplanes that are described by inequalities as in \eqref{eq:regionInequalities}. We call these hyperplanes boundaries and say a boundary is induced by some edge $e$ if the hyperplane is described by a inequality depending on $e$. Note that some regions $R_{\vec t}$, $\vec t \in \bar K$ may be empty.
\shortversiononly{For an example consider Example~\ref{ex:regionsExample} in Appendix~\ref{appendixexample\thesection subsection} and Figure~\ref{fig:regions} in Appendix~\ref{appendixfigure\thesection subsection}.}
\begin{appendixfigure}{t}
\centering
\begin{subfigure}[b]{.45\linewidth}
\begin{center}
\begin{tikzpicture}
\draw (0,0) node[solid] (s) {} node [below left] {$s$};
\draw (2,1) node[solid] (v) {} node [above] {$v$};
\draw (4,0) node[solid] (t) {} node [below right] {$t$};

\draw[->]	(s) edge node[midway, above left] {$e_1$} (v) 
			(v) edge node[midway, above right] {$e_3$} (t) 
			(s) edge node[midway, below] {$e_2$} (t);
\end{tikzpicture}
\end{center}
\caption{Graph.}
\end{subfigure}
\hspace{0.02\linewidth}
\begin{subfigure}[b]{.45\linewidth}
\begin{center}
\begin{align*}
l_{e_1} (x) &= 
\begin{cases}
x	& \text{if } x < 2, \\
4x - 6 & \text{if } x \geq 2,
\end{cases} \\
l_{e_2} (x) &= 
\begin{cases}
x	& \text{if } x < 1, \\
\nicefrac{1}{4} x + \nicefrac{3}{4} & \text{if } x \geq 1,
\end{cases} \\
l_{e_3} (x) &=
\begin{cases}
x & \text{if } x < 1, \\
\frac{1}{2} \, x + \frac{1}{2} & \text{if } x \geq 1. 
\end{cases} 
\end{align*}
\end{center}
\caption{Cost functions.}
\end{subfigure}
\\
\begin{subfigure}[t]{.45\linewidth}
\begin{center}
\begin{tikzpicture}[
declare function={
    regionId(\x, \y)=0 + 
    	and(and(\x < 2, \y<1), \y<\x + 1) * 1   +
        and(and(\x > 2, \y<1), \y<\x + 1) * 2   +
        and(and(\x < 2, \y<1), \y>\x + 1) * 3   +
        and(and(\x < 2, \y>1), \y>\x + 1) * 4   +
        and(and(\x < 2, \y>1), \y<\x + 1) * 5   +
        and(and(\x > 2, \y>1), \y<\x + 1) * 6   +
        and(and(\x > 2, \y>1), \y>\x + 1) * 7
   ;
   vfx(\r) = 0 + (\r==1)*1/3/2  + (\r==2)*2/3/2 + (\r==3)*2/5 + (\r==4)*1/7 + (\r==5)*1/9 + (\r==6)*4/21 + (\r==7)*4/19;
   vfy(\r) = 0 + (\r==1)*2/3/2  + (\r==2)*5/6/2 + (\r==3)*3/5 + (\r==4)*3/14 + (\r==5)*2/9 + (\r==6)*5/21 + (\r==7)*9/38;
   vfcolor(\r) =  + (\r==1)*1  + (\r==2)*2 + (\r==3)*3 + (\r==4)*4 + (\r==5)*6 + (\r==6)*7 + (\r==7)*8;
}
]
\newcommand{\xMin}{-0.5} \newcommand{\xMax}{4.5}
\newcommand{\yMin}{-0.5} \newcommand{\yMax}{4.5}
\path[fill=color1,fill opacity=0.25] ({max(\xMin,\yMin-1)},{max(\xMin+1,\yMin)}) -- ({max(\xMin,\yMin-1)},\yMin)--(2,\yMin)--(2,1)--(0,1);
\path[fill=color2,fill opacity=0.25] (2,\yMin)--(2,1)--(\xMax,1) -- (\xMax, \yMin);
\path[fill=color3,fill opacity=0.25] ({max(\xMin,\yMin-1)},{max(\xMin+1,\yMin)}) -- (\xMin,{max(\xMin+1,\yMin)}) -- ({max(\xMin,\yMin-1)},1) -- (0,1);
\path[fill=color4,fill opacity=0.25] (\xMin, \yMax) -- (\xMin, 1) -- (0,1) -- (2,3) -- (2, \yMax);
\path[fill=color6,fill opacity=0.25] (0,1) -- (2,3) -- (2,1);
\path[fill=color7,fill opacity=0.25] (2,1) -- (2,3) -- ( { min(\xMax , \yMax - 1 ) } , { min(\xMax + 1, \yMax) } ) -- (\xMax, \yMax) -- (\xMax, 1);
\path[fill=color8,fill opacity=0.25] (2,3) -- ( { min(\xMax , \yMax - 1 ) } , { min(\xMax + 1, \yMax) } ) -- ( { min(\xMax , \yMax - 1 ) }, \yMax) -- (2, \yMax);

\pgfmathsetmacro{\vectorscale}{1.5}
\foreach \x in {-0.25,0.25,...,4.25} {
	\foreach \y in {-0.25,0.25,...,4.25} {
		\pgfmathparse{int(regionId(\x,\y))}
    	\ifnum\pgfmathresult>0 
    	\pgfmathparse{int(vfcolor(regionId(\x,\y)))}
		\draw[->, opacity=0.5, color\pgfmathresult] 
			({\x - \vectorscale * vfx(regionId(\x,\y))/2}, {\y - \vectorscale * vfy(regionId(\x,\y))/2}) -- 
			({\x + \vectorscale * vfx(regionId(\x,\y))/2}, {\y + \vectorscale * vfy(regionId(\x,\y))/2});
		\fi
	}
}

\draw[thick, ->] (\xMin,0) -- (\xMax,0) node[anchor=west] {$\pi_v$};
\foreach \i in {1,...,4}
{
	\draw[thick] (\i, 0.1) -- (\i,-0.1) node[anchor=north] {\i};
}
\foreach \i in {1,...,4}
{
	\draw[thick] (0.1, \i) -- (-0.1, \i) node[anchor=east] {\i};
}
\draw[thick, ->] (0,\yMin) -- (0,\yMax) node[anchor=south] {$\pi_t$};

\draw[dashed] (\xMin, 1) -- (\xMax, 1);
\draw[dashed] (2, \yMin) -- (2, \yMax);
\draw[dashed] ({max(\xMin,\yMin)},{max(\xMin+1,\yMin)}) -- ({min(\xMax,\yMax-1)},{min(\xMax+1,\yMax)});

\draw[thin, ->] (2,2.5) -- (3, 2.5) node[above] {$\vec{\gamma}_{e_1}$};
\draw[thin, ->] (3.25,1) -- (3.25, 2) node[right] {$\vec{\gamma}_{e_2}$};
\draw[thin, ->] (1,2) -- (0.3, 2.7) node[right] {$\vec{\gamma}_{e_3}$};

\draw[thick, color4!70!black] (0,0) -- (1/2,1) -- (1, 2) -- (2, 7/2) -- ( {min( \xMax, \yMax*8/9 - 5/4 )} , { min(9/8*\xMax + 5/4 , \yMax } );
\draw[thick, color5!70!black] (0,0) -- (2,1) -- ( {min( \xMax,  5*\yMax - 3)} , { min(1/5*\xMax + 3/5 , \yMax } );

\fill[color4!70!black] (1/2,1) circle(2pt);
\fill[color4!70!black] (1,2) circle(2pt);
\fill[color4!70!black] (2,7/2) circle(2pt);
\fill[color5!70!black] (2,1) circle(2pt);

\fill[color7!70!black] (4,3) circle(2pt);
\fill[color6!70!black] (4,1/2) circle(2pt);
\end{tikzpicture}
\end{center}
\caption{The potential space with $\pi_s := 0$. }
\end{subfigure}
\hspace{0.02\linewidth}
\begin{subfigure}[t]{.45\linewidth}
\begin{center}
\begin{tikzpicture}
\newcommand{\xMin}{-2.5} \newcommand{\xMax}{2.5}
\newcommand{\yMin}{-1.5} \newcommand{\yMax}{3.5}

\newcommand{\xScale}{1/4}
\newcommand{\yScale}{1/4}

\draw[thick, ->] (\xMin,0) -- (\xMax,0) node[anchor=west] {$y_v$};
\foreach \i in {-9, -8, -7, -6, -4, -3, -2,-1,1,2,3,4,6,7,8,9}
{
	\draw (\i *\xScale, 0.05) -- (\i*\xScale,-0.05);
}
\foreach \i in {-5,5}
{
	\draw[thick] (\i*\xScale, 0.1) -- (\i*\xScale,-0.1) node[anchor=north] {\i};
}
\foreach \i in {-4,-3,-2,-1,1,2,3,4,6,7,8,9,11,12,13}
{
	\draw (0.05, \i * \yScale) -- (-0.05, \i * \yScale);
}
\foreach \i in {-5,5,10}
{
	\draw[thick] (0.1, \i * \yScale) -- (-0.1, \i * \yScale) node[anchor=east] {\i};
}
\draw[thick, ->] (0,\yMin) -- (0,\yMax) node[anchor=south] {$y_t$};


\draw[very thick, color4!70!black] (0,0) -- (0, \yMax);
\draw[very thick, color5!70!black] (0,0) -- ( \xMax, 0 );

\fill[color4!70!black] (0,3/2*\yScale) circle(2pt);
\fill[color4!70!black] (0,6*\yScale) circle(2pt);
\fill[color4!70!black] (0,13*\yScale) circle(2pt);
\fill[color5!70!black] (3*\xScale, 0) circle(2pt);

\fill[color7!70!black] (3.5*\xScale,8*\yScale) circle(2pt);
\fill[color6!70!black] (6 * \xScale, -3 * \yScale) circle(2pt);
\end{tikzpicture}
\end{center}
\caption{The excess space with $y_s := -y_v - y_t$}
\end{subfigure}
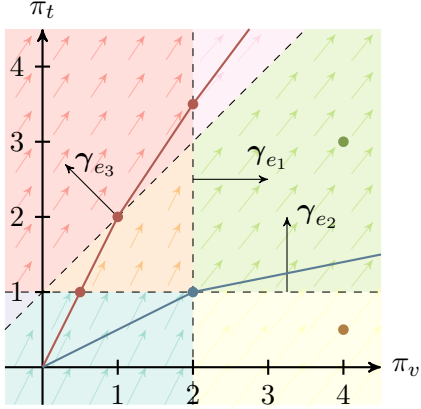
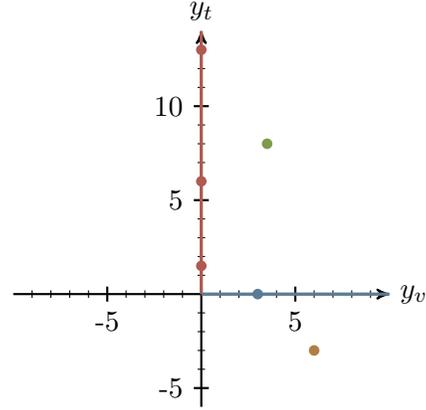
\caption{An example graph (a) with piecewise linear cost functions (b). The breakpoints of the cost functions induce boundaries (dashed lines with normal vectors $\vec{\gamma}_{e_i}$) in the potential space (c). The colored points in the potential space induce the corresponding points in the excess space (d) and vice versa. The red (blue) line mark the potential and excess vectors that correspond to $s$-$t$-flows ($s$-$v$-flows) for demands $\lambda \geq 0$. \\
For every point $\vec{\pi}$ in the potential space there is exactly one point $\vec{y}$ in the excess space, such that $\vec{y} = \laplaceM_{\vec{t}} - \tilde{\vec{d}}_{\vec{t}}$ where $R_{\vec{t}}$ is the region such that $\pi \in R_{\vec{t}}$. 
The red/blue lines correspond to potential and excess vectors that belong to $s$-$t$-/$s$-$v$-flows. The arrows in (c) indicate the potential directions $\Delta \vec{\pi}$ in the regions corresponding to the $s$-$t$-flow excess direction $\Delta \vec{y} = (-1,0,1)$ in the excess space.}
\label{fig:regions}
\end{appendixfigure}

\begin{appendixexample} \label{ex:regionsExample}
See Figure~\ref{fig:regions} for an example graph (a). The three breakpoints of the cost functions induce three boundaries (dashed lines) in the potential space (c) which subdivide the space into seven non-empty regions. The normal vectors of these boundaries are the corresponding columns $\vec{\gamma}_e$ of the incidence matrix $\vec{\Gamma}$. \\
Any arbitrary potential vector $\vec{\pi} \in \R^m$ lies in some region $R_{\vec{t}}$. Given the potential vector, we can compute a corresponding Wardrop flow $\vec{x}$\longversiononly{ by equation~\eqref{eq:x_e}}. Note, that we are considering the undirected case -- this means while the edges have some direction (as in the example graph), we admit negative flows which we interpret as flows in the opposite direction. The flow $\vec{x}$ induced by the potential $\vec{\pi}$ does not satisfy any flow conservation in general but induces some excess or demand (negative excess) $y_v = \sum_{e \in \delta^-(v)} x_e - \sum_{e \in \delta^+(v)} x_e$ at every vertex $v$. So every potential vector $\vec{\pi}$ induces a unique excess vector $\vec{y} \in \R^n$. In Figure~\ref{fig:regions}, every colored point in the potential space (c) corresponds to the points in the excess space (d). In particular, the red line segments in the potential space correspond to the red line in the excess space. These potentials are exactly the potentials that induce excess vectors with $y_s = - \lambda, y_t = \lambda$ and $y_v = 0$ for all vertices $v \neq s,t$, i.e. these potentials induce $s$-$t$-flows for the demand $\lambda$. In the same way, the blue lines correspond to potential/excess vectors that belong to $s$-$v$-flows. 
\end{appendixexample}

We may conveniently express \eqref{eq:inverseCostFunction} as
$ \vec x = \vec C_{\vec t} \, \vec{\Gamma}^{\top} \vec{\pi} - \vec d_{\vec t}, $
with the diagonal matrix $\vec C_{\vec t} = \diag(c_{1,t_1},c_{2,t_2},\dots,c_{m,t_m}) = \diag(1/a_{1,t_1},\dots,1/a_{m,t_m})$ and the vector $\vec d_{\vec t} = (d_{1,t_1},\dots,d_{m,t_m}) = (b_{1,t_1}/a_{1,t_1},\dots,b_{m,t_m}/a_{m,t_m})$. 
Finally, we express the excess flow $y_v = \sum_{e \in \delta^-(v)} x_e - \sum_{e \in \delta^+(v)} x_e$ at vertex $v$ as
\begin{equation} \label{eq:potentialLaplacianExcess}
\vec y
= \vec \Gamma \vec x 
= \vec \Gamma \bigl( \vec C_{\vec t} \vec \Gamma^{\top} \vec \pi - \vec d_{\vec t} \bigr) 
= \vec L_{\vec t} \vec \pi - \tilde{\vec d}_{\vec t}
,
\end{equation}
where $\vec L_{\vec t} = \vec \Gamma  \vec C_{\vec t} \vec \Gamma^{\top}$ is the Laplacian matrix and $\tilde{\vec d}_{\vec t} = \vec \Gamma \vec d_{\vec t}$. 
With equation~\eqref{eq:inverseCostFunction} we can compute a Wardrop equilibrium flow (since it satisfies the conditions of Lemma~\ref{lem:characterisationWEPotentials} by definition). However, this flow does not need to fulfill flow-conservation and is not a $s$-$t$-flow in general but induces the excess vector $\vec{y}$. Our aim is to find potentials that induce excess vectors of the form $\vec{y} = \lambda \, \Delta \vec{y}$ where $\Delta \vec{y}$ is defined by $\Delta y_s = -1, \Delta y_t = 1$, and $\Delta y_v = 0$ for all $v \neq s,t$, i.e. excess vectors that correspond to $s$-$t$-flows with demand $\lambda$. Flows induced by these potentials then are $s$-$t$-flows that satisfy the Wardrop equilibrium conditions from Lemma~\ref{lem:characterisationWEPotentials}.

As can be seen in equations~\eqref{eq:inverseCostFunction} and~\eqref{eq:potentialLaplacianExcess},  the flow $\vec{x}$ (and thus the excess $\vec{y}$) depends solely on the potential differences $\pi_w - \pi_v$ on the edges. It is, thus, no restriction to assume that $\pi_s = 0$. Further, note that every vector $\vec{y}$ of induced excess has zero sum.
We introduce a new notation. For any vector $\vec{v} \in \R^n$, we define $\hat{\vec{v}} := (v_2, \dotsc, v_n)^{\top}$ to be the vector $\vec{v}$ without the first component and, for any matrix $\vec{A} \in \R^{n \times n}$, we define $\hat{\vec{A}}$ to be the submatrix obtained by deleting the first row and column.
Hence, we can write $\vec{\pi} = ( 0 , \hat{\vec{\pi}} )^{\top}$, and 
$\vec{y} = ( - \sum_{i=2}^n y_i, \hat{\vec{y}} )^{\top}$. This means that it is sufficient to deal with the vectors $\hat{\vec{\pi}}$ and $\hat{\vec{y}}$ and we obtain from \eqref{eq:potentialLaplacianExcess} that
\begin{flexequation*}
\hat{\vec{y}} = \hat{\vec{L}}_{\vec{t}} \hat{\vec{\pi}} - \hat{\vec{d}}_{\vec{t}}
\end{flexequation*}
By Lemma~\ref{lem:LaplacianMatrix}(\ref{lem:LaplacianMatrix:subMatrix}), the matrix $\hat{\vec{L}}_{\vec{t}}$ is non-singular and thus there is a linear, one-to-one correspondence between potential vectors (with $\pi_s = 0$) and excess vectors (and also between the potential vectors and the demand) in every region.%
\longversiononly{The matrix $\smash{\hat{\vec{L}}_{\vec{t}}^{-1}}$ acts as a pseudo inverse of $\vec{L}$ on the subspace $\{\vec{\pi} \in \R^n : \pi_s = 0\}$, see Appendix~\ref{app:laplacian:inverse}. \par}
The set of all excess vectors of $s$-$t$-flows with demand $\lambda \geq 0$ is a line in the excess space (see Figure~\ref{fig:regions}(d)\shortversiononly{ in Appendix~\ref{appendixfigure\thesection subsection}}). Since there is a linear dependence between $\vec{y}$ and $\vec{\pi}$ in every region, the potentials inducing $s$-$t$-flows are line segments in the potential space (see Figure~\ref{fig:regions}(c)\shortversiononly{ in Appendix~\ref{appendixfigure\thesection subsection}}). We denote by $\vec{\pi} (\lambda )$ the unique potential that induces the excess $\lambda \, \Delta \vec{y}$ and call the function $\vec{\pi} ( \cdot )$ the \textit{solution curve} in the potential space.

Our algorithm computes the function $\vec \pi(\sCurveP)$ and the function $f(\sCurveP) = \vec x (\sCurveP) = \vec C_{\vec t} \, \vec{\Gamma}^{\top} \vec{\pi} (\sCurveP) - \vec d_{\vec t}$ of Wardrop equilibrium edge flows by starting with a solution $\vec{\pi} ( \lambda_0)$ and then following the line segments in the potential space by computing the direction $\smash{\Delta \hat{\vec{\pi}}_{\vec{t}} = \hat{\vec{L}}_{\vec{t}}^{-1} \, \Delta \hat{\vec{y}}}$ of these line segments in every region. Schematically, the algorithm proceeds as follows:\shortversiononly{\\}
\begin{compactenum}
\item Start with some potential $\vec \pi$ in some region $R_{\vec t}$;
\item \label{it:compute} Compute the direction vector $\Delta \vec{\pi}$ of the line segment in this region;
\item Find the first boundary that the solution curve hits;
\item Find the adjacent region $R_{\tilde {\vec t}}$ and repeat with \ref{it:compute}.  
\end{compactenum}
We will describe these steps in more detail below. 
\subparagraph{Initialization}
We start with the initial potential $\vec \pi^1$ such that
$\vec{y}^1 = \vec{L}_{\vec t^1} \vec \pi^1 - \tilde{\vec d}_{\vec t^1} = \vec 0$
where $\vec t^1$ is chosen such that $\vec \pi^1 \in R_{\vec t^1}$. The initial potential is chosen such that there is no excess at any vertex. If all cost functions satisfy $l_e(0) = 0$ this means in particular that the initial flow
$f(\sCurveB^1) := \vec N^1 := \vec C_{\vec t^1} \vec \Gamma^{\top} \vec \pi^1 - \vec d_{\vec t^1}$
is the zero flow. The breakpoint $\sCurveB^1 = 0$ is the first breakpoint of the solution function $f$ and $\vec N^1$ the flow at this breakpoint.


\subparagraph{Main Loop}
At the start of every iteration $\sCurveI$, the solution curve is in some point $\vec \pi^\sCurveI$ in some region $R_{\vec t^\sCurveI}$.
To find a direction vector $\Delta \vec \pi^\sCurveI$ such that 
$\Delta \vec{y} = \vec L_{\vec t^{\sCurveI}} \Delta \vec{\pi}^i$
we compute $\smash{\Delta \hat{\vec \pi}^i = \hat{\vec L}_{\vec t^{\sCurveI}}^{-1} \, \Delta \hat{\vec{y}}}$ which is well-defined since $\hat{\vec L}_{\vec t^{\sCurveI}}$ is invertible by Lemma~\ref{lem:LaplacianMatrix}(\ref{lem:LaplacianMatrix:subMatrix}). We obtain $\Delta \vec \pi^i = (0, \Delta \pi_2^i, \dotsc, \Delta \pi_n^i)^{\top}$. The slope of the solution flow function $f(\lambda)$ is then computed as
$\vec M^i := \vec C_{\vec t^i} \incidenceMatrix^{\top} \Delta \vec \pi^i$.
%
For every edge $e=(v,w)$ we have $\pi_w - \pi_v = \vec \gamma_e^{\top} \vec \pi$. We then compute 
\begin{align*}
\epsilon(e) :=
	\begin{cases}
		\frac{\sigma_{e, t^i_e+1} - \vec \gamma_e^{\top} \vec \pi^i}{\vec \gamma_e^{\top} \Delta \vec \pi^i}  & \text { if } \vec \gamma_e^{\top} \Delta \vec \pi^i > 0,\\
		\frac{\sigma_{e, t^i_e\phantom{+1}} - \vec \gamma_e^{\top} \vec \pi^i}{\vec \gamma_e^{\top} \Delta \vec \pi^i} & \text{ if } \vec \gamma_e^{\top} \Delta \vec \pi^i < 0,\\
		\infty & \text{ else.}
	\end{cases}
\end{align*}
Determining 
$\epsilon := \min_{e \in E} \epsilon (e)$
we can find the closest boundary. This definition ensures that $\vec \pi^i + \lambda \, \Delta \vec \pi^i$ is in the region $R_{\vec t^i}$ for all $0 \leq \lambda \leq \epsilon$. If there is exactly one edge 
$e^* \in E^* := \argmax_{e \in E} \epsilon (e)$
we proceed with the boundary crossing. If there are more than one edge for that $\epsilon(e)$ is minimal, we call the point $\vec \pi^i + \epsilon \Delta \vec \pi^i$ \emph{degenerate}. In a degenerate point, multiple boundaries induced by different edges intersect (e.g. the point $(2,1)$ in Figure~\ref{fig:regions}~(c)\shortversiononly{ in Appendix~\ref{appendixfigure\thesection subsection}}). We will discuss these points below and assume for now that there is an unique boundary to be crossed (e.g. the point $(1,2)$ in Figure~\ref{fig:regions}~(c)\shortversiononly{ in Appendix~\ref{appendixfigure\thesection subsection}}).
At $\lambda^{i+1} := \lambda^{i} + \epsilon$ is the next breakpoint of the solution. The potential at this point is
$\vec \pi^{i+1} := \vec \pi^i + \epsilon \Delta \vec \pi^i$
and value of the flow function at the breakpoint is
$\smash{f(\lambda^{i+1}) := \vec N^{i+1} := \vec C_{\vec t^i} \vec \Gamma^{\top} \Delta \vec \pi^i - \vec d_{\vec t^i}}$.

\subparagraph{Non-degenerate Points}

We call any two regions $R_{\vec{t}^1}, R_{\vec{t}^2}$ \emph{neighboring} if
\begin{flexequation} \label{eq:neighborRegionVector}
\vec{t}^2 = \vec{t}^1 \pm \vec{u}_e
\end{flexequation}
where $\vec{u}_e$ is the unit vector with the non-zero entry in the component corresponding to the edge $e$. The edge $e$ is called \emph{boundary edge} in this case. Therefore, a neighboring region is a region with a region vector that differs by one in exactly one component.

When the solution curve hits a unique boundary induced by some edge $e$ it moves on to the region $R_{\vec{t}^{i+1}}$ behind the boundary which is identified by the vector $\vec{t}^{i+1} = \vec{t}^i + \sgn\big( \vec{\gamma}_e^{\top} \, \Delta \vec{\pi}^i \big) \vec{u}_e$. 

The following theorem gives us a connection between the inverses of the reduced Laplacian matrices of two neighbouring regions and the direction vectors $\Delta \vec \pi$ in these region and was proven in similar form by Fujisawa and Kuh~\cite{fujisawa1972piecewise}. Intuitively, it states that the new inverse of the Laplacian can be computed by a rank 1 transformation of the old Laplacian.
\longversion{
\begin{theorem} \label{thm:boundaryCrossing}
Let $R_{\vec t^i}$ and $R_{\vec t^{i+1}}$ be two neighbouring regions with the boundary edge $e$ and $\Delta c_{e} := c_{e, t^{i+1}_{e}} - c_{e, t^i_{e}}$. Then
\begin{compactenum}[(i)]
\item \label{thm:boundaryCrossing:Laplacians}
	The inverse of the reduced Laplacian matrix in region $R_{\vec t^{i+1}}$ is
	\begin{align*}
	\hat{\vec L}_{\vec t^{\sCurveI+1}}^{-1} = \Big(
	\vec{I}_{n-1} - \frac{\Delta c_{e}}{1 + \Delta c_{e} \hat{\vec{\gamma}}_{e}^{\top} \hat{\vec L}_{\vec t^{\sCurveI}}^{-1} \hat{\vec{\gamma}}_{e} } \, \hat{\vec L}_{\vec t^{\sCurveI}}^{-1} \, \hat{\vec{\gamma}}_{e} \, \hat{\vec{\gamma}}_{e}^{\top}
	\Big) \hat{\vec L}_{\vec t^{\sCurveI}}^{-1}
	\end{align*}
	where $\vec{I}_{n-1}$ is the identity matrix in $\R^{n-1 \times n-1}$ and $\hat{\vec \gamma}_{e}$ denotes the vector $\vec \gamma_{e}$ without the first component.
\item \label{thm:boundaryCrossing:Directions}
	Let $\Delta \vec \pi^i$ and $\Delta \vec \pi^{i+1}$ be the direction vectors in the regions $R_{\vec t^i}$ and $R_{\vec t^{i+1}}$. Then
	$\sgn \big( \vec \gamma_{e}^{\top} \Delta \vec \pi^{i} \big) = \sgn \big( \vec \gamma_{e}^{\top} \Delta \vec \pi^{i+1} \big)$.
\end{compactenum}
\end{theorem}
}{
\begin{theorem} \label{thm:boundaryCrossing}
For two neighbouring regions $R_{\vec t^i}$ and $R_{\vec t^{i+1}}$ with direction vectors $\Delta \vec \pi^i$ and $\Delta \vec \pi^{i+1}$ and boundary edge $e$ with $\Delta c_{e} := c_{e, t^{i+1}_{e}} - c_{e, t^i_{e}}$:
\begin{compactenum}[(i)]
\item \label{thm:boundaryCrossing:Laplacians}
	$
	\hat{\vec L}_{\vec t^{\sCurveI+1}}^{-1} = \Big(
	\vec{I}_{n-1} - \frac{\Delta c_{e}}{1 + \Delta c_{e} \hat{\vec{\gamma}}_{e}^{\top} \hat{\vec L}_{\vec t^{\sCurveI}}^{-1} \hat{\vec{\gamma}}_{e} } \, \hat{\vec L}_{\vec t^{\sCurveI}}^{-1} \, \hat{\vec{\gamma}}_{e} \, \hat{\vec{\gamma}}_{e}^{\top}
	\Big) \hat{\vec L}_{\vec t^{\sCurveI}}^{-1}$
	where $\vec{I}_{n-1}$ is the identity matrix in $\R^{n-1 \times n-1}$.
\item \label{thm:boundaryCrossing:Directions}
	$\sgn \big( \vec \gamma_{e}^{\top} \Delta \vec \pi^{i} \big) = \sgn \big( \vec \gamma_{e}^{\top} \Delta \vec \pi^{i+1} \big)$.
\end{compactenum}
\end{theorem}
}

\begin{appendixproof}{Theorem~\ref{thm:boundaryCrossing}}
By the definition of $\Delta c_{e}$ we have
$
\vec{C}_{\vec t^{\sCurveI+1}} = \vec{C}_{\vec t^{\sCurveI}} + \diag(0, \dotsc, 0, \Delta c_{e}, 0, \dotsc, 0)
$.
Therefore we have for the Laplacian matricies that
$
\vec L_{\vec t^{\sCurveI+1}} =\vec \Gamma \vec{C}_{\vec t^{\sCurveI+1}} \vec \Gamma^{\top} = \vec L_{\vec t^\sCurveI} + \Delta c_{e^*} \vec \gamma_{e} \vec \gamma_{e}^{\top}
$
which directly implies
\begin{align} \label{eq:laplaceConnection}
\hat{\vec L}_{\vec t^{\sCurveI+1}} = \hat{\vec L}_{\vec t^{\sCurveI}} + \Delta c_{e} \hat{\vec{\gamma}}_{e} \hat{\vec{\gamma}}_{e}^{\top}
\end{align}
for the reduced Laplacian matrices. Using the Sherman-Morrison formula \cite{hager1989updating, sherman1950adjustment}, we directly obtain statement $(i)$.
For any two vectors $\vec{v}, \vec{w} \in \R^{n}$, the determinant identity $\det(\vec{I}_n + \vec{v} \vec{w}^{\top}) = 1+\vec{w}^{\top} \vec{v}$ holds true (see \cite[Corollary 18.1.3]{harville1997matrix}). Hence, we get
\begin{align*}
\det ( \hat{\vec L}_{\vec t^{\sCurveI+1}} ) &= \det \big(\vec I_{n-1} + \Delta c_{e} \hat{\vec{\gamma}}_{e} \hat{\vec{\gamma} }_{e}^{\top} \hat{\vec L}_{\vec t^\sCurveI}^{-1}\big)  \det (\hat{\vec L}_{\vec t^\sCurveI})
= (1 + \Delta c_{e} \hat{\vec{\gamma} }_{e}^{\top} \hat{\vec L}_{\vec t^\sCurveI}^{-1}\hat{\vec{\gamma} }_{e})
\det (\hat{\vec L}_{\vec t^\sCurveI}) 
\end{align*}
and thus
\begin{align*}
\hat{\vec \gamma}_{e}^{\top} \Delta \vec \pi^{i+1} 
&= \hat{\vec \gamma}_{e}^{\top} 
 \Big(
	\vec{I}_{n-1} - \frac{\Delta c_{e}}{1 + \Delta c_{e} \hat{\vec{\gamma}}_{e}^{\top} \hat{\vec L}_{\vec t^{\sCurveI}}^{-1} \hat{\vec{\gamma}}_{e} } \, \hat{\vec L}_{\vec t^{\sCurveI}}^{-1} \, \hat{\vec{\gamma}}_{e} \, \hat{\vec{\gamma}}_{e}^{\top}
	\Big) \hat{\vec L}_{\vec t^{\sCurveI}}^{-1}
 \Delta \hat{\vec{y}}  \\
&=
\frac{1}{1+ \Delta c_{e} \hat{\vec \gamma}_{e}^{\top} \hat{\vec L}_{\vec t^\sCurveI}^{-1} \hat{\vec \gamma}_{e}} \, \hat{\vec \gamma}_{e}^{\top} \hat{\vec L}_{\vec t^\sCurveI}^{-1} \Delta \hat{\vec{y}} 
= \frac{\det ( \hat{\vec L}_{\vec t^{\sCurveI+1}} )}{\det ( \hat{\vec L}_{\vec t^{\sCurveI}} )} \hat{\vec \gamma}_{e}^{\top} \Delta \vec \pi^{i+1}
.
\end{align*}
The matrices $\hat{\vec L}_{\vec t^\sCurveI}$ and $\hat{\vec L}_{\vec t^\sCurveI + 1}$ are positive definite by Lemma~\ref{lem:LaplacianMatrix}(\ref{lem:LaplacianMatrix:subMatrix}) which implies $\frac{\det (\hat{\vec L}_{\vec t^\sCurveI}) }{\det ( \hat{\vec L}_{\vec t^{\sCurveI+1}} )} > 0$ and finally proves \eqref{thm:boundaryCrossing:Directions}.
\end{appendixproof}

The second part of the theorem in particular implies that if the solution curve crosses a boundary its direction vector is oriented away from the boundary hyperplane after crossing the boundary. To see this, note that $\vec \gamma_{e}$ is the normal vector of the hyperplane that is crossed (see Figure~\ref{fig:regions}~(c)\shortversiononly{ in Appendix~\ref{appendixfigure\thesection subsection}}). In order to reach the boundary the direction vector $\Delta \vec \pi^i$ must not be orthogonal to $\vec \gamma_{e}$. By Theorem~\ref{thm:boundaryCrossing}~(\ref{thm:boundaryCrossing:Directions}) the vector $\Delta \vec \pi^{i+1}$ has the same orientation as $\Delta \vec \pi^i$ with respect to the normal vector $\vec \gamma_{e}$ and since $\Delta \vec \pi^i$ is directed towards the boundary from region $R_{\vec t^i}$ the vector $\Delta \vec \pi^{i+1}$ is directed away from the boundary in region $R_{\vec t^{i+1}}$. A pseudocode summarizing the steps of the algorithm (in the non-degenerate case) is given in Appendix~\ref{app:pseudo} and a concrete example in Appendix~\ref{app:examples}.

\subparagraph{Degenerate Points}

Assume the solution curve is in a point $\vec \pi^i$ and the next point $\vec \pi^{i+1} = \vec \pi^i + \epsilon \Delta \vec \pi^i$ is degenerate, i.e.,  the set
$
E^* := \argmin_{e \in E} \epsilon (e)
$
contains more than one edge and we thus have more than one choice to cross a boundary in this point. This means that the boundaries induced by the edges in $E^*$ intersect in the point $\vec{\pi}^{i+1}$. Intuitively, when adding a small perturbation to the point $\vec \pi^i$, the solution curve will not hit the degenerate point (at least as long the perturbation is not directed along $\Delta \vec{\pi}^i$) and traverse the intersecting regions until it moves away from the point (see Figure~\ref{fig:perturbedSolution}\shortversiononly{ in Appendix~\ref{appendixfigure\thesection subsection}}). We will now use this fact to derive a lexicographic rule for dealing with degenerate points without actually perturbing the solution curve.

\begin{appendixfigure}{t}
\centering
\begin{center}
\begin{tikzpicture}
\newcommand{\xMin}{-3} \newcommand{\xMax}{3}
\newcommand{\yMin}{-2} \newcommand{\yMax}{2}
\path[fill=color1,fill opacity=0.25] (\xMin,0) -- (0,0) -- ({max(\xMin,\yMin)}, {max(\xMin,\yMin)}) -- (\xMin,\yMin);
\path[fill=color2,fill opacity=0.25] ({max(\xMin,\yMin)},{max(\xMin,\yMin)}) -- (0,0) -- (0,\yMin);
\path[fill=color3,fill opacity=0.25] (0,\yMin) -- (0,0) -- (\xMax,0) -- (\xMax,\yMin);
\path[fill=color4,fill opacity=0.25] (\xMax,0) -- (0,0) -- ({min(\xMax,\yMax)},{min(\xMax,\yMax)}) -- (\xMax,\yMax);
\path[fill=color5,fill opacity=0.25] ({min(\xMax,\yMax)},{min(\xMax,\yMax)}) -- (0,0) -- (0,\yMax);
\path[fill=color6,fill opacity=0.25] (0,\yMax) -- (0,0) -- (\xMin,0) -- (\xMin,\yMax);
\draw[thin, dashed] (0,\yMin) -- (0,\yMax);
\draw[thin, dashed] ({max(\xMin,\yMin)}, {max(\xMin,\yMin)}) -- ({min(\xMax,\yMax)},{min(\xMax,\yMax)});
\draw[thin, dashed] (\xMin,0) -- (\xMax,0);

\node[anchor=north west, color1!80!black] at (\xMin, 0) {$\regionp[0]$};
\node[anchor=south, color2!80!black] at ({max(\xMin,\yMin)/2}, \yMin) {$\regionp[1]$};
\node[anchor=south, color3!80!black] at ({\xMax/2}, \yMin) {$\regionp[2]$};
\node[anchor=east, color4!80!black] at (\xMax, {\yMax/2}) {$\regionp[3]$};
\node[anchor=north east, color5!80!black] at ({min(\xMax,\yMax)-.3},{min(\xMax,\yMax)}) {$\regionp[4]$};

\draw (-1.5,-0.75) node[solid] {} node[above] {$\vec \pi^i$};
\draw (0,0) node[solid] {} node[above left] {$\vec \pi^{i+1}$};
\draw[thick] (-1.5,-0.75)--(0,0)--(0.5,2);

\draw[thick, darkgreen] (-1.7,-1.2)--(-0.7,-0.7)--(0,-0.5)--(0.25,0)--(0.5,0.5)--(0.875,2);
\draw (-1.7,-1.2) node[solid, darkgreen] {} node[left, darkgreen ] {$\potentialVp[0]$};
\draw (-0.7,-0.7) node[solid, darkgreen] {} node[below=.2, darkgreen ] {$\potentialVp[1]$};
\draw (0,-0.5) node[solid, darkgreen] {} node[below right, darkgreen ] {$\potentialVp[2]$};
\draw (0.25,0) node[solid, darkgreen] {} node[below right, darkgreen ] {$\potentialVp[3]$};
\draw (0.5,0.5) node[solid, darkgreen] {} node[right, darkgreen ] {$\potentialVp[4]$};

\end{tikzpicture}
\end{center}
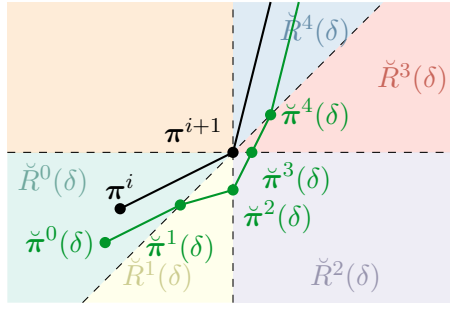
\caption{The solution curve (black) hits the degenerate point $\vec \pi^{i+1}$. By perturbing $\vec \pi^i$, the perturbed solution curve (green) traverses the intersecting regions avoiding the degenerate point.}
\label{fig:perturbedSolution}
\end{appendixfigure}

We call a vector $\vec v \in \R^n$ \textit{lexicographically smaller} than a vector $\vec w \in \R^n$ (denoted by $\vec v \lexle \vec w$) if there is some $j \in \{1, \dotsc, n\}$ with $v_{j} < w_{j}$ and $v_{n-j'} = w_{n-j'}$ for all $1 \leq j' < j$. For some $\delta > 0$, we define the \textit{perturbation vector}
$\smash{\vec \delta := (0, \delta^{n-1}, \delta^{n-2}, \dotsc, \delta^2, \delta )}$.
For further reference, we note that for any vector $\vec{v} \in \R^n$ the value $\vec{v}^{\top} \vec{\delta}$ is a polynomial in $\delta$ and that for two vectors $\vec{v}, \vec{w} \in \R^n$ with $v_1 = w_1$ there is a $\delta^* > 0$ such that 
$\vec v^{\top} \vec \delta < \vec w^{\top} \vec \delta$
for all $0 < \delta < \delta^*$ if and only if
$\vec v \lexle \vec w$\footnote{Note that we are going to add the perturbation vector to potentials. Thus, we do not perturb the first component since $\pi_1 = 0$ is constant.}.

%

To bypass the degenerate point $\vec{\pi}^{i+1}$, we perturb the potential $\vec{\pi}^i$ with the vector $\vec{\delta}$ and track the course of the solution curve from this point on. We refer to this curve as the \textit{perturbed solution curve}. We denote the regions traversed by this perturbed solution curve by $\regionp[l]$ and the points in the potential space where the perturbed solution curve enters these regions as $\potentialVp[l]$. We index these values by $0 \leq l \leq \alpha$ where $\alpha$ is the number of steps the perturbed solution curve takes.

As we are solely interested in the behavior of the solution curve in the neighborhood of the degenerate point $\vec{\pi}^{i+1}$ we only consider the boundaries intersecting in this point. (If $\delta$ is small enough the continuity ensures that the perturbed solution curve will only cross these boundaries before moving away from the degenerate point.) These boundaries are induced by the edges in $E^*$ and some specific breakpoints of the inverse cost function that we will denote by $\potentialBPd_e$.

By definition, the perturbed solution curve starts in the point $\potentialVp[0] := \vec{\pi}^i + \vec{\delta}$ in the region $\regionp[0] = R_{\vec{t}^i}$. In every step, we compute the distance to the boundary induced by $e \in E^*$ as $\smash{\epsilonp[l] = \frac{\potentialBPd_e  - \incidenceV_e^{\top} \potentialVp[l]}{\incidenceV_e^{\top} \potentialDp[l]}}$ where $\potentialDp[l]$ is the direction vector in the region $\regionp[l]$. We use the convention that $\epsilonp[l] = \infty$ if the denominator is zero. Note, that some or all values $\epsilonp[l]$ can be negative if the solution curve moves away from some or all boundaries intersecting in the degenerate point. 
We define by $E^*_l (\delta) := \{ e \in E^* : 0 < \epsilonp[l] < \infty \}$ the set of edges which lie in front of the  perturbed solution curve in step $l$. If $E^*_l = \emptyset$ the perturbed solution curve moves away from all boundaries, meaning the the real solution curve can proceed in exactly this region $\regionp[l]$. Otherwise, we denote by $e^*_l$ the edge minimizing $\epsilonp[l]$ for all $e \in E^*_l (\delta)$. This edge $e^*_l$ induces the boundary the solution curve crosses next in the point $\potentialVp[l+1]$.
Below, we will show that for $\delta$ small enough, there is always an unique choice for $e^*_l$. First, we observe that the potential points $\potentialVp[l]$ and the distances to the boundary $\epsilonp[p]$ are a affine linear function of the perturbation $\vec{\delta}$.

\begin{lemma} \label{lem:perturbedValues}
There are matrices $\vec{M}^l \in \R^{n \times n}$ such that $\potentialVp[0] = \vec{\pi}^i + \vec{M}^0 \vec{\delta}$ and $\potentialVp[l] = \vec{\pi}^{i+1} + \vec{M}^{l} \vec{\delta}$ for every $1 \leq l \leq \alpha$. Further, there are vectors $\vec{m}_{l, e} \in \bar{\R}^{n}$ such that $\epsilonp[0] = \epsilon(e) + \vec{m}_{0, e}^{\top} \vec{\delta}$ and $\epsilonp[l] = \vec{m}_{l, e}^{\top} \vec{\delta}$ for all $1 \leq l \leq \alpha$ and $e \in E^*$.
\end{lemma}
\begin{appendixproof}{Lemma~\ref{lem:perturbedValues}}
We define the matrices 
\[
\tilde{\vec{M}}^l := \vec{I}_n - \frac{1}{\incidenceV_{ e^*_{l-1} }^{\top} \potentialDp[l-1]} \potentialDp[l-1] \, \incidenceV_{ e^*_{l-1} }^{\top}
\quad \text{for } 1 \leq l \leq \alpha
\]
and $\vec{M}^l := \prod_{j=0}^{l-1} \tilde{\vec{M}}^{l-j}$ for $0 \leq l \leq \alpha$. (For $l=0$ the empty product is defined as $\vec{M}^0 = \vec{I}_n$.) Further, we define the vectors
\[
\vec{m}_{l,e}^{\top} := - \frac{1}{\incidenceV_{ e }^{\top} \potentialDp[l]} \, \incidenceV_{ e }^{\top} \, \vec{M}^l
\quad \text{for } 0 \leq l \leq \alpha
\]
where the vector $\vec{m}_{l,e}$ is set to infinity in every component if $\incidenceV_{ e }^{\top} \potentialDp[l] = 0$. \\
For $l=0$, we get $\potentialVp[0] = \vec{\pi}^{i} + \vec{M}^0 \vec{\delta}$ by definition and obtain
\[
\epsilonp[0] =
\frac{\potentialBPd_e - \incidenceV_e^{\top} \potentialVp[0]}{\incidenceV_e^{\top} \potentialDp[0]}
= \frac{\potentialBPd_e - \incidenceV_e^{\top} \vec{\pi}^{i} }{\incidenceV_e^{\top} \Delta \potentialV^i} - \frac{1}{\incidenceV_{ e }^{\top} \potentialDp[0]} \, \incidenceV_{ e }^{\top} \, \vec{M}^0 \vec{\delta}
= \epsilon(e) + \vec{m}_{0,e}^{\top} \vec{\delta}
\]
where we used that $\potentialDp[0] = \Delta \potentialV^i$. \\ 
For $l=1$, we obtain
\begin{align*}
\potentialVp[1] &= \potentialVp[0] + \epsilonp^0 (e^*_0, \delta) \, \potentialDp[0] \\
&= \vec{\pi}^{i} + \epsilon(e) \, \Delta \potentialV^i + \vec{M}^0 \, \vec{\delta} + \vec{m}_{0, e_0^*}^{\top} \, \vec{\delta} \, \potentialDp[0] \\
&= \vec{\pi}^{i+1} + \vec{M}^0 \, \vec{\delta} + \potentialDp[0] \, \vec{m}_{0, e_0^*}^{\top} \, \vec{\delta} \\
&= \vec{\pi}^{i+1} + \tilde{\vec{M}}^1 \, \vec{M}^0 \, \vec{\delta} = \vec{\pi}^{i+1} + \vec{M}^1 \, \vec{\delta}
\end{align*}
and with $\potentialBPd_e = \incidenceV_e^{\top} \vec{\pi}^{i+1}$ also 
$
\epsilonp[1] = \frac{\potentialBPd_e  - \incidenceV_e^{\top} \potentialVp[1]}{\incidenceV_e^{\top} \potentialDp[1]}
= \vec{m}_{1,e}^{\top} \, \vec{\delta}
$. \\
By induction, we finally get for $l \geq 2$ that
\begin{align*}
\potentialVp[l] &= \potentialVp[l-1] + \epsilonp^{l-1} (e^*_{l-1}, \delta) \, \potentialDp[l-1] \\
&= \vec{\pi}^{i+1} + \vec{M}^{l-1} \, \vec{\delta} + \vec{m}_{l-1, e^*_{l-1}}^{\top} \, \vec{\delta} \, \potentialDp[l-1] \\
&= \vec{\pi}^{i+1} + \vec{M}^{l-1} \, \vec{\delta} + \potentialDp[l-1] \, \vec{m}_{l-1, e^*_{l-1}}^{\top} \, \vec{\delta} \\
&= \vec{\pi}^{i+1} + \tilde{\vec{M}}^l \, \vec{M}^{l-1} \, \vec{\delta} = \vec{\pi}^{i+1} + \vec{M}^l \, \vec{\delta}
\end{align*}
and again also 
$
\epsilonp[l] = \frac{\potentialBPd_e  - \incidenceV_e^{\top} \potentialVp[l]}{\incidenceV_e^{\top} \potentialDp[l]}
= \vec{m}_{l,e}^{\top} \, \vec{\delta}
$.
\end{appendixproof}

Lemma~\ref{lem:perturbedValues} shows that the distances to the boundaries $\epsilonp[l]$ are polynomials in $\delta$. Note, that we in particular admit $\infty$-values in the vectors $\vec{m}_{l,e}$ since some $\epsilonp[l]$ values might be infinite. In this case $\vec{m}_{l,e} = \vec{\infty}$ is the vector containing only $\infty$ values.
For $\delta > 0$ small enough, we can now express all sets and values that are necessary for determining the edge $e^*_l$ with the vectors $\vec{m}_{l,e}$. We have, $E^*_0 = E^*$ (since $\epsilon(e) > 0$ for all $e \in E^*$) and $E^*_l = \{ e \in E^* : \vec{0} \lexle \vec{m}_{l,e} \lexle \vec{\infty} \}$. Note that for all $\delta > 0$ small enough the set $E^*_l$ is independent of $\delta$. The perturbed solution curve thus moves away from all boundaries if all vectors $\vec{m}_{l,e}$ are lexicographic negative (or infinity).
The edge $e^*_l$ is the edge in $E^*_l$ with the lexicogrpahically smallest vector $\vec{m}_{l, e^*}$. Thus, determining the boundaries that the perturbed solution curve crosses reduces to iteratively computing the vectors $\vec{m}_{l,e}$ and finding the lexicographically smallest, positive one.
The next Theorem proves, that there is indeed always a unique boundary to cross.

\begin{theorem} \label{thm:perturbationUniqueMinimum}
For every $0 \leq l \leq \alpha - 1$ there is a unique lexicographic minimum in $\{ \vec{m}_{l,e} : e \in E^*_l \}$.
\end{theorem}
\begin{appendixproof}{Theorem~\ref{thm:perturbationUniqueMinimum}}
If there are not at least two edges in $E^*_l$ then there is nothing to show. Otherwise, we need to show that for any two edges $e_1, e_2 \in E^*_l$ we have $\vec{m}_{l,e_1} \neq \vec{m}_{l, e_2}$. For $l=0$, the vectors $\vec{m}_{0,e}$ are multiples of the vectors $\vec{\gamma}_e$ which are pairwise linear independent since we assumed that there are no parallel edges. \\
For $l \geq 1$, we claim that 
\begin{equation} \label{eq:thm:perturbationUniqueMinimum:1}
\vec{v}^{\top} \, \vec{M}^{l} = \vec{0} \text{ if and only if } \vec{v} = \beta \, \vec{\gamma}_{e^*_{l-1}}
\end{equation}
 for some $\beta \in \R$. Observe that the matrix $\tilde{\vec{M}}^l$ has a rank of at least $n-1$ (as it is the difference of the full-rank identity matrix and a rank one matrix) and thus the (left) nullspace of $\tilde{\vec{M}}^l$ has a dimension of at most $1$. It can be easily verified that $\smash{\vec{\gamma}^{\top}_{e^*_{l-1}} \, \tilde{\vec{M}}^l = \vec{0}}$, so we established the claim for the matrices $\tilde{\vec{M}}^l$. Since $\vec{M}^1 = \tilde{\vec{M}}^1$ the claim also holds true for $\vec{M}^1$. 
Using induction and the recursion formula $\vec{M}^l = \tilde{\vec{M}}^l \, \vec{M}^{l-1}$ we get $\smash{\vec{v} = \beta \, \vec{\gamma}_{e^*_{l-1}}} \Rightarrow \vec{v}^{\top} \, \vec{M}^{l} = \vec{0}$ and $\vec{v}^{\top} \, \vec{M}^{l} = \vec{0} \Rightarrow \smash{\vec{v} = \beta \, \vec{\gamma}_{e^*_{l-1}}}$ or $\smash{\vec{v}^{\top} \, \tilde{\vec{M}}^{l} = \beta \, \vec{\gamma}_{e^*_{l-2}}^{\top}}$. 
Again, it is easy to verify by definition that $\tilde{\vec{M}}^l \, \potentialDp[l-1] = \vec{0}$, thus $\smash{\vec{v}^{\top} \, \tilde{\vec{M}}^{l} = \beta \, \vec{\gamma}_{e^*_{l-2}}^{\top}}$ would imply $\smash{\vec{\gamma}_{e^*_{l-2}}^{\top} \, \potentialDp[l-1] = 0}$. With Theorem~\ref{thm:boundaryCrossing}(\ref{thm:boundaryCrossing:Directions}) we obtain $\smash{\vec{\gamma}_{e^*_{l-2}}^{\top} \, \potentialDp[l-2] = 0}$ which means that $\epsilonp^{l-2} (e^*_{l-2}, \delta) = \infty$ and, thus, $e^*_{l-2}$ could be chosen as the edge minimizing $\epsilonp[l-2]$ which is a contradiction. Hence, the claim follows. \\
Finally, assume that $\vec{m}_{l,e_1} = \vec{m}_{l,e_2}$. Then
\[
\vec{0} = \vec{m}_{l,e_1} - \vec{m}_{l,e_2} = 
-
\Big( 
\frac{1}{\vec{\gamma}_{e_1}^{\top} \potentialDp[l]} \vec{\gamma}_{e_1}^{\top}
-
\frac{1}{\vec{\gamma}_{e_2}^{\top} \potentialDp[l]} \vec{\gamma}_{e_2}^{\top}
\Big)
\vec{M}^l
\]
which together with \eqref{eq:thm:perturbationUniqueMinimum:1} implies that $\vec{\gamma}_{e^*_{l-1}} = \beta \Big(  \frac{1}{\vec{\gamma}_{e_1}^{\top} \potentialDp[l]} \vec{\gamma}_{e_1} - \frac{1}{\vec{\gamma}_{e_2}^{\top} \potentialDp[l]} \vec{\gamma}_{e_2} \Big)$. Thus,
\[
\vec{\gamma}_{e^*_{l-1}}^{\top} \potentialDp[l-1] = \vec{\gamma}_{e^*_{l-1}}^{\top} \potentialDp[l] = \beta \Big(  \frac{1}{\vec{\gamma}_{e_1}^{\top} \potentialDp[l]} \vec{\gamma}_{e_1}^{\top} - \frac{1}{\vec{\gamma}_{e_2}^{\top} \potentialDp[l]} \vec{\gamma}_{e_2}^{\top} \Big) \potentialDp[l] = 0
\]
where we again used Theorem~\ref{thm:boundaryCrossing}(\ref{thm:boundaryCrossing:Directions}) for the first equality. But this is, as in the proof of the subclaim above, a contradiction to the fact $e^*_{l-1}$ was chosen to be the minimizer of $\epsilonp[l-1]$.

\end{appendixproof}

Overall, we can retrace the perturbed solution curve by iteratively finding the lexicographically smallest vector $\vec m_{l,e}$ and computing the matrices $\vec M^l$. If, finally, $E^*_l = \emptyset$ (which is the case if all vectors $\vec m_{l, e}$ are lexicographically non-positive or infinite) the perturbed solution curve is in a region where the direction vector is directed away from all boundaries intersecting in the degenerate point. This is the region where the unperturbed solution curve can proceed. If there are finitely many boundaries intersecting in the degenerate point then the perturbed solution curve performs also only finitely many boundary crossings (see Theorem~\ref{thm:termination} in the next subsection) and thus there are only finitely many steps in the lexicographic rule until $E^*_l = \emptyset$. 
For a concrete example of the lexicographic rule and the computation of the matrices $\vec{M}$ and the vectors $\vec{m}$ consider Example~\ref{ex:lexicographicRule} in Appendix~\ref{app:examples}. A pseudocode of the lexicographic rule is given in Appendix~\ref{app:pseudo}.

\subparagraph{Termination and Time Complexity}
When a cost function has an infinite number of breakpoints the algorithm may not terminate in finite time even for a single edge. In contrast, the algorithm terminates in finite time when the number of breakpoints is finite. 

\begin{theorem} \label{thm:termination}
The algorithm considers every region $R_{\vec{t}}$ at most once.
\end{theorem}
\begin{appendixproof}{Theorem~\ref{thm:termination}}
Let $\vec t^{i_1}$ and $\vec t^{i_2}$ be two region vectors at two different iterations $i_1$ and $i_2$. We will prove that $\vec t^{i_1} \neq \vec t^{i_2}$, i.e. no region $R_{\vec t}$ is visited twice during the algorithm.
Note that every two successive iterations $i_1, i_2$ in the algorithm satisfy $\vec{t}_{i_1} \neq  \vec{t}_{i_2}$ by definition, i.e. the region changes in every iteration.
If we assume that there are two different iterations $i_1 < i_2$ with $\vec t^{i_1} = \vec t^{i_2} = \vec{t}$ we thus in particular know that $i_2 > i_1 + 1$. So the potential
$
\vec{\pi}^{i_1 + 1} = \vec{\pi}^{i_1} + \epsilon^{i_1} \, \Delta \vec{\pi}^{i_1}
$
lies on the boundary of the region $R_{\vec{t}}$. Since the regions are convex we have
\begin{equation} \label{eq:thm:termination:1}
\vec{\pi}^{i_1} + \lambda \, \Delta \vec{\pi}^{i_1} \notin R_{\vec{t}}
\quad \text{for all} \quad
\lambda > \epsilon^{i_1} = \lambda_{i_1 + 1} - \lambda_{i_1}
.
\end{equation}
By definition of the Laplace matrix $\vec{L}_{\vec{t}}$ and the breakpoints $\lambda_i$ we have
\begin{align*}
\lambda_{i_j} \Delta \, \vec{y} = \vec{L}_{\vec{t}} \, \vec{\pi}^{i_j} - \tilde{\vec{d}}_{\vec{t}}
\end{align*}
for $j=1,2$. This implies
\begin{alignat*}{2}
&&(\lambda_{i_2} - \lambda_{i_1}) \, \Delta \hat{\vec{y}} &= \hat{\vec{L}}_{\vec{t}} \, ( \hat{\vec{\pi}}^{i_2} - \hat{\vec{\pi}}^{i_1} ) \\
&\Leftrightarrow \quad &
(\lambda_{i_2} - \lambda_{i_1}) \, \hat{\vec{L}}_{\vec{t}}^{-1} \Delta \hat{\vec{y}} &=  \hat{\vec{\pi}}^{i_2} - \hat{\vec{\pi}}^{i_1} \\
&\Leftrightarrow &
 (\lambda_{i_2} - \lambda_{i_1}) \, \Delta \hat{\vec{\pi}}^{\vec{t}} &=  \hat{\vec{\pi}}^{i_2} - \hat{\vec{\pi}}^{i_1} 
\end{alignat*}
where $\Delta \hat{\vec{\pi}}^{\vec{t}} := \Delta \hat{\vec{\pi}}^{i_1} = \Delta \hat{\vec{\pi}}^{i_2}$ is the direction vector in region $R_{\vec{t}}$. Thus
\begin{align*}
\vec{\pi}^{i_2} = \vec{\pi}^{i_1} + (\lambda_{i_2} - \lambda_{i_1}) \Delta \vec{\pi}^{\vec{t}}
\end{align*}
and with $\vec{\pi}^{i_1}, \vec{\pi}^{i_2} \in R_{\vec{t}}$ and the convexity of the region $R_{\vec{t}}$ we get that
\begin{align*}
\vec{\pi}^{i_1} + \lambda \, \Delta \vec{\pi}^{\vec{t}} \in R_{\vec{t}}
\end{align*}
for all $\lambda \in [0, \lambda_{i_2} - \lambda_{i_1}]$. This is a contradiction to \eqref{eq:thm:termination:1} since $\lambda_{i_1} < \lambda_{i_1 +1} < \lambda_{i_2}$.
\end{appendixproof}

The proof of Theorem~\ref{thm:termination} uses the fact that if there are two different points $\vec \pi^1$ and $\vec \pi^2$ for two different demands $\lambda^1$ and $\lambda^2$ that lie in the same region $R_{\vec{t}}$ then by convexity of the regions all potentials on the line segment between $\vec \pi^1$ and $\vec \pi^2$ are in the region $R_{\vec t}$. Furthermore, by linearity, all the potentials on this line segments are potentials for demands between $\lambda^1$ and $\lambda^2$. In particular, the difference $\vec \pi^2 - \vec \pi^1$ must be a multiple of the direction vector $\Delta \vec \pi$ in the region $R_{\vec t}$. Thus, if there are two potentials on the solution curve that are in the same region, the solution cannot have left the region in between.

\begin{corollary}
The number of iterations of the algorithm is bounded by the number of regions. In particular,
if every cost function has a finite number of breakpoints, the algorithm terminates in finite time.
\end{corollary}

When there is no degenerate point visited in the course of the algorithm then every iteration takes only polynomial time in the size of the graph. The time critical step is the rank one update of the inverse $\hat{\vec{L}_{\vec{t}}}^{-1}$ which can be done with the Sherman-Morrison formula as shown in Theorem~\ref{thm:boundaryCrossing}. Hence, every iterations has a time complexity of $\mathcal{O}(n^2)$. Further, in the non-degenerate case, every iteration computes exactly one function part of the output Wardrop equilibrium flow functions so that the algorithm is output-polynomial time. The first inverse $\hat{\vec{L}}_{\vec{t}_1}^{-1}$ can be computed in $\mathcal{O}(n^{2.4})$ time by using fast matrix multiplication like the Coppersmith–Winograd algorithm \cite{coppersmith1987matrix}. Thus the overall time complexity is $\mathcal{O}(n^{2.4} + n^2 K)$ where $K$ is the number of breakpoints of the output function.

There are networks that yield Wardrop equilibrium functions with exponentially many different function parts (see Appendix~\ref{app:braess} for a family of such graphs). Thus, the algorithm is not polynomial in the input.
If there are degenerate points that are visited in the course of the algorithm, the algorithm may not have polynomial runtime in the output size since the perturbed solution curve may traverse exponentially many regions around the degenerate point leading to exponentially many steps while applying the lexicographic rule in one degenerate point. In this case, the computation of the direction can done by solving an associated quadratic convex program (see Appendix~\ref{app:convex-program}) which can be done in polynomial time with the ellipsoid method as shown by Kozlov et al.~\cite{kozlov1980polynomial}.

\section{Discontinuous Costs and Directed Networks} \label{sec:discontinuous}
In this section, we consider single commodity networks with discontinuous functions. In particular, we consider strictly increasing cost functions $l_e(x_e)$ that are non-negative and lower semi-continuous for $x_e > 0$, and non-positive and upper semi-continuous for $x_e < 0$. As it turns out, this will also enable us to treat directed networks as a special case of this setting.
Since Wardrop equilibria may fail to exist when cost functions are discontinuous, we adopt
\longversiononly{%
the User Equilibrium concept of Bernstein and Smith \cite{bernstein1994} which is a generalization of the equilibrium concept introduced by Dafermos in \cite{dafermos1971}. We say a%
}{%
the more general concept if a User equilibrium \cite{bernstein1994,dafermos1971}. A%
}
flow vector $\vec x$ is an \emph{User equilibrium} if $l_P(\vec x) \leq \liminf_{\epsilon \downarrow 0} l_Q(\vec x + \epsilon \, \vec \chi_Q - \epsilon \, \vec \chi_P)$
for all paths $P,Q \in \mathcal{P}$ with $P$ flow carrying.\longversiononly{The vector $\vec \chi_P$ is the indicator vector of the path $P \in \mathcal{P}$ as defined in \S~\ref{sec:prelim}.} User equilibria and Wardrop equilibria coincide in the case of continuous cost functions.
We proceed to give a potential based characterization of User equilibria similar to Lemma~\ref{lem:characterisationWEPotentials}.
\longversion{
\begin{lemma} \label{lem:UECharacterization}
Let $\vec x$ be an undirected, single commodity flow in the graph $G$ with semi-continuous cost functions. Then the following two statements are equivalent:
\begin{compactenum}[(i)]
\item
	The flow $\vec x$ is an User equilibrium flow.
\item
	There exists a potential vector $\vec \pi \in \R^n$ such that
	\begin{align}
	\smash{\lim_{x \uparrow x_e} l_e (x) \leq \vec{\gamma}_e^{\top} \, \vec{\pi} \leq \lim_{x \downarrow x_e} l_e (x) \text{ for all edges $e = (v,w) \in E$}.}
	\label{eq:lem:UECharacterization:1}
	\end{align}
\end{compactenum}
\end{lemma}
}{
\begin{lemma} \label{lem:UECharacterization}
An undirected, single commodity flow $\vec x$ in a graph $G$ with semi-continuous cost functions is a User equilibrium flow if and only if there exists a potential vector $\vec \pi \in \R^n$ such that
	\begin{align}
	\smash{\lim_{x \uparrow x_e} l_e (x) \leq \vec{\gamma}_e^{\top} \, \vec{\pi} \leq \lim_{x \downarrow x_e} l_e (x) \text{ for all edges $e = (v,w) \in E$}.}
	\label{eq:lem:UECharacterization:1}
	\end{align}
\end{lemma}
}
\begin{appendixproof}{Lemma~\ref{lem:UECharacterization}}
First note that since all functions $l_e(x)$ are increasing the upper/lower semi-continuity for negative/positive $x$ implies that
\begin{align*}
l_e (x_e) = \lim_{x \uparrow x_e} 
	\quad \text{for } x_e > 0 \quad \text{and} \quad
l_e (x_e) = \lim_{x \downarrow x_e}
	\quad \text{for } x_e < 0 .
\end{align*}

For the direction $(i) \Rightarrow (ii)$ we need to construct a potential vector $\vec \pi$. To this end, construct a graph $\tilde{G} = (V, \tilde{E})$ where the edge set $\tilde{E} := \tilde{E}^F \cup \tilde{E}^B$ consists of forward edges $\tilde{E}^F$ and backward edges $\tilde{E}^B$ which are defined as follows.
\begin{align*}
\tilde{E}^F &:= \big\{ (v,w) \; | \; e = (v,w) \in E \text{ and } x_e \geq 0 \big\} \\
\tilde{E}^B &:= \big\{ (w,v) \; | \; e = (v,w) \in E \text{ and } x_e \leq 0 \big\} 
\end{align*}
Hence, every edge with non-negative flow $x_e$ is represented by a forward edge in $\tilde{G}$ and every edge with non-positive flow $x_e$ is represented by a backward edge. Note that edges with $x_e = 0$ are represented by both a forward and a backward edge.
We then assign the weights
\[
w_e :=
\begin{cases}
\displaystyle \lim_{x \downarrow x_e} l_e (x) & \text{if } e \in \tilde{E}^F, \\
\displaystyle - \lim_{x \uparrow x_e} l_e (x) & \text{if } e \in \tilde{E}^B
\end{cases}
\]
to the edges in $\tilde{E}$. The construction ensures that all edge weights in $\tilde{G}$ are non-negative. Thus, there exists a shortest path potential $\vec \pi$ with $\pi_v :=$ length of a shortest path from $s$ to $v$ in $\tilde{G}$. This potential satisfies
\begin{align*}
\vec{\gamma}_e^{\top} \vec{\pi} &\leq \lim_{x \downarrow x_e} l_e (x_e)
	&\text{for all } e \in E \text{ with } x_e > 0, \\
\lim_{x \uparrow x_e} l_e (x_e) \leq \vec{\gamma}_e^{\top} \vec{\pi} &
	&\text{for all } e \in E \text{ with } x_e < 0, \\
\lim_{x \uparrow x_e} l_e (x_e) \leq \vec{\gamma}_e^{\top} \vec{\pi} &\leq \lim_{x \downarrow x_e} l_e (x_e)
	&\text{for all } e \in E \text{ with } x_e = 0
\end{align*}
by construction.
It remains to be shown that $\vec{\gamma}_e^{\top} \vec \pi \geq \lim_{x \uparrow x_e} l_e (x) = l_e (x_e)$ for edges with $x_e > 0$ and $\vec{\gamma}^{\top} \vec \pi \leq \lim_{x \downarrow x_e} l_e (x) = l_e( x_e )$ for edges with $x_e < 0$. We prove this for $x_e > 0$, the other statement follows in a similar way. \\
Let $e \in E$ be an edge with $x_e > 0$. If $\vec{\gamma}_e^{\top} \vec \pi = \lim_{x \downarrow x_e} l_e (x_e)$ then $\vec{\gamma}_e^{\top} \vec \pi \geq \lim_{x \uparrow x_e} l_e (x)$ follows immediately. Otherwise, there is another shortest path $\tilde{Q}$ in $\tilde{G}$ from $v$ to $w$ such that
\begin{align*}
\pi_w - \pi_v &= \sum_{e \in \tilde{Q} \cap \tilde{E}^F} \lim_{x \downarrow x_e} l_e (x) - \sum_{e \in \tilde{Q} \cap \tilde{E}^B} \lim_{x \uparrow x_e} l_e (x) \\
&= \sum_{e \in \tilde{Q}} \lim_{\epsilon \downarrow 0} \chi_{\tilde{Q}, e} l_e (x_e + \epsilon \, \chi_{\tilde{Q}, e})
= \liminf_{\epsilon \downarrow 0} \sum_{e \in \tilde{Q}} \chi_{\tilde{Q}, e} l_e (x_e + \epsilon \, \chi_{\tilde{Q}, e})
.
\end{align*}
Since $x_e > 0$ the edge $e \in E$ is contained in some flow-carrying path $P = P_1 + e + P_2$. We know since $\vec x$ is an user equilibrium that $l_P(\vec x) \leq \liminf_{\epsilon \downarrow 0} l_Q(\vec x + \epsilon \, \vec \chi_Q - \epsilon \, \vec \chi_P)$ any path $Q$. If we consider $Q := P_1 + \tilde{Q} + P_2$ and subtract $\sum_{e \in P \cap Q} l_e (x_e)$ on both sides, we obtain
\begin{align*}
l_e (x_e) \leq \liminf_{\epsilon \downarrow 0} \sum_{e \in \tilde{Q}} \chi_{\tilde{Q}, e} l_e (x_e + \epsilon \, \chi_{\tilde{Q}, e}) = \pi_w - \pi_v.
\end{align*}
Thus the claim follows for edges with $x_e > 0$. For edges with $x_e < 0$ this can be proven similarly.

For the direction $(ii) \Rightarrow (i)$ let $P$ be some flow carrying path. By definition of a flow carrying path all backward edges carry negative flow and all forward edges carry positive flow. Since every cost function $l_e$ is upper/lower semi-continuous for negative/positive flow, the inequalities \eqref{eq:lem:UECharacterization:1} yield
\begin{align*}
l_e(x_e) = \liminf_{x \to x_e} l_e(x_e) &\leq \vec{\gamma}_e^{\top} \vec \pi
\quad \text{for } e \in P \text{ with } x_e > 0 \text{ and } \\
l_e(x_e) = \limsup_{x \to x_e} l_e(x_e) &\geq \vec{\gamma}_e^{\top} \vec \pi
\quad \text{for } e \in P \text{ with } x_e < 0.
\end{align*}
Let $Q \in \mathcal{P}$ be an arbitrary path. Then $P$ can be decomposed in subpath $P_1, \dotsc, P_{k}$ that connect the vertices $v_i, w_i$ such that either $P_i \subset P \cap Q$ or $P_i \subset P \setminus Q$. Likewise, the path $Q$ can be decomposed in subpaths $Q_1, \dotsc, Q_{k}$ with the same vertices $v_i, w_i$ such that either $Q_i \subset P \cap Q$ or $Q_i \subset P \setminus Q$. We than compute
\begin{align*}
l_P(\vec{x}) &= \sum_{e \in P} l_e (x_e) \, \chi_{P,e} \\
	&= \sum_{e \in P \cap Q}  l_e (x_e) \, \chi_{P,e} + \sum_{\substack{e \in P \setminus Q:\\ \chi_{P,e}>0}} l_e(x_e) - \sum_{\substack{e \in P \setminus Q:\\ \chi_{P,e}<0}} l_e(x_e) \\
	&\leq \sum_{e \in P \cap Q}  l_e (x_e) \, \chi_{P,e} + \sum_{e = (v,w) \in P \setminus Q} \pi_w - \pi_v \\
	&= \sum_{e \in P \cap Q}  l_e (x_e) \, \chi_{P,e} + \sum_{i = 1}^k \pi_{w_i} - \pi_{v_i} \\
	&= \sum_{e \in P \cap Q}  l_e (x_e) \, \chi_{P,e} + \sum_{e = (v,w) \in Q \setminus P} \pi_w - \pi_v \\
	&\leq \sum_{e \in P \cap Q}  l_e (x_e) \, \chi_{P,e} + \sum_{\substack{e \in Q \setminus P:\\ \chi_{Q,e} > 0}} \liminf_{x \to x_e} l_e(x) - \sum_{\substack{e \in Q \setminus P:\\ \chi_{Q,e}<0}} \limsup_{x \to x_e} l_e(x) \\
	&= \sum_{e \in P \cap Q}  l_e (x_e) \, \chi_{P,e} + \sum_{\substack{e \in Q \setminus P:\\ \chi_{Q,e} > 0}} \liminf_{\epsilon \downarrow 0} l_e(x_e + \epsilon) - \sum_{\substack{e \in Q \setminus P:\\ \chi_{Q,e}<0}} \limsup_{\epsilon \downarrow 0} l_e(x_e - \epsilon) \\
	&= \liminf_{\epsilon \downarrow 0} l_Q(\vec x + \epsilon \, \vec \chi_Q - \epsilon \, \vec \chi_P)
\end{align*}
for every path $Q \in \mathcal{P}$.
\end{appendixproof}
This characterization enables us to formulate equilibrium flows in directed networks as a special case of the undirected case.
\begin{corollary} \label{cor:directedWardrop}
Let $l_e: \R_{\geq 0} \to \R_{\geq 0}$ be continuous, non-decreasing cost functions on the edges. Then a directed flow $\vec x$ is a Wardrop equilibrium for the cost functions $l_e (x)$ if and only if the
undirected flow $\vec x$ is a User equilibrium in a graph with cost functions $\hat{l}_e(x) = l_e(x)$ if $x \geq 0$ and $\hat{l}_e(x) = -\infty$ otherwise.
\end{corollary}
\begin{appendixproof}{Corollary~\ref{cor:directedWardrop}}
For the special cost functions $\hat{l}$, the characterization \eqref{eq:lem:UECharacterization:1} simplifies to
\begin{align*}
\vec \gamma_e^{\top} \vec \pi &= l_e (x_e) &\text{for } e \in E \text{ with } x_e > 0, \\
-\infty \leq \vec \gamma_e^{\top} \vec \pi &\leq l_e(0) &\text{for } e \in E \text{ with } x_e = 0, \\
\vec \gamma_e^{\top} \vec \pi &= - \infty &\text{for } e \in E \text{ with } x_e < 0.
\end{align*}
Every directed Wardrop equilibrium flow obviously satisfies these conditions since the first two conditions are implied by Lemma~\ref{lem:characterisationWEPotentials} and there is no edge with negative flow. Conversely, for every user equilibrium there has to be a real-valued potential $\vec \pi$ such that these conditions are fulfilled. This means there can not be any edge $e \in E$ with negative flow $x_e < 0$. Thus this flow is also a Wardrop equilibrium.
\end{appendixproof}

We proceed to adjust the base algorithm developed in \S~\ref{sec:basic} so that it can handle discontinuous cost functions. To this end, we define for every breakpoint $\tau_{e,t_e}$ of the piecewise linear cost functions $l_e$ the values
$\sigma_{e,t}^+ := \lim_{x \downarrow \tau_{e,t_e}} l_e (x)$  and 
$\sigma_{e,t}^- := \lim_{x \uparrow \tau_{e,t_e}} l_e (x)$ (see Figure~\ref{fig:discontinuousCost} \shortversiononly{ in Appendix~\ref{appendixfigure\thesection subsection}}).
With these values, we define the inverse function of $l_e$ as the function
\begin{align*}
l_e^{-1} (v_e) = 
\begin{cases}
c_{e,t}\, v_e - d_{e,t}
	& \text{ for all } v_e \in (\sigma^+_{e,t}, \sigma^-_{e,t+1}) \\
\tau_{e,t} 
	& \text{ for all } v_e \in [\sigma^-_{e,t}, \sigma^+_{e,t}].
\end{cases}
\end{align*}

\begin{appendixfigure}{t}
\centering
\begin{subfigure}[b]{.35\linewidth}
\centering
\begin{tikzpicture}
\newcommand{\xMin}{-1} \newcommand{\xMax}{3}
\newcommand{\yMin}{-0.5} \newcommand{\yMax}{3}

\draw[thick, ->] (-0.5, \yMin) -- (-0.5, \yMax) node[anchor=south] {$l_e(x_e)$};
\draw[thick, ->] (\xMin, 0) -- (\xMax, 0) node[anchor=west] {$x_e$};

\draw[thick] (0,.1) -- (0, -.1) node[below] {$\tau_{e,1}$};
\draw[thick] (1,.1) -- (1, -.1) node[below] {$\tau_{e,2}$};

\draw[thick] (-0.5+.1,0.5) -- (-0.5-.1,0.5) node[left] {$\sigma^+_{e,1}$};
\draw[thick] (-0.5+.1,1.25) -- (-0.5-.1,1.25) node[left] {$\sigma^-_{e,2}$};
\draw[thick] (-0.5+.1,2) -- (-0.5-.1,2) node[left] {$\sigma^+_{e,2}$};

\draw[color4!70!black, thick]			(0, 0.5) -- (1,1.25)
										(1,2) -- (2.75,3);
\draw[color4!70!black, thick, dotted]	(0, \yMin) -- (0, 0.5) 
										(1,1.25) -- (1,2);

\end{tikzpicture}
\caption{Cost function $l_e(x_e)$.}
\end{subfigure}
\hspace{0.05\linewidth}
\begin{subfigure}[b]{.35\linewidth}
\centering
\begin{tikzpicture}
\newcommand{\xMin}{-0.5} \newcommand{\xMax}{3}
\newcommand{\yMin}{-1} \newcommand{\yMax}{2.5}

\draw[thick, ->] (0, \yMin) -- (0, \yMax) node[anchor=south] {$l^{-1}_e(v_e)$};
\draw[thick, ->] (\xMin, -0.5) -- (\xMax, -0.5) node[anchor=west] {$v_e$};

\draw[color4!70!black, thick]	(\xMin, 0) -- (0.5,0) -- (1.25,1) -- (2, 1) -- (3, 2.75);

\draw[thick] (0.5,.1 - 0.5) -- (0.5, -.1 - 0.5) node[below] {$\sigma^+_{e,1}$};
\draw[thick] (1.25,.1 - 0.5) -- (1.25, -.1 - 0.5) node[below] {$\sigma^-_{e,2}$};
\draw[thick] (2,.1 - 0.5) -- (2, -.1 - 0.5) node[below] {$\sigma^+_{e,2}$};

\draw[thick] (.1,0) -- (-.1, 0) node[left] {$\tau_{e,1}$};
\draw[thick] (.1,1) -- (-.1,1) node[left] {$\tau_{e,2}$};

\end{tikzpicture}
\caption{Inverse cost function $l^{-1}_e(v_e)$.}
\end{subfigure}
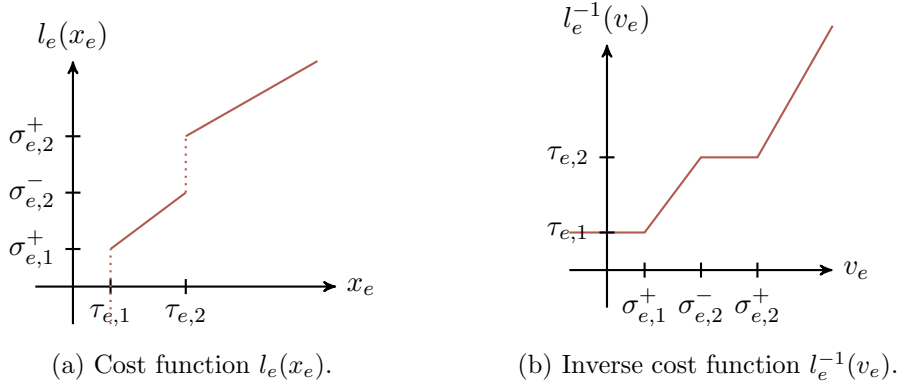
\caption{A discontinuous cost function $l_e$ with breakpoints $\tau_{e,1}$ and $\tau_{e,2}$ (a) and its inverse function (b). The limit $\sigma_{e,1}^- := \lim_{x \uparrow \tau_{e,1}} l_e (x)$ is assumed to be $- \infty$, modeling a directed edge. The inverse cost function is constant on $(\sigma_{e,1}^-, \sigma_{e,1}^+)$ and $(\sigma_{e,2}^-, \sigma_{e,2}^+)$.}
\label{fig:discontinuousCost}
\end{appendixfigure}

This definition ensures that for every potential difference $\vec{\gamma}_e^{\top} \vec \pi$ the condition \eqref{eq:lem:UECharacterization:1} is satisfied for the flow $x_e := l_e^{-1} ( \vec{\gamma}_e^{\top} \vec \pi )$.
For convenience, we rename all the breakpoints to $\sigma_{e,t}$ such that the inverse cost functions have again the form
$l_e^{-1} (v_e) = 
c_{e,t}\, v_e - d_{e,t}$ for all $v_e \in (\sigma_{e,t}, \sigma_{e,t+1})$
where the $\sigma_{e,t}$ are the $\sigma^+_{e,t}$ and $\sigma^-_{e,t}$ from before, renamed and reordered. The major difference to the continuous case from before is that some of the slopes $c_{e,t}$ may now be equal to zero if they correspond to a slope in a $(\sigma^-_{e,t}, \sigma^+_{e,t+1})$ interval.

Even in the presence of discontinuities, given some potential $\vec{\pi}$, we still can compute the flow $\vec x = \vec C_{\vec t} \vec \incidenceMatrix^{\top} \vec \pi - \vec d_{\vec t}$ satisfying \eqref{eq:lem:UECharacterization:1} 
in a region $R_{\vec t}$. Further, the excess induced by the potential is 
$
\smash{\vec y = \vec L_{\vec t} \vec \pi - \tilde{\vec d}_{\vec t}}
\, .
$
If we are in a region $R_{\vec t}$ such that all vertices are connected by edges with $c_{e,t_e} > 0$, the reduced Laplacian matrix $\smash{\hat{\vec L}_{\vec t}}$ is still invertible by Lemma~\ref{lem:LaplacianMatrix}(\ref{lem:LaplacianMatrix:subMatrix}) and the direction vector can be computed as
$
\smash{\Delta \hat{\vec \pi} = \hat{\vec L}_{\vec t}^{-1} \Delta \hat{\vec y}}
\, .
$
In contrast to the last section, we now need to deal with regions $R_{\vec t}$ where the reduced Laplacian $\smash{\hat{ \vec L}^{-1}_{\vec t}}$ is not invertible. We call such regions \emph{ambiguous} since the direction vector $\Delta \vec \pi$ can not be determined uniquely in these regions.

To formalize this, we call an edge $e$ active in a region $R_{\vec t}$ if $c_{e, t_e} > 0$ and inactive in $R_{\vec t}$ if $c_{e, t_e} = 0$. Two vertices $v$ and $w$ are called actively connected in $R_{\vec t}$ if there is an undirected path between $v$ and $w$ containing only active edges in $R_{\vec t}$. An actively connected component $\mathcal{C}$ is a maximal subset of vertices such that all vertices $v,w \in \mathcal{C}$ with $v \neq w$ are actively connected.

By Lemma~\ref{lem:LaplacianMatrix}, $\smash{\hat{\vec L}_{\vec t}^{-1}}$ is invertible if and only if there is exactly one actively connected component. Conversely, a region $R_{\vec t}$ is ambiguous if and only if there are more then one actively connected components in $R_{\vec t}$. In fact, we can use these connected components to find a direction vector $\Delta \vec \pi$ in the ambiguous region that can be added to the current potential without changing the flow. Following this direction to the next boundary eventually leads to an non-ambiguous region\longversiononly{ (see Figure~\ref{fig:ambiguousRegion})}. %
%
\shortversiononly{%
\begin{collect}{proofcollection\thesection}{}{}
}%
\longversion{%
To derive this direction formally, we first of all observe the following property of subsets $U \subseteq V$ of vertices.%
}{
The proofs of Lemma~\ref{lem:enteringAmbiguousRegions} and Theorem~\ref{thm:ambiguousRegion} require the following technical lemma.
}
\begin{lemma} \label{lem:potentialSumOfComponent}
For some non-ambiguous region $R_{\vec t}$ let $\Delta \vec \pi$ solve $\Delta \vec y = \vec L_{\vec t} \Delta \vec \pi$. Then for any $U \subseteq V$
\begin{align*}
\sum_{v \in U} \Delta y_v 
= \sum_{e \in \delta^+(U)} c_{e,t_e} \vec \gamma_e^{\top} \Delta \vec \pi - \sum_{e \in \delta^-(U)} c_{e,t_e} \vec \gamma_e^{\top} \Delta \vec \pi
.
\end{align*}
In particular, we have
$
\Delta \pi_w - \Delta \pi_v = \pm \frac{\sum_{v \in U} \Delta y_v}{c_{e,t_e}}
$
if there is exactly one active edge $e = (v,w)$ entering/leaving $U$ with positive/negative sign if $e$ is entering/leaving $U$.
\end{lemma}
\begin{proof}
For every $v \in V$ we have 
\[
\Delta y_v = (\vec L_{\vec t} \Delta \vec \pi)_v = \sum_{e \in \delta^+ (v)} c_{e,t_e} \vec \gamma_e^{\top} \Delta \vec \pi - \sum_{e \in \delta^- (v)} c_{e,t_e} \vec \gamma_e^{\top} \Delta \vec \pi
\]
and thus the claim follows by taking the sum over all vertices $v \in U$.
\end{proof}
\shortversiononly{
\end{collect}
}

The next lemma will show that if the solution curve leaves a non ambiguous region then either the adjacent region is also non ambiguous or there is a direction $\Delta \vec{\pi}^i$ the solution curve can follow without changing the induced flow.

\begin{lemma} \label{lem:enteringAmbiguousRegions}
Assume the solution curve is in a region $R_{\vec t^{i-1}}$ where $\hat{\vec L}_{\vec t^{i-1}}$ is invertible. If the solution curve crosses a unique boundary to some region $R_{\vec t^{i+1}}$ in the point $\vec{\pi}^i$, then exactly one of the following statements is true.
\begin{compactenum}[(i)]
\item
	The matrix $\hat{\vec L}_{\vec t^{i}}$ is invertible, i.e. the region $R_{\vec{t}^i}$ is non-ambiguous.
\item
	There are exactly two actively connected components $\mathcal{C}_1, \mathcal{C}_2$ in the region $R_{\vec t^{i}}$ with $\mathcal{C}_1 \cup \mathcal{C}_2 = V$, $s \in \mathcal{C}_1$, and $t \in \mathcal{C}_2$, such that
	\begin{flexequation*}
	\vec{x} = \vec{C}_{\vec{t}^i} \vec{\Gamma}^{\top} ( \vec{\pi}^i + \lambda \, \Delta \vec{\pi}^i ) - \vec{d}_{\vec{t}^i}
	\text{ for all } \lambda \in \R
	\end{flexequation*}
	where $\Delta \vec{\pi}^i$ is the direction with $\Delta \pi^i_v = 1$ if $v \in \mathcal{C}_2$ and $\Delta \pi^i_v = 0$ if $v \in \mathcal{C}_1$.
\end{compactenum}
\end{lemma}
\begin{appendixproof}{Lemma~\ref{lem:enteringAmbiguousRegions}}
If $(i)$ is true, then $(ii)$ is obviously not true since $\hat{\vec L}_{\vec t^{i}}$ is singular if there is more than one actively connected component.

So assume $\hat{\vec L}_{\vec t^{i}}$ is not invertible. This implies that there are at least two actively connected components in $R_{\vec t^{i}}$. Since there is only one boundary between $R_{\vec t^{i-1}}$ and $R_{\vec t^{i}}$ at most one $c_e$ can change and thus there can be at most one more inactive edge in $R_{\vec t^{i}}$ than in $R_{\vec t^{i-1}}$. Making one $e$ inactive can only cause one actively connected component to split into two components. So there are also at most two and thus exactly two actively connected components $\mathcal{C}_1, \mathcal{C}_2$ in $R_{\vec t^{i}}$.

Now we show that $s$ and $t$ are in two different connected components. Assume towards contradiction that $s,t \in \mathcal{C}_1$. By assumption, a unique bridging edge $\hat{e}=(\hat{v}, \hat{w})$ between $\mathcal{C}_1$ and $\mathcal{C}_2$ was active in $R_{\vec t^{i-1}}$ but is now inactive in $R_{\vec t^{i}}$. Thus, $c_{\hat{e}}$ changed which is only the case if the boundary between $R_{\vec t^{i-1}}$ and $R_{\vec t^{i}}$ was induced by the edge $\hat{e}$. Hence, $\Delta \vec{\pi}^i$ must have pointed towards the boundary meaning that 
\begin{equation*}
0 \neq \vec \gamma_{\hat{e}}^{\top} \Delta \vec{\pi}^{i-1} = \Delta \pi_{w}^{i-1} - \Delta \pi_v^{i-1}
.
\end{equation*}
Since $s,t \notin \mathcal{C}_2$  we have $\Delta y_v = 0$ for all $v \in \mathcal{C}_2$ and thus Lemma~\ref{lem:potentialSumOfComponent} yields
\begin{align*}
\Delta \pi_w^{i-1} - \Delta \pi_v^{i-1} = \pm \frac{\sum_{v \in \mathcal{C}_2} \Delta y_v}{c_{e,t_e}} = 0
\end{align*}
which is a contradiction.

Finally, we prove that the flow is indeed constant along the above defined direction $\Delta \vec{\pi}^i$.
For all edges $e$ in $\mathcal{C}_1$ and $\mathcal{C}_2$ we have $\vec \gamma^{\top}_e \Delta \vec \pi = 0$ by the definition of $\Delta \vec \pi$. For all edges between the components we have $c_{e, t_e} = 0$ because the components are not actively connected. Thus $c_{e, t_e} \vec \gamma_e^{\top} \Delta \vec \pi = 0$ for all edges $e \in E$ which implies
$
\vec C_{\vec t} \vec \Gamma^{\top} \Delta \vec \pi = \vec 0
$
and thus proves the statement.
\end{appendixproof}

\begin{appendixfigure}{t}
\centering
\begin{tikzpicture}
\newcommand{\xMin}{-0.5} \newcommand{\xMax}{4}
\newcommand{\yMin}{0} \newcommand{\yMax}{3}
\path[fill=color1,fill opacity=0.25] (\xMin, \yMin) -- (0.5, \yMin) -- (0.5,0.5) -- (\xMin, 0.5);
\path[fill=color2,fill opacity=0.25] (0.5, \yMin) -- (0.5,0.5) -- (3, 0.5) -- (3, \yMin);
\path[fill=color3,fill opacity=0.25] (3, 0.5) -- (3, \yMin) -- (\xMax, \yMin) -- (\xMax, 0.5);
\path[fill=color5,fill opacity=0.25] (\xMin, 0.5) -- (0.5, 0.5) -- (0.5,2.5) -- (\xMin, 2.5);
\path[fill=color4,fill opacity=0.25] (0.5, 0.5) -- (0.5,2.5) -- (3, 2.5) -- (3,0.5);
\path[fill=color6,fill opacity=0.25] (3, 2.5) -- (3,0.5) -- (\xMax, 0.5) -- (\xMax, 2.5);
\path[fill=color7,fill opacity=0.25] (\xMin, 2.5) -- (0.5, 2.5) -- (0.5, \yMax) -- (\xMin, \yMax);
\path[fill=color8,fill opacity=0.25] (0.5, 2.5) -- (0.5, \yMax) -- (3, \yMax) -- (3,2.5);
\path[fill=color9,fill opacity=0.25] (3, \yMax) -- (3,2.5) -- (\xMax, 2.5) -- (\xMax, \yMax);

\foreach \i in {1,2,3,5,6,7} {
	\draw[decoration={zigzag,segment length=3,amplitude=.75}, decorate, color4!85!black, thin] (.5,0.5 + 1/4*\i) -- (3,0.5 + 1/4*\i);
}

\node[color=color4!30!black, anchor=north east] at (3,2.5) {$R_{\vec{t}^i}$};

\draw[thin, dashed] (\xMin, .5) -- (\xMax, .5);
\draw[thin, dashed] (\xMin, 2.5) -- (\xMax, 2.5);

\draw[thin, dashed] (.5, \yMin) -- (.5, \yMax);
\draw[thin, dashed] (3, \yMin) -- (3, \yMax);
\draw[thick] (\xMin, 1) -- (.5,1.5);
\draw[thick] (3,1.5) -- ({min(\xMax, 4*\yMax-3)}, {min(1/4*\xMax+3/4, \yMax)});
\draw[thick, decoration={zigzag,segment length=4,amplitude=.9}, decorate, color4!35!black] (.5,1.5) -- (3,1.5);

\fill[black] (0.5,1.5) circle(2pt) node[anchor=south east] {$\vec{\pi}^i$};
\fill[black] (3,1.5) circle(2pt) node[anchor=north west] {$\vec{\pi}^{i+1}$};

\draw[thick, color4!40!black] (1.25, 1.25) edge[->] node[midway, below] {$\Delta \vec{\pi}^i$} (2.25, 1.25);
\end{tikzpicture}
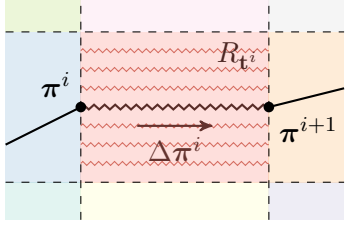
\caption{A degenerate region $R_{\vec{t}^i}$ in the potential space. Along the direction $\Delta \vec{\pi}^i$ the flow is unchanged, the zigzag lines indicate potentials that induce the same flow. In particular, $\vec{\pi}^i$ and $\vec{\pi}^{i+1}$ induce the same flow -- the solution curve can jump from $\vec{\pi}^i$ to $\vec{\pi}^{i+1}$ without changing the flow and skip the ambiguous region.}
\label{fig:ambiguousRegion}
\end{appendixfigure}

Finally, the concluding theorem in this section will show that following the direction established in Lemma~\ref{lem:enteringAmbiguousRegions} will lead to a neighboring, non-ambiguous region. The solution curve can move instantaneously to the neighboring region since moving along this direction does not change the induced flow.

\begin{theorem} \label{thm:ambiguousRegion}
Assume the solution curve enters an ambiguous region $R_{\vec t^i}$ in the point $\vec \pi^i \in R_{\vec t^i}$. Then there is a neighboring, non-ambiguous region $R_{\vec t^{i+1}}$ and a potential vector $\vec \pi^{i+1} \in R_{\vec t^{i+1}}$ that induces the same flow, i.e.
$\vec L_{\vec t^i} \vec \pi^i - \vec d_{\vec t^i} = \vec L_{\vec t^{i+1}} \vec \pi^{i+1} - \vec d_{\vec t^{i+1}}$.
Furthermore, the direction vector $\Delta \vec \pi^{i+1}$ in region $R_{\vec t^{i+1}}$ is directed away from $R_{\vec t^i}$, i.e. the solution curve will not enter $R_{\vec t^i}$ again when starting in $\vec \pi^{i+1}$.
\end{theorem}
\begin{appendixproof}{Theorem~\ref{thm:ambiguousRegion}}
Lemma~\ref{lem:enteringAmbiguousRegions} implies that there is the direction $\Delta \vec{\pi}$ along which the induced flow is constant. Let $\vec{\pi}^{i+1} := \vec{\pi}^i + \epsilon \, \Delta \vec{\pi}^i$ where $\epsilon$ is chosen such that $\vec{\pi}^{i+1}$ lies on the boundary to the next region $R_{\vec{t}^{i+1}}$. Let $\hat{e} = (\hat{v}, \hat{w})$ be the edge inducing that boundary.
Then $\hat e$ is in particular a bridging edge between $\mathcal{C}_1$ and $\mathcal{C}_2$ because $\vec \gamma^{\top} \Delta \vec \pi^i = 0$ for all non bridging edges. Assume without loss of generality that $\hat w \in \mathcal{C}_2$ (otherwise change the direction of $\hat e$). Then $\vec \gamma_{\hat e}^{\top} \Delta \vec \pi^i = \Delta \pi^i_w - \Delta  \pi^i_v = 1 - 0 > 0$. Since $\hat{e}$ is the unique, active bridging edge between $\mathcal{C}_1$ and $\mathcal{C}_2$ (recall that there was no active bridging edge in the region $R_{\vec{t}^i}$) we have by Lemma~\ref{lem:potentialSumOfComponent} that 
\begin{align*}
\vec \gamma_{\hat e}^{\top} \Delta \vec \pi^{i+1} = \Delta \pi^{i+1}_w - \Delta  \pi^{i+1}_v = \frac{\sum_{v \in \mathcal{C}_2} \Delta y_v}{c_{\hat{e}, t^{i+1}_{\hat{e}}}} = \frac{1}{c_{\hat{e}, t^{i+1}_{\hat{e}}}} > 0
.
\end{align*}
So $\Delta \vec \pi^{i}$ and $\Delta \vec \pi^{i+1}$ have the same direction with respect to the normal vector $\vec \gamma_{\hat{e}}$ and, thus, $\Delta \vec \pi^{i+1}$ is directed away from the boundary. So overall, the vector $\vec \pi^{i+1}$ lies in the neighboring region $R_{\vec t^{i+1}}$ and induces the same flow as $\vec \pi^i$ since it lies also in $R_{\vec t^i}$ and on the line $\vec \pi^i + \lambda \Delta \vec \pi^i$. Furthermore, the direction vector $\Delta \vec \pi^{i+1}$ is directed away from the boundary between $R_{\vec t^i}$ and $R_{\vec t^{i+1}}$.
\end{appendixproof}

So with the last theorem we have a way to circumvent ambiguous regions. If the solution curve enters an ambiguous region in $\vec \pi^i$ then Theorem~\ref{thm:ambiguousRegion} gives us a point $\vec \pi^{i+1}$ that induces the same flow from where the algorithm can proceed (see Figure~\ref{fig:ambiguousRegion}). The handling of discontinuous costs is illustrated in detail in Example~\ref{ex:ambiguousRegion} in Appendix~\ref{app:examples}.
\longversiononly{
	\section{Constant Costs} \label{sec:constantCost}
In this section we want to deal with cost function that are (piecewise) constant. This means, there are regions $R_{\vec{t}}$ where the slope of some cost function $a_{e, t_e}$ is zero. Regions with constant costs are hard to handle since the inverse cost functions $l_e^{-1}$ are not well-defined in this case. In particular, the Laplacian matrix $\vec{L}_{\vec{t}}$ is not well defined. Below, we will prove that the inverse of the Laplacian $\vec{L}_{\vec{t}}^{-1}$ still exists even in the presence of constant costs. To this end, we need to study the behavior of the inverse of the Laplacian as some slope $a_{e, t_e}$ tends to zero. 

\subparagraph{Degenerate regions}
In a region $R_{\vec{t}}$ where some cost function $l_e (x)$ of some edge $e = (v,w)$ is constant, a Wardrop equilibrium flow has to satisfy
\[
\pi_w - \pi_v = l_e(x_e) = \underbrace{a_{e,t_e}}_{=0} \, x_e + b_{e, t_e} = b_{e, t_e}
,
\]
i.e. the potential difference $\pi_w - \pi_v$  is constant in the whole region. This means that such regions are lower dimensional hyperplanes in the potential space. In the example described in Figure~\ref{fig:degenerateRegion} the region $R_2$ is a line in the two dimensional potential space (see Figure~\ref{fig:degenerateRegion}(d)) in the case of a constant cost function ($\alpha = \infty$). If the solution curve is in such a region it is permanently moving along the boundaries induced by the edge with constant cost. Therefore, we call these regions \emph{degenerate} (or more precisely $e$-\emph{degenerate} where $e$ is the edge with constant cost).
\begin{figure}[t]
\centering
\begin{subfigure}[b]{.45\linewidth}
\begin{center}
\begin{tikzpicture}
\draw (0,0) node[solid] (s) {} node [below left] {$s$};
\draw (2,1) node[solid] (v) {} node [above] {$v$};
\draw (4,0) node[solid] (t) {} node [below right] {$t$};

\draw[->]	(s) edge node[midway, above left] {$e_1$} (v) 
			(v) edge node[midway, above right] {$e_3$} (t) 
			(s) edge node[midway, below] {$e_2$} (t);
\end{tikzpicture}
\end{center}
\caption{Graph.}
\end{subfigure}
\hspace{0.02\linewidth}
\begin{subfigure}[b]{.45\linewidth}
\begin{center}
\begin{align*}
l_{e_1} (x) &= l_{e_2} (x) =  x , \\
l_{e_3} (x) &=
\begin{cases}
x & \text{if } x \leq 1, \\
\frac{1}{\alpha} \, x + \frac{\alpha - 1}{\alpha} & \text{if } 1 < x \leq 3, \\
x + \frac{2 - 2 \alpha}{\alpha} & \text{if } x > 3.
\end{cases} 
\end{align*}
\end{center}
\caption{Cost functions.}
\end{subfigure}
\\
\begin{subfigure}[t]{.45\linewidth}
\begin{center}
\begin{tikzpicture}
\newcommand{\xMin}{-0.5} \newcommand{\xMax}{3.5}
\newcommand{\yMin}{-0.5} \newcommand{\yMax}{3.5}
\newcommand{\xScale}{0.5} \newcommand{\yScale}{0.5}

\path[fill=color1,fill opacity=0.25] ({max(\xMin,\yMin)},{max(\xMin+1/2,\yMin)}) -- ({max(\xMin,\yMin)},\yMin)--(\xMax,\yMin)--(\xMax,{min(\xMax+1/2,\yMax)})--({min(\xMax,\yMax-1/2)},{min(\xMax+1/2,\yMax)});
\path[fill=color2,fill opacity=0.25] ({max(\xMin,\yMin)},{max(\xMin+1/2,\yMin)}) -- ({max(\xMin,\yMin)},{max(\xMin+1,\yMin)}) -- ({min(\xMax,\yMax-1)},{min(\xMax+1,\yMax)}) -- ({min(\xMax,\yMax-1/2)},{min(\xMax+1/2,\yMax)}); ;
\path[fill=color3,fill opacity=0.25] ({max(\xMin,\yMin)},{max(\xMin+1,\yMin)}) -- ({max(\xMin,\yMin)},\yMax)--({min(\xMax,\yMax-1)},\yMax)--({min(\xMax,\yMax-1)},{min(\xMax+1,\yMax)});
\node[color=color1!30!black, anchor=south east] at (\xMax-0.05,0.05) {$R_1$};
\node[color=color2!30!black, anchor=north] at ({min(\xMax,\yMax-1)},{min(\xMax+1,\yMax)}) {$R_2$};
\node[color=color3!30!black, anchor=north west] at (0.05,\yMax-0.05) {$R_3$};

\draw[thick, ->] (\xMin,0) -- (\xMax,0) node[anchor=west] {$\pi_v$};
\foreach \i in {2,4,6} 
{
	\draw[thick] (\i*\xScale, 0.1) -- (\i*\xScale,-0.1) node[anchor=north] {\i};
	\draw[thick] ({(\i-1)*\xScale}, 0.05) -- ({(\i-1)*\xScale},-0.05);
}
\draw[thick, ->] (0,\yMin) -- (0,\yMax) node[anchor=south] {$\pi_t$};
\foreach \i in {2,4,6}
{
	\draw[thick] (0.1, \i*\yScale) -- (-0.1, \i*\yScale) node[anchor=east] {\i};
	\draw[thick] (0.05, {(\i-1)*\yScale}) -- (-0.05, {(\i-1)*\yScale});
}

\draw[thin, dashed] ({max(\xMin,\yMin)},{max(\xMin+1/2,\yMin)}) -- ({min(\xMax,\yMax-1/2)},{min(\xMax+1/2,\yMax)});
\draw[thin, dashed] ({max(\xMin,\yMin)},{max(\xMin+1,\yMin)}) -- ({min(\xMax,\yMax-1)},{min(\xMax+1,\yMax)});

\draw[thick] (0,0) -- (1/2,1) -- (3/2, 5/2) -- ( {min( \xMax, \yMax/2 + 1/4 )} , { min(2*\xMax - 1/2 , \yMax } );
\end{tikzpicture}
\end{center}
\caption{The potential space with $\alpha = 2$.}
\label{fig:degenerateRegion:relaxed}
\end{subfigure}
\hspace{0.02\linewidth}
\begin{subfigure}[t]{.45\linewidth}
\begin{center}
\begin{tikzpicture}
\newcommand{\xMin}{-0.5} \newcommand{\xMax}{3.5}
\newcommand{\yMin}{-0.5} \newcommand{\yMax}{3.5}
\newcommand{\xScale}{0.5} \newcommand{\yScale}{0.5}

\path[fill=color1,fill opacity=0.25] ({max(\xMin,\yMin)},{max(\xMin+1/2,\yMin)}) -- ({max(\xMin,\yMin)},\yMin)--(\xMax,\yMin)--(\xMax,{min(\xMax+1/2,\yMax)})--({min(\xMax,\yMax-1/2)},{min(\xMax+1/2,\yMax)});
\path[fill=color3,fill opacity=0.25] ({max(\xMin,\yMin)},{max(\xMin+1/2,\yMin)}) -- ({max(\xMin,\yMin)},\yMax)--({min(\xMax,\yMax-1/2)},\yMax)--({min(\xMax,\yMax-1/2)},{min(\xMax+1/2,\yMax)});
\node[color=color1!30!black, anchor=south east] at (\xMax-0.05,0.05) {$R_1$};
\node[color=color2!30!black, anchor=north] at ({min(\xMax,\yMax-1/2)},{min(\xMax+1/2,\yMax)}) {$R_2$};
\node[color=color3!30!black, anchor=north west] at (0.05,\yMax-0.05) {$R_3$};

\draw[thick, ->] (\xMin,0) -- (\xMax,0) node[anchor=west] {$\pi_v$};
\foreach \i in {2,4,6} 
{
	\draw[thick] (\i*\xScale, 0.1) -- (\i*\xScale,-0.1) node[anchor=north] {\i};
	\draw[thick] ({(\i-1)*\xScale}, 0.05) -- ({(\i-1)*\xScale},-0.05);
}
\draw[thick, ->] (0,\yMin) -- (0,\yMax) node[anchor=south] {$\pi_t$};
\foreach \i in {2,4,6}
{
	\draw[thick] (0.1, \i*\yScale) -- (-0.1, \i*\yScale) node[anchor=east] {\i};
	\draw[thick] (0.05, {(\i-1)*\yScale}) -- (-0.05, {(\i-1)*\yScale});
}

\draw[color2] ({max(\xMin,\yMin)},{max(\xMin+1/2,\yMin)}) -- ({min(\xMax,\yMax-1/2)},{min(\xMax+1/2,\yMax)});
\draw[dashed] ({max(\xMin,\yMin)},{max(\xMin+1/2,\yMin)}) -- ({min(\xMax,\yMax-1/2)},{min(\xMax+1/2,\yMax)});

\draw[thick] (0,0) -- (1/2,1) -- (3/2, 2) -- ( {min( \xMax, \yMax/2 + 1/4 )} , { min(2*\xMax - 1/2 , \yMax } );
\end{tikzpicture}
\end{center}
\caption{Potential space with constant cost ($\alpha = \infty$) on $e_3$ in $R_2$. The region $R_2$ is $e_3$-degenerate.}
\label{fig:degenerateRegion:degenerate}
\end{subfigure}
\caption{An example where the slope of one cost function depends on the parameter $\alpha$. For $\alpha = \infty$ the region is $e_3$-degenerate.}
\label{fig:degenerateRegion}
\end{figure}
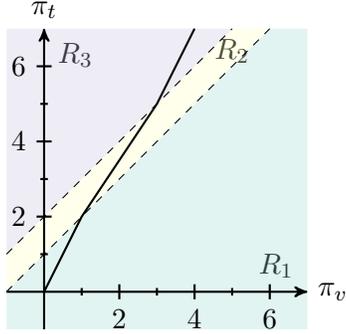
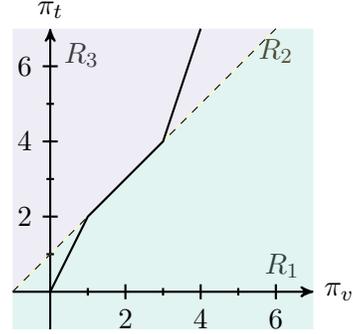

\subparagraph{The Sherman-Morrison Operator}

In order to derive the inverse Laplacian matrix in a degenerate region, we need to investigate further the transformation of the inverse Laplacian when crossing a boundary. We know from Theorem~\ref{thm:boundaryCrossing} that we can compute the new inverse behind some boundary with the Sherman-Morrison formula.  To ease notation and further analysis of this formula we define the following operator.

\begin{definition} 
For two vectors $\vec{v}, \vec{w} \in \R^n$ and a constant $\alpha \in \R \setminus \{0\}$ we define the Sherman-Morrison operator by
\begin{align*}
\shermanOp^{\alpha}_{\vec{v}, \vec{w}} &: \mathcal{D}_{\shermanOp^{\alpha}_{\vec{v}, \vec{w}}} \to \R^{n \times n} \\
\vec{A} &\mapsto
\bigg( \identityM - \frac{\alpha}{1 + \alpha \, \vec{v}^\top \vec{A} \vec{w}} \, \vec{A} \, \vec{w} \, \vec{v}^\top \bigg) \vec{A}
\end{align*}
where
\[
\mathcal{D}_{\shermanOp^{\alpha}_{\vec{v}, \vec{w}}} := \Big\{ \vec{A} \in \R^{n \times n} \; | \;  \vec{v}^\top \vec{A} \vec{w} \neq - \frac{1}{\alpha} \Big\}
\]
is the domain of the operator.
\end{definition}

The definition holds both for regular and singular matrices. In the case of regular matrices $\vec{A}$, the operator $\shermanOp^{\alpha}_{\vec{v}, \vec{w}}$ can be simplified by the use of the Sherman-Morrison formula.

\begin{lemma} \label{lem:shermanOpRegular}
Let $\vec{A} \in \R^{n \times n}$ be a regular matrix. Then $\vec{A} \in \mathcal{D}_{\shermanOp^{\alpha}_{\vec{v}, \vec{w}}}$ if and only if the matrix $\big( \vec{A}^{-1} + \alpha \, \vec{w} \, \vec{v}^{\top} \big)$ is regular. In this case, the Sherman-Morrison operator can be written as
\[
\shermanOp^{\alpha}_{\vec{v}, \vec{w}} ( \vec{A} ) = \big( \vec{A}^{-1} + \alpha \, \vec{w} \, \vec{v}^{\top} \big)^{-1}
.
\]
In particular, the matrix $\shermanOp^{\alpha}_{\vec{v}, \vec{w}} ( \vec{A} )$ is regular as well.
\end{lemma}
\begin{proof}
This is the Sherman-Morrison formula, see \cite{hager1989updating, sherman1950adjustment}.
\end{proof}

With the Sherman-Morrison operator, we can express the result of Theorem~\ref{thm:boundaryCrossing} as follows. If the solution curve crosses a boundary between the region $R_{\vec{t}^i}$ and $R_{\vec{t}^{i+1}}$ induced by some edge $\edge^*$ and the capacity of $\edge^*$ changes by $\Delta c_{e^*}$ then
\[
\laplaceMsI_{\vec{t}^{i+1}} = \shermanOp^{\Delta c_{e^*}}_{\incidenceVs_{\edge^*}, \incidenceVs_{\edge^*}} \big( \laplaceMsI_{\vec{t}^i} \big)
.
\]

By Lemma~\ref{lem:shermanOpRegular} we know that applying the Sherman-Morrison operator to some regular matrix $\vec{A}$ is the same as adding some rank $1$ matrix to the inverse $\vec{A}^{-1}$. This implies that the result of applying the Sherman-Morrison operator multiple times to some matrix $\vec{A}$ does not depend on the order. The next lemma shows that this is also true in the case if the matrix $\vec{A}$ is singular.

Since we need to deal with convergence of matrices, we define componentwise convergence as follows.
Let $\vec{A} \in \R^{n \times n}$ be some matrix with entries $a_{ij}$ and $( \vec{A}^k )_{k \in \N}$ be some sequence of matrices in $\R^{n \times n}$ with entries $a_{ij}^k$. If $\lim_{k \to \infty} a_{ij}^k = a_{ij}$ for every $i,j = 1, \dotsc, n$ then we say the sequence $\vec{A}^k$ converges to $\vec{A}$ componentwise and write $\lim_{k \to \infty} \vec{A}^k = \vec{A}$.

\begin{lemma} \label{lem:shermanOpCommutating}
For all vectors $\vec{v}^1, \vec{v}^2, \vec{w}^1, \vec{w}^2$ and constants $\alpha_1, \alpha_2 \neq 0$ we have
\begin{enumerate}[(i)]
\item 
	$
	\shermanOp^{\alpha_1}_{\vec{v}^1, \vec{w}^1} \big( \shermanOp^{\alpha_2}_{\vec{v}^2, \vec{w}^2} (\vec{A}) \big)
	=
	\shermanOp^{\alpha_2}_{\vec{v}^2, \vec{w}^2} \big( \shermanOp^{\alpha_1}_{\vec{v}^1, \vec{w}^1} (\vec{A}) \big)
	$
	for all matrices $\vec{A}$ with $\vec{A} \in \mathcal{D}_{\shermanOp^{\alpha_1}_{\vec{v}^1, \vec{w}^1}} \cap \mathcal{D}_{\shermanOp^{\alpha_2}_{\vec{v}^2, \vec{w}^2}}$ and $\shermanOp^{\alpha_2}_{\vec{v}^2, \vec{w}^2} (\vec{A}) \in \mathcal{D}_{\shermanOp^{\alpha_1}_{\vec{v}^1, \vec{w}^1}}$.
\item 
	$
	\shermanOp^{\alpha_1}_{\vec{v}^1, \vec{w}^1} \big( \shermanOp^{\alpha_2}_{\vec{v}^1, \vec{w}^1} ( \vec{A} ) \big)
	=
	\shermanOp^{\alpha_1 + \alpha_2}_{\vec{v}^1, \vec{w}^1} ( \vec{A} )
	$
	for all matrices $\vec{A}$ with $\vec{A} \in \mathcal{D}_{\shermanOp^{\alpha_2}_{\vec{v}^1, \vec{w}^1}} \cap \mathcal{D}_{\shermanOp^{\alpha_1 + \alpha_2}_{\vec{v}^1, \vec{w}^1}}$.
\end{enumerate}

\end{lemma}
\begin{proof}
If $\vec{A}$ is regular then both statements follow immediately from Lemma~\ref{lem:shermanOpRegular}.
If $\vec{A}$ is singular, then define $\vec{A}^k := \big( \vec{A} - \frac{1}{k} \, \identityM \big)$ for $k \in \N$. Since $\vec{A}$ has at most $n$ different eigenvalues $\det ( \vec{A}^k ) \neq 0$ for almost all $k \in \N$ and thus almost all matrices $\vec{A}^k$ are regular. Furthermore, $\lim_{k \to \infty} \vec{A}^{k} = \vec{A}$ componentwise. Since the domain $\mathcal{D}_{\shermanOp^{\alpha}_{\vec{v}, \vec{w}}}$ is an open set, almost all $\vec{A}^k$ are also in the domain $\mathcal{D}_{\shermanOp^{\alpha}_{\vec{v}, \vec{w}}}$ if $\vec{A} \in \mathcal{D}_{\shermanOp^{\alpha}_{\vec{v}, \vec{w}}}$. By the definition of the Sherman-Morrison operator, we get that
\[
\lim_{k \to \infty} \shermanOp^{\alpha}_{\vec{v}, \vec{w}} \big( \vec{A}^k \big) = \shermanOp^{\alpha}_{\vec{v}, \vec{w}} \Big( \lim_{k \to \infty} \vec{A}^k \Big)
\]
and hence both statements transfer to singular matrices as well.
\end{proof}

If the slope of the cost function of some edge tends to zero, i.e. $a_{e,t_e} \to 0$, then we have for the conductivity that $c_{e, t_e} \to \infty$. We are, thus, interested in the behavior of the operator $\shermanOp^{\Delta c_{e^*}}_{\incidenceVs_{\edge^*}, \incidenceVs_{\edge^*}}$ as $\Delta c_{e^*} \to \infty$.

\begin{definition}
For any two vectors $\vec{v}, \vec{w} \in \R^n$ we define a new operator $\shermanOp_{\vec{v}, \vec{w}}$ by
\[
\shermanOp_{\vec{v}, \vec{w}} := \bigg( \identityM - \frac{1}{\vec{v}^\top \vec{A} \vec{w}} \, \vec{A} \, \vec{w} \, \vec{v}^\top \bigg) \vec{A}
\]
with the domain
$
\mathcal{D}_{\shermanOp_{\vec{v}, \vec{w}}} := \big\{ \vec{A} \in \R^{n \times n} \; | \;  \vec{v}^\top \vec{A} \vec{w} \neq 0 \big\}
.
$
\end{definition}
By definition of the new operator $\shermanOp_{\vec{v}, \vec{w}}$ we get
\[
\lim_{\alpha \to \infty} \shermanOp^{\alpha}_{\vec{v}, \vec{w}}  = \shermanOp_{\vec{v}, \vec{w}}
\]
pointwise. Since the operators $\shermanOp_{\vec{v}, \vec{w}}^{\alpha}$ and $\shermanOp_{\vec{v}, \vec{w}}$ are continuous in every component, all properties from Lemma~\ref{lem:shermanOpCommutating} hold  true for the operator $\shermanOp_{\vec{v}, \vec{w}}$ as well. Furthermore, the operator $\shermanOp_{\vec{v}, \vec{w}}$ satisfies the following property.
\begin{lemma} \label{lem:shermanOpLimitProperties}
For some vectors $\vec{v}, \vec{w} \in \R^n$ let $\vec{A} \in \mathcal{D}_{\shermanOp_{\vec{v}, \vec{w}}}$. Then
\begin{align*}
\vec{v}^{\top}  \shermanOp_{\vec{v}, \vec{w}} ( \vec{A} ) = \vec{0} =  \shermanOp_{\vec{v}, \vec{w}} ( \vec{A} ) \, \vec{w}
.
\end{align*}
\end{lemma}
\begin{proof}
This follows directly from the definition of $\shermanOp_{\vec{v}, \vec{w}}$.
\end{proof}

\subparagraph{Entering a degenerate region}
We want to deal with the transition from a normal region to a $e$-degenerate region. Assume $R_{\vec{t}^i}$ is a region where all cost functions are strictly increasing and $R_{\vec{t}^{i+1}}$ a neighboring region where the cost function of some edge $e$ is constant. We can compute the inverse Laplacian matrix by replacing the zero slope by some small $\frac{1}{\alpha} > 0 $ and then compute the limit for $\alpha \to \infty$. Consider Figure~\ref{fig:degenerateRegion} for an example. For $\alpha \in \R$ the region $R_2$ is a normal region, the inverse of the Laplacian can be computed with the Sherman-Morrison operator (e.g. for $\alpha = 2$ in Figure~\ref{fig:degenerateRegion}(c)). The inverse in the region $R_{\vec{t}^{i+1, \alpha}}$ in dependence of the parameter $\alpha$ is thus
\begin{align*}
\laplaceMsI_{\vec{t}^{i+1, \alpha}} = \shermanOp^{\alpha - c_{e, t^i_e}}_{\hat{\vec{\gamma}}_e, \hat{\vec{\gamma}}_e} \big( \laplaceMsI_{\vec{t}^i} \big)
\end{align*}
and hence the inverse of the Laplacian in the region $R_{\vec{t}^{i+1}}$ with constant cost on edge $e$ can be computed as
\begin{align*}
\laplaceMsI_{\vec{t}^{i+1}} = \lim_{\alpha \to \infty} \shermanOp^{\alpha - c_{e, t^i_e}}_{\hat{\vec{\gamma}}_e, \hat{\vec{\gamma}}_e} \big( \laplaceMsI_{\vec{t}^i} \big)
= \shermanOp_{\hat{\vec{\gamma}}_e, \hat{\vec{\gamma}}_e} \big( \laplaceMsI_{\vec{t}^i} \big)
.
\end{align*}
By definition, the operator $\shermanOp_{\incidenceVs_\edge, \incidenceVs_\edge}$ (an thus the limit) is well-defined if $\incidenceVs_\edge^{\top} \, \laplaceMsI_{\vec{t}^i} \, \incidenceVs_\edge \neq 0$ which is always the case since $\smash{\laplaceMsI_{\vec{t}^i}}$ is positive definite by Lemma~\ref{lem:LaplacianMatrix}(\ref{lem:LaplacianMatrix:subMatrix}).
Lemma~\ref{lem:shermanOpLimitProperties} implies $\Delta \pi^{i+1}_w - \Delta \pi^{i+1}_v = \incidenceVs_e^{\top} \, \Delta \vec{\pi}^{i+1} = \incidenceVs_e^{\top} \, \shermanOp_{\hat{\vec{\gamma}}_e, \hat{\vec{\gamma}}_e} \big( \vec{L}_{\vec{t}^i}^{-1} \big) \, \Delta \vec{y} = 0$ for the edge $e =(v,w)$ with constant cost and thus the potential difference is indeed constant within this region.

With the slope $a_{e, t_e^{i+1}}$ being zero, the capacity $c_{e, t_e^{i+1}}$ is not well defined and thus the matrix $\vec{C}_{\vec{t}^{i+1}}$ is not well defined either. Hence, we can not compute the flows directly as before. The next Lemma gives a direct formula to compute the flow on all edges (in particular on the edge with constant cost) and proves that the potential direction  $\Delta \vec{\pi}^{i+1} = \vec{L}_{\vec{t}^{i+1}}^{-1} \Delta \vec{y} = \shermanOp_{\hat{\vec{\gamma}}_e, \hat{\vec{\gamma}}_e} \big( \vec{L}_{\vec{t}^i}^{-1} \big) \Delta \vec{y}$ induced by the Laplacian Matrix is indeed a valid potential driection.

\begin{lemma} \label{lem:degenerateRegionFlows}
Let $R_{\vec{t}^i}$ be some region with the inverse Laplacian matrix $\laplaceMsI_{\vec{t}^i}$ and the potential direction $\Delta \vec{\pi}^i$.
If a neighboring region $R_{\vec{t}^{i+1}}$ is $e$-degenerate then the vector
\begin{equation*}
\Delta x^{i+1}_{\altEdge} = 
\begin{cases}
c_{\altEdge, \regionI^{i+1}_{\altEdge}} \, \incidenceV_\altEdge^{\top} \, \Delta \potentialV^{i+1}
&\text{if } \altEdge \neq \edge, \\
\frac{\incidenceV_\edge^{\top} \Delta \potentialV^{i}}{\incidenceVs^{\top}_\edge \, \laplaceMsI_{\regionV^{i}} \, \incidenceVs_\edge}
&\text{if } \altEdge = \edge
\end{cases}
\end{equation*}
is a $s$-$t$-flow for the demand $r=1$ that satisfies
\begin{equation} \label{eq:flowProperty}
\aM_{\regionV^{i+1}} \, \Delta \flowV^{i+1} = \incidenceM \, \Delta \potentialV^{i+1}
\end{equation}
where $\Delta \potentialV^{i+1} = \shermanOp_{\incidenceVs_\edge, \incidenceVs_\edge} \big( \laplaceMsI_{\regionV^i} \big) \, \Delta \excessV$.
\end{lemma}
\begin{proof}
Assume the slope of the cost of edge $\edge$ in region $R_{\regionV^{i+1}}$ is $\tilde{a}_{\edge, t^{i+1}_\edge} = \frac{1}{\alpha}$. Then 
\[
\Delta \tilde{\flowV}^{i+1} = \tilde{\cM}_{\regionV^{i+1}} \, \incidenceM \, \Delta \tilde{\potentialV}^{i+1}
\]
is a $s$-$t$-flow that satisfies \eqref{eq:flowProperty} by definition. Taking the limit $\alpha \to \infty$ we get for $\altEdge \neq \edge$ that
\begin{align*}
\lim_{\alpha \to \infty} \Delta \tilde{\flow}^{i+1}_{\altEdge} &= \lim_{\alpha \to \infty} c_{\altEdge, t_{\altEdge}^{i+1}} \, \incidenceV_{\altEdge}^{\top} \,  \Delta \tilde{\potentialV}^{i+1} \\
&= \lim_{\alpha \to \infty} c_{\altEdge, t_{\altEdge}^{i+1}} \, \incidenceV_{\altEdge}^{\top} \, \shermanOp^{\alpha - c_{\altEdge, t_{\altEdge}^{i}}}_{\incidenceVs_\altEdge, \incidenceVs_\altEdge} \Big( \laplaceMsI_{\regionV^i} \Big) \, \Delta \excessV \\
&= c_{\altEdge, t_{\altEdge}^{i+1}} \, \incidenceV_{\altEdge}^{\top} \, \shermanOp_{\incidenceVs_\altEdge, \incidenceVs_\altEdge} \Big( \laplaceMsI_{\regionV^i} \Big) \, \Delta \excessV \\
&= c_{\altEdge, t_{\altEdge}^{i+1}} \, \incidenceV_{\altEdge}^{\top} \, \Delta \potentialV^{i+1}
= \Delta \flow_{\altEdge}^{i+1}
\end{align*}
and for $\altEdge = \edge$ that
\begin{align*}
\lim_{\alpha \to \infty} \Delta \tilde{\flow}^{i+1}_{\altEdge} &= \lim_{\xi \to \infty} \alpha \,  \incidenceV_{\edge}^{\top} \,  \Delta \tilde{\potentialV}^{i+1} \\
&=\lim_{\alpha \to \infty} \alpha \,  \incidenceV_{\edge}^{\top} \,  \shermanOp^{\alpha - c_{\edge, t_{\edge}^{i}}}_{\incidenceVs_\edge, \incidenceVs_\edge} \Big( \laplaceMsI_{\regionV^i} \Big) \, \Delta \excessV \\
&=\lim_{\alpha \to \infty} \alpha \, \bigg(  1 - \frac{(\alpha - c_{\edge, t_{\edge}^{i}}) \, \incidenceV_\edge^\top \, \laplaceMsI_{\regionV^i} \incidenceV_\edge}{1 + (\alpha - c_{\edge, t_{\edge}^{i}}) \, \incidenceVs_\edge^\top \, \laplaceMsI_{\regionV^i} \, \incidenceVs_\edge} \bigg) \, \incidenceVs_\edge^\top \, \laplaceMsI_{\regionV^i} \, \Delta \excessVs \\
&=\lim_{\alpha \to \infty} \frac{\alpha}{1 + (\alpha - c_{\edge, t_{\edge}^{i}}) \, \incidenceVs_\edge^\top \, \laplaceMsI_{\regionV^i} \, \incidenceVs_\edge} \, \incidenceVs_\edge^\top \, \laplaceMsI_{\regionV^i} \, \Delta \excessVs \\
&= \frac{\incidenceV_\edge^\top \, \Delta \potentialV^{i}}{\incidenceVs_\edge^\top \, \laplaceMsI_{\regionV^i} \, \incidenceVs_\edge} = \Delta \flow_{\edge}^{i+1}
.
\end{align*}
\end{proof}
Note that the flow $\Delta x_{e}^{i+1} = \frac{\incidenceV_\edge^{\top} \Delta \potentialV^{i}}{\incidenceVs^{\top}_\edge \, \laplaceMsI_{\regionV^{i}} \, \incidenceVs_\edge}$ on the edge $e$ with constant cost is well-defined since $\smash{\incidenceVs^{\top}_\edge \, \laplaceMsI_{\regionV^{i}} \, \incidenceVs_\edge} > 0$. Further, this implies that the boundary crossing is well-defined in the following sense: if the potential difference $\incidenceV_\edge^{\top} \Delta \potentialV^{i}$ was positive before crossing the boundary, the flow was increasing on the edge $e$. Since $\Delta x_{e}^{i+1}$ is positive in this case, too, the flow does not hit the same breakpoint again. The same holds if the potential difference $\incidenceV_\edge^{\top} \Delta \potentialV^{i}$ was negative (we will formalize this in Lemma~\ref{lem:flowsBoundaryCrossing}).

Lemma~\ref{lem:degenerateRegionFlows} shows that the flow $\Delta \vec{x}^{i+1}$ is a $s$-$t$-flow. By Lemma~\ref{lem:characterisationWEPotentials} and with $\eqref{eq:flowProperty}$ we know that this flow is a Wardrop equilibrium for the cost functions $l_e(x) = a_e \, x$ with the potentials $\Delta \vec{\pi}^{i+1}$. By linearity, adding the flow $\Delta \vec{x}^{i+1}$ to the  Wardrop equilibrium $\vec{x}^i$ and the potentials $\Delta \vec{\pi}^{i+1}$ to the potentials $\vec{\pi}^i$ yields again a Wardrop flow and potential by Lemma~\ref{lem:characterisationWEPotentials}. Thus, using the potential direction $\Delta \vec{\pi}$ and the flow direction $\Delta \vec{x}$ is the justified.

\subparagraph{Finding the closest boundary}
Now assume, the region $R_{\regionV^i}$ is a degenerate region. The distance to the next boundary can be computed by
\[
\epsilon (\altEdge) = 
\begin{cases}
		\frac{\sigma_{\altEdge, t^i_e+1} - \vec \gamma_\altEdge^{\top} \vec \pi^i}{\vec \gamma_\altEdge^{\top} \Delta \vec \pi^i}  & \text { if } \vec \gamma_\altEdge^{\top} \Delta \vec \pi^i > 0,\\
		\frac{\sigma_{\altEdge, t^i_e\phantom{+1}} - \vec \gamma_\altEdge^{\top} \vec \pi^i}{\vec \gamma_\altEdge^{\top} \Delta \vec \pi^i} & \text{ if } \vec \gamma_\altEdge^{\top} \Delta \vec \pi^i < 0,\\
		\infty & \text{ else.}
	\end{cases}
\]
for all edges $\altEdge \neq \edge$ as before. For the edge $\edge$ we can use the flow formulation and compute
\[
\epsilon( \edge) = 
\begin{cases}
		\frac{\tau_{\edge, \regionI^{i}_\edge + 1} - \flow_\edge^i}{\Delta \flow^{i}_\edge}  & \text { if } \Delta \flow_\edge^{i} > 0,\\
		\frac{\tau_{e, t^i_e\phantom{+1}} - \flow_\edge^i}{\Delta \flow^{i}_\edge} & \text{ if }  \Delta \flow_\edge^{i} < 0,\\
		\infty & \text{ else}
	\end{cases} 
\]
where $\flow_\edge^i$ is the flow on edge $e$ when entering the region $R_{\vec{t}^i}$. 
Depending on whether the minimal epsilon is induced by some $\altEdge \neq \edge$ or by $\edge$ the solution curve either crosses a boundary within the degenerate region or leaves the degenerate region.

\subparagraph{Crossing a boundary within the degenerate region}

Assume that the solution curve hits a boundary induced by some edge $\altEdge$ within a $e$-degenerate region. The next lemma proves that the inverse Laplacian matrix can be computed as in the non-degenerate case with the Sherman-Morrison operator.

To compute the flow on the edge $e$ with constant cost, we first have to introduce some new notation. As we have seen in Lemma~\ref{lem:degenerateRegionFlows}, when entering the $e$-degenerate region, the flow on this edge can be computed depending on the Laplacian matrix and the direction $\Delta \vec{\pi}$ in a neighboring, non-$e$-degenerate region.
This implies that we need a non-$e$-degenerate, neighboring region to compute the flow.
Recall, that two regions $R_{\vec{t}^1}, R_{\vec{t}^2}$ are called neighboring, if there is an edge $e$ such that the vectors $\vec{t}^1, \vec{t}^2$ satisfy $t^1_e = t^2_e \pm 1$ and $t^1_\altEdge = t^2_\altEdge$ for all edges $\altEdge \neq e$. We denote by
\[
\neighboringR_{\vec{t}}^e := \{ \vec{t}' \in \bar{T} : t'_e = t_e \pm 1 \text{ and } t'_\altEdge = t_\altEdge \text{ for all } \altEdge \neq e \}
\]
the set of region vectors of regions neighboring $R_{\vec{t}}$ with the common boundary induced by edge $e$. We call these regions $e$-neighboring. The set $\neighboringR_{\vec{t}} := \bigcup_{e \in E} \neighboringR_{\vec{t}}^e$ contains the region vectors of all neighboring regions.
Note that the set $\neighboringR_{\vec{t}}^e$ contains  either no element, one element or two elements. Further, if $R_{\vec{t}}$ is $e$-degenerate, the regions $R_{\vec{t}'}$ with $\vec{t}' \in \neighboringR_{\vec{t}}^e$ are non-$e$-degenerate.

\begin{lemma} \label{lem:degenerateBoundaryCrossing}
Let $R_{\vec{t}^i}$ be a $e$-degenerate region and $R_{\vec{t}^{i+1}}$ a neighboring $e$-degenerate region where the boundary is induced by some edge $\altEdge \neq e$. Then
\begin{align*}
\laplaceMsI_{\vec{t}^{i+1}} = \shermanOp^{\Delta c_{\altEdge}}_{\incidenceVs_\altEdge, \incidenceVs_\altEdge} \big( \laplaceMsI_{\vec{t}^i} \big)
.
\end{align*}
The flow direction on the edge $e$ with constant cost can be computed as
\[
\Delta x^{i+1}_e = \frac{\vec{\gamma}_e^{\top} \Delta \vec{\pi}^{i+1}_{-e}}{\incidenceVs_\edge^{\top} \laplaceMsI_{\vec{t}^{i+1}_{-e}} \incidenceVs_\edge}
\]
where $\vec{t}^{i+1}_{-e} \in \neighboringR_{\vec{t}^{i+1}}^e$ is some $e$-neighboring region with the Laplace inverse $\laplaceMsI_{\vec{t}^{i+1}_{-e}}$ and the potential direction $\Delta \vec{\pi}^{i+1}_{-e}$ defined by $\Delta \hat{\vec{\pi}}^{i+1}_{-e} := \laplaceMsI_{\vec{t}^{i+1}_{-e}} \, \Delta \hat{\vec{y}}$.
\end{lemma}
\begin{proof}
Let $\vec{t}^{i}_{-e} \in \neighboringR_{\vec{t}^i}^e$ be a region vector of a $e$-neighboring region of $R_{\vec{t}^{i}}$. Then the region $R_{\vec{t}^i_{-e}}$ is non-$e$-degnerate. Then 
\[
\laplaceMsI_{\vec{t}^{i}} = \shermanOp_{\incidenceVs_\edge, \incidenceVs_\edge} \big( \laplaceMsI_{\vec{t}^{i}_{-e}} \big) = \lim_{\alpha \to \infty} \shermanOp^{\alpha - c_{e, t^{i-1}_e}}_{\incidenceVs_\edge, \incidenceVs_\edge} \big( \laplaceMsI_{\vec{t}^{i}_{-e}} \big)
.
\]
Let $\vec{t}^{i+1}_{-e} \in \neighboringR_{\vec{t}^i_{-e}} \cap \neighboringR_{\vec{t}^{i+1}}$. Since $R_{\vec{t}^i}$ and $R_{\vec{t}^i}$ are $\altEdge$-neighboring, the regions $R_{\vec{t}^{i}_{-e}}$ and $R_{\vec{t}^{i+1}_{-e}}$ are $\altEdge$-neighboring as well. Thus,
$\laplaceMsI_{\vec{t}^{i+1}_{-e}} = \shermanOp^{\Delta c_{\altEdge}}_{\incidenceVs_{\altEdge}, \incidenceVs_{\altEdge}} \big( \laplaceMsI_{\vec{t}^i_{-e}} \big)$ and with Lemma~\ref{lem:shermanOpCommutating} we get
\[
\laplaceMsI_{\vec{t}^{i+1}} = \shermanOp_{\incidenceVs_e, \incidenceVs_e} \big( \laplaceMsI_{\vec{t}^{i+1}_{-e}}  \big)
= \shermanOp_{\incidenceVs_e, \incidenceVs_e} \bigg( \shermanOp^{\Delta c_{\altEdge}}_{\incidenceVs_{\altEdge}, \incidenceVs_{\altEdge}} \big( \laplaceMsI_{\vec{t}^i_{-e}} \big) \bigg) = \shermanOp^{\Delta c_{\altEdge}}_{\incidenceVs_{\altEdge}, \incidenceVs_{\altEdge}} \bigg( \shermanOp_{\incidenceVs_e, \incidenceVs_e} \big( \laplaceMsI_{\vec{t}^i_{-e}} \big) \bigg)
=\shermanOp^{\Delta c_{\altEdge}}_{\incidenceVs_\altEdge, \incidenceVs_\altEdge} \big( \laplaceMsI_{\vec{t}^i} \big)
.
\]
The formula for the flow direction $\Delta x^{i+1}_e$ follows directly from Lemma~\ref{lem:degenerateRegionFlows}.
\end{proof}
\begin{remark}
Lemma~\ref{lem:degenerateBoundaryCrossing} is also true if the region $R_{\vec{t}^{i+1}}$ is $\altEdge$-degenerate. The new inverse can be computed as $\laplaceMsI_{\vec{t}^{i+1}} = \shermanOp_{\incidenceVs_\altEdge, \incidenceVs_\altEdge} \big( \laplaceMsI_{\vec{t}^i} \big)$ in this case. Further, Lemma~\ref{lem:degenerateBoundaryCrossing} relies on the fact that there is some $e$-neighboring region $R_{\vec{t}^{i+1}_{-e}}$. It may be the case that the neighborhood $\neighboringR_{\vec{t}^{i+1}}^e$ is empty. In this case, we can introduce a new, artificial region where the edge $e$ has some stricly positive slope (e.g. $a_e = 1$), compute the inverse Laplacian matrix for this region and use this region as $e$-neighboring region.
\end{remark}

The next Lemma proves that all boundary crossings are well-defined, i.e. that $\Delta x^i_e$ and $\Delta x^{i+1}_e$ have the same sign when crossing a boundary that is induced by some edge $e$. 

\begin{lemma} \label{lem:flowsBoundaryCrossing}
When the solution curve crosses a boundary induced by some edge $e$ between two regions $R_{\vec{t}^1}$ and $R_{\vec{t}^2}$, we have
\[
\sgn( \Delta x^1_e ) = \sgn( \Delta x^2_e )
.
\]
\end{lemma}
\begin{proof}
Note, that all capacity values $c_{e, t_e} = \frac{1}{a_{e, t_e}}$ are strictly positive.  If the the regions $R_{\vec{t}^1}, R_{\vec{t}^2}$ are non-degenerate, the claim follows directly from Theorem~\ref{thm:boundaryCrossing}(\ref{thm:boundaryCrossing:Directions}) since $\Delta x_e = c_{e, t_e} \, \incidenceV_e^{\top} \Delta \vec{\pi}$. \\
If $R_{\vec{t}^1}$ is non-$e$-degenerate and $R_{\vec{t}^2}$ is $e$-degenerate, then Lemma~\ref{lem:degenerateRegionFlows} implies that 
\[
\Delta x^2_e = \frac{\incidenceV_e^\top \, \Delta \vec{\pi}^1}{\incidenceVs_e^\top \laplaceMsI_{\vec{t}^1} \incidenceVs_e^\top} = \frac{1}{c_{e, t^1_e} \, \incidenceVs_e^\top \laplaceMsI_{\vec{t}^1} \incidenceVs_e^\top} \Delta x^1_e
.
\]
Since $\incidenceVs_e^\top \laplaceMsI_{\vec{t}^1} \incidenceVs_e^\top > 0$ by Lemma~\ref{lem:LaplacianMatrix}(\ref{lem:LaplacianMatrix:subMatrix}) the claim follows. By symmetry, this also holds if $R_{\vec{t}^1}$ is $e$-degenerate and $R_{\vec{t}^1}$ is non-$e$-degenerate.

If the regions $R_{\vec{t}^1}, R_{\vec{t}^2}$ are degenerate they both are $\altEdge$-degenerate for some other $\altEdge \neq e$. (If, say, $R_{\vec{t}^1}$ was $e$-degenerate, then after crossing a boundary induced by $e$, the region $R_{\vec{t}^2}$ must be non-$e$-degenerate.) Then, by Lemma~\ref{lem:degenerateBoundaryCrossing} we get
\begin{align*}
\Delta x^2_e 
&= c_{e,t^2_e} \, \incidenceVs_e^\top \, \Delta \hat{\vec{\pi}}^2 \\
&= c_{e,t^2_e} \, \incidenceVs_e^\top \, \laplaceMsI_{\vec{t}^2} \, \Delta \hat{\vec{y}}  \\
&= c_{e,t^2_e} \, \incidenceVs_e^\top \, \shermanOp^{\Delta c_e}_{\incidenceVs_e, \incidenceVs_e} \big( \laplaceMsI_{\vec{t}^2} \big) \, \Delta  \hat{\vec{y}} \\
&= c_{e,t^2_e} \, \incidenceVs_e^\top \, \bigg( \vec{I}_{n-1} - \frac{\Delta c_e}{1 + \Delta c_e \incidenceVs_e^{\top} \laplaceMsI_{\vec{t}^1} \incidenceVs_e} \laplaceMsI_{\vec{t}^1} \incidenceVs_e \incidenceVs_e^{\top} \bigg) \, \laplaceMsI_{\vec{t}^1} \, \Delta \hat{\vec{y}}  \\
&= c_{e,t^2_e}  \, \bigg( \frac{1}{1 + \Delta c_e \incidenceVs_e^{\top} \laplaceMsI_{\vec{t}^1} \incidenceVs_e} \bigg) \, \incidenceVs_e^\top \, \laplaceMsI_{\vec{t}^1} \, \Delta \hat{\vec{y}}
= \frac{c_{e,t^2_e}}{c_{e,t^1_e} (1 + \Delta c_e \incidenceVs_e^{\top} \laplaceMsI_{\vec{t}^1} \incidenceVs_e )} \, \Delta x^1_e
\end{align*}
Since $\frac{1}{c_{e, t^1_e}} \geq \incidenceVs_e^{\top} \laplaceMsI_{\vec{t}^1} \incidenceVs_e$ 
(by Theorem~\ref{thm:effectiveTravelTime}(\ref{thm:effectiveTravelTime:distance}) in Appendix~\ref{app:laplacian}), $\Delta c_e > - c_{e,t^1_e}$, and $\incidenceVs_e^{\top} \laplaceMsI_{\vec{t}^1} \incidenceVs_e \geq 0$ by Lemma~\ref{lem:LaplacianMatrix}(\ref{lem:LaplacianMatrix:semiDefinite}) we get
\[
1 + \Delta c_e \, \incidenceVs_e^{\top} \laplaceMsI_{\vec{t}^1} \incidenceVs_e > 1 - c_{e, t^1_e} \incidenceVs_e^{\top} \laplaceMsI_{\vec{t}^1} \incidenceVs_e \geq 0
\]
and hence the claim follows.
\end{proof}

As we have seen, to obtain the flow in a $e$-degenerate region, we always need a neighboring, non-$e$-degenerate region. So it is convenient to store the Laplacian matrix (or the inverse) when entering a degenerate region. Performing all further boundary crossings that are not related to $e$ yields always a Laplace matrix (respectively its inverse) of a neighboring, non-$e$-degenerate region. \\
If the solution curve starts in a $e$-degenerate region and/or there is no non-$e$-degenerate, one may compute a Laplacian matrix of a artificial region, where the edge $e$ has slope $a_{e,t_e} = 1$ and track all the boundaries with this Laplace matrix. In the worst case, there are up to $m$ additional Laplacian matricies of neighboring, non-degenerate regions that need to be tracked. 

\subparagraph{Leaving a degenerate region}
If, in a $e$-degenerate region, the flow on the edge $e$ reaches a breakpoint, the solution curve needs to leave the $e$-degenerate region. Since the (inverse) Laplacian matrix has lower rank in the $e$-degenerate region, the new inverse outside this region can not be obtained by using the Sherman-Morrison operator. Either, the inverse has to be recomputed by inverting the Laplace matrix, or the new inverse can be obtained by using the inverse Laplacian matrix of a neighboring, non-$e$-degenerate region.

\subparagraph{Degeneracy and degenerate regions}
To handle degenerate points in the presence of degenerate regions we can again use perturbation as before. However, it is not convenient to perturb the potential since even a small perturbation in the potential difference of some edge can lead to an enormous change in flows if the solution curve is in a degenerate region. Consider, for example, the extreme case where the solution curve is in a point such that all edges have constant cost---there is no perturbation of the potential without changing the region of at least one edge. Thus, perturbing the potentials is not valid in this case.

We develop a new perturbation method that again will lead to a lexicographic rule. As before, we assume that the potential $\vec{\pi}^{i+1}$ is a degenerate point, i.e. there is again more than one edge in the set $E^* := \argmin_{e \in E} \epsilon(e)$ in the point $\vec{\pi}^i$. Note that if the region $R_{\vec{t}^i}$ is degenerate, the points $\vec{\pi}^i$ and $\vec{\pi}^{i+1}$ lie on some lower-dimensional hyperplane. It might even occur that $\vec{\pi}^i$ and $\vec{\pi}^{i+1}$ describe the same point if, for example, all edges have constant cost. Therefore, we also consider the corresponding flows $\vec{x}^i$ and $\vec{x}^{i+1}$ and apply the perturbation to the flows. \\
We are going to track a \textit{perturbed solution curve} that traverses the potentials $\potentialVp[l]$ with the flows $\flowVp[l]$ in the regions $\regionp[l]$. For $\delta > 0$, we define the perturbation vector $\vec{\delta} := (\delta^m, \delta^{m-1}, \dotsc, \delta^2, \delta) \in \R^m$. The perturbed solution curve starts in the potential $\potentialVp[0] := \potentialV^i$ in the region $\regionp[0] := R_{\vec{t}^i}$ with the perturbed flows $\flowVp[0] := \vec{x}^i + \vec{\delta}$. As before, when tracking the perturbed solution curve, we are only interested in the boundaries and regions intersecting in the degenerate point. This means, we only consider boundaries induced by edges $e \in E^*$ and some specific flow-breakpoints $\flowBPd_e$.

Perturbing the flows while leaving the potential unchanged disrupts the connection between potentials and flows. The potential $\vec{\pi}^i$ may not be a potential for the flow $\flowVp[0] = \flowV^i + \vec{\delta}$ in the sense of Lemma~\ref{lem:characterisationWEPotentials}. So, it is not clear how to interpret the perturbed solution curve and the regions it traverses. To have a meaningful interpretation of the perturbed solution curve, we consider the cost on some edge induced by the perturbed flow. For every edge, there is some exponent $k_e \in \N$ such that we can express the perturbed flow on edge $e$ as $\flowp[0] = x^i_e + \delta^{k_e}$. Then the cost on the edge is
\[
l_e (\flowp[0]) = a_{e, t^i_e} ( x^i_e + \delta^{k_e} ) + b_{e, t^i_e}
.
\]
Thus, the perturbation acts a right-shift on the cost function $l_e$ by $\delta^{k_e}$. So we can interpret the perturbation of the flows as perturbation of the offsets $b_{e, t^i_e}$ and breakpoints $\tau_e$ of the cost functions. In the potential space, this can be visualized as a slight translation of the boundaries which dissolves the degeneracy in $\vec{\pi}^{i+1}$ (see Figure~\ref{fig:flowPerturbation}). Since the boundaries are only translated in the potential space but not altered in any other way, we still know if the perturbed solution curve moves away from all the boundaries in some region $\regionp[l]$ then the actual solution curve can proceed in this region as well. As before, the step where the solution curve moves away from all boundaries is indexed by the value $\alpha \in \N$.
\begin{figure}[t]
\centering
\begin{subfigure}{.4\linewidth}
\begin{center}
\begin{tikzpicture}
\newcommand{\xMin}{-3} \newcommand{\xMax}{3}
\newcommand{\yMin}{-1.5} \newcommand{\yMax}{2.5}
\path[fill=color1,fill opacity=0.25] (\xMin,0) -- (0,0) -- ({max(\xMin,\yMin)}, {max(\xMin,\yMin)}) -- (\xMin,\yMin);
\path[fill=color2,fill opacity=0.25] ({max(\xMin,\yMin)},{max(\xMin,\yMin)}) -- (0,0) -- (0,\yMin);
\path[fill=color3,fill opacity=0.25] (0,\yMin) -- (0,0) -- (\xMax,0) -- (\xMax,\yMin);
\path[fill=color4,fill opacity=0.25] (\xMax,0) -- (0,0) -- ({min(\xMax,\yMax)},{min(\xMax,\yMax)}) -- (\xMax,\yMax);
\path[fill=color5,fill opacity=0.25] ({min(\xMax,\yMax)},{min(\xMax,\yMax)}) -- (0,0) -- (0,\yMax);
\path[fill=color6,fill opacity=0.25] (0,\yMax) -- (0,0) -- (\xMin,0) -- (\xMin,\yMax);
\draw[thin, dashed] (0,\yMin) -- (0,\yMax);
\draw[thin, dashed] ({max(\xMin,\yMin)}, {max(\xMin,\yMin)}) -- ({min(\xMax,\yMax)},{min(\xMax,\yMax)});
\draw[thin, dashed] (\xMin,0) -- (\xMax,0);


\draw (-1.5,-0.75) node[solid] {} node[above] {$\vec \pi^i$};
\draw (0,0) node[solid] {} node[above left] {$\vec \pi^{i+1}$};
\draw[thick] (-1.5,-0.75)--(0,0)--( {min(\xMax, \yMax/4)} , {min(4*\xMax, \yMax)} );

\end{tikzpicture}
\end{center}
\caption{The potential space without perturbation of the flows.}
\end{subfigure}
\hspace{0.075\linewidth}
\begin{subfigure}{.4\linewidth}
\begin{center}
\begin{tikzpicture}
\newcommand{\xMin}{-3} \newcommand{\xMax}{3}
\newcommand{\yMin}{-1.5} \newcommand{\yMax}{2.5}
\path[fill=color1,fill opacity=0.25] (\xMin,-.25) -- (-.75,-.25) -- ({max(\xMin,\yMin - 0.5)}, {max(\xMin + 0.5,\yMin)}) -- (\xMin,\yMin);
\path[fill=color2,fill opacity=0.25] ({max(\xMin,\yMin - 0.5)}, {max(\xMin + 0.5,\yMin)}) -- (-.75,-.25) -- (0.5,-.25) -- (0.5,\yMin);
\path[fill=color3,fill opacity=0.25] (0.5,\yMin) -- (0.5,-.25) -- (\xMax,-.25) -- (\xMax,\yMin);
\path[fill=color4,fill opacity=0.25] (\xMax,-.25) -- (0.5,-.25) -- (0.5,1) -- ({min(\xMax,\yMax - 0.5)},{min(\xMax + 0.5,\yMax)}) -- (\xMax,\yMax);
\path[fill=color5,fill opacity=0.25] ({min(\xMax,\yMax - 0.5)},{min(\xMax + 0.5,\yMax)}) -- (0.5,1) -- (0.5,\yMax);
\path[fill=color6,fill opacity=0.25] (0.5,\yMax) -- (0.5,1) -- (-0.75,-0.25) -- (\xMin,-0.25) -- (\xMin,\yMax);
\path[fill=color7,fill opacity=0.25] (0.5,1) -- (-0.75,-0.25) -- (0.5,-0.25);

\draw[thin, dashed] (0+.5,\yMin) -- (0+.5,\yMax);
\draw[thin, dashed] ({max(\xMin,\yMin - 0.5)}, {max(\xMin + 0.5,\yMin)}) -- ({min(\xMax,\yMax - 0.5)},{min(\xMax + 0.5,\yMax)});
\draw[thin, dashed] (\xMin,0-.25) -- (\xMax,0-.25);

\draw[thin, dashed, black!40] (0,\yMin) -- (0,\yMax);
\draw[thin, dashed, black!40] ({max(\xMin,\yMin)}, {max(\xMin,\yMin)}) -- ({min(\xMax,\yMax)},{min(\xMax,\yMax)});
\draw[thin, dashed, black!40] (\xMin,0) -- (\xMax,0);


\draw[thick, black!40] (-1.5,-0.75)--(0,0)--( {min(\xMax, \yMax/4)} , {min(4*\xMax, \yMax)} );;
\draw (-1.5,-0.75)  node[above] {$\vec \pi^i$};
\fill[black] (-1.5,-0.75) -- (-1.5,-0.75) + (45:0.075) arc (45:225:0.075);
\draw[black!40] (0,0) node[solid] {} node[above left] {$\vec \pi^{i+1}$};

\draw[thick, darkgreen] (-1.5,-0.75)--(-1,-0.5)--(0.25,-0.25)--(0.5,0.25)--(0.75,5/4) -- ( {min(\xMax, (\yMax + 7/4)/4)} , {min(4*\xMax - 7/4, \yMax)} );;
\fill[darkgreen] (-1.5,-0.75) -- (-1.5,-0.75) + (-135:0.075) arc (-135:45:0.075);
\draw (-1,-0.5) node[solid, darkgreen] {} node[below right=-0.1, darkgreen ] {\footnotesize $\potentialVp[1]$};
\draw (0.25,-0.25) node[solid, darkgreen] {} node[below, darkgreen ] {\footnotesize $\potentialVp[2]$};
\draw (0.5,0.25) node[solid, darkgreen] {} node[right, darkgreen ] {\footnotesize $\potentialVp[3]$};
\draw (0.75,5/4) node[solid, darkgreen] {} node[right, darkgreen ] {\footnotesize $\potentialVp[4]$};

\end{tikzpicture}
\end{center}
\caption{The potential space with perturbed boundaries.}
\end{subfigure}
\caption{The perturbation of the flows can be interpreted as the perturbation of the boundaries. (a) shows a degenerate point, (b) shows the same point with perturbed boundaries. The unperturbed boundaries and the degenerate point are depicted in light gray, the green line represents the perturbed solution curve.}
\label{fig:flowPerturbation}
\end{figure}
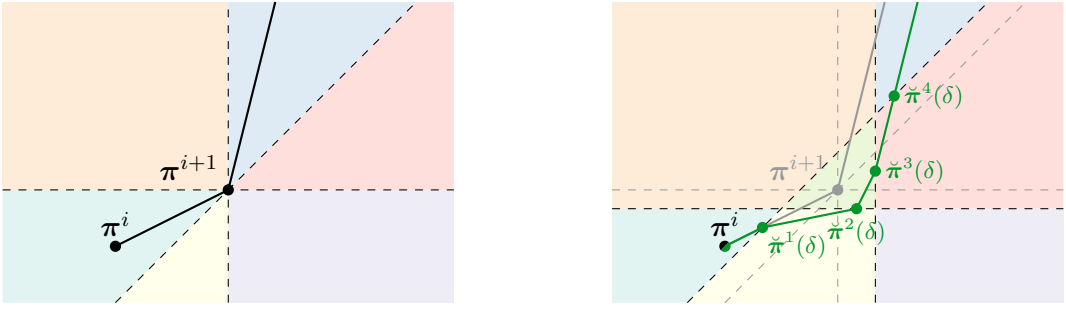

To derive a implicit procedure to track the perturbed solution curve, we observe that we can express the flows $\flowVp[l]$ and the distances to the boundaries $\epsilonp[l]$ as affine linear functions depending on the perturbation vector $\vec{\delta}$. Similar to Lemma~\ref{lem:perturbedValues}, we state
\begin{lemma} \label{lem:flowPertubationValues}
There are matrices $\vec{M}^l \in \R^{m \times m}$ such that $\flowVp[0] = \flowV^i + \vec{M}^0 \, \vec{\delta}$ and $\flowVp[l] = \flowV^{i+1} + \vec{M}^l \, \vec{\delta}$ for all $1 \leq l \leq \alpha$. Further, there are vectors $\vec{m}_{l,e} \in \bar{\R}^m$ such that $\epsilonp[0] = \epsilon(e) + \vec{m}_{l,e}^{\top} \, \vec{\delta}$ and $\epsilonp[l] = \vec{m}_{l,e}^{\top} \, \vec{\delta}$ for all $1 \leq l \leq \alpha$ and $e \in E^*$.
\end{lemma}
\begin{proof}
For $0 \leq l \leq \alpha$, we define the matrices
\[
\tilde{\vec{M}}^l := \vec{I}_m - \frac{1}{\flowdp^l_{e^*_{l-1}}} \flowDp[l-1] \, \vec{u}_{e^*_{l-1}}^{\top}
\quad \text{for } 1 \leq l \leq \alpha
\]
and $\vec{M}^l = \prod_{j=0}^{l-1} \tilde{\vec{M}}^{l-j}$ where $\vec{u}^T_{e}$ is the unit vector with a $1$ in the component corresponding to the edge $e$. (For $l=0$ the empty product is defined as $\vec{M}^0 = \vec{I}_m$.) Furthermore, we define the vectors
\[
\vec{m}_{l,e}^{\top} := - \frac{1}{\flowdp[l]} \vec{u}_e^{\top} \vec{M}^l 
\quad \text{for } 1 \leq l \leq \alpha
\]
where the vector $\vec{m}_{l,e}$ is set to infinity in every component if $\flowdp[l] = 0$.\\
For $l=0$, we get $\flowVp[0] = \flowV^i + \vec{M}^0 \vec{\delta}$ by definition and obtain
\[
\epsilonp[0]
= \frac{\flowBPd_e - \flowp[0]}{\flowdp[0]} = \frac{\flowBPd_e - x^i_e}{\Delta x^i_e} - \frac{1}{\flowdp[0]} \vec{u}_{e}^{\top} \, \vec{M}^0 \, \vec{\delta}
= \epsilon(e) + \vec{m}_{0,e}^{\top} \, \vec{\delta}
\]
where we used that $\flowDp[0] = \Delta \flowV^i$. \\
For $l=1$, we obtain
\begin{align*}
\flowVp[1] &= \flowVp[0] + \epsilonp^0 (e^*_0, \delta) \, \flowDp[0] \\
&= \vec{x}^{i} + \epsilon(e) \, \Delta \flowV^i + \vec{M}^0 \, \vec{\delta} + \vec{m}_{0, e_0^*}^{\top} \, \vec{\delta} \, \flowDp[0] \\
&= \vec{x}^{i+1} + \vec{M}^0 \, \vec{\delta}  + \flowDp[0] \, \vec{m}_{0, e_0^*}^{\top} \, \vec{\delta} \\
&= \vec{x}^{i+1} + \tilde{\vec{M}}^1 \, \vec{M}^0 \, \vec{\delta} = \vec{x}^{i+1} + \vec{M}^1 \, \vec{\delta}
\end{align*}
and with $\flowBPd_e = x^{i+1}_e$ also 
$
\epsilonp[1] = \frac{\flowBPd_e - \flowp[1]}{\flowdp[1]} = \frac{\flowBPd_e - x^{i+1}_e - \vec{u}_e^{\top} \vec{M}^1 \vec{\delta}}{\flowdp[1]}
= \vec{m}_{1,e}^{\top} \, \vec{\delta}
$. \\
By induction, we finally get for $l \geq 2$ that
\begin{align*}
\flowVp[l] &= \flowVp[l-1] + \epsilonp^{l-1} (e^*_{l-1}, \delta) \, \flowDp[l-1] \\
&= \vec{x}^{i+1} + \vec{M}^{l-1} \, \vec{\delta} + \vec{m}_{l-1, e_{l-1}^*}^{\top} \, \vec{\delta} \, \flowDp[l-1] \\
&= \vec{x}^{i+1} + \vec{M}^{l-1} \, \vec{\delta} + \flowDp[l-1] \, \vec{m}_{l-1, e_{l-1}^*}^{\top} \, \vec{\delta}  \\
&= \vec{x}^{i+1} + \tilde{\vec{M}}^l \, \vec{M}^{l-1} \, \vec{\delta} = \vec{x}^{i+1} + \vec{M}^l \, \vec{\delta}
\end{align*}
and also 
$
\epsilonp[l] = \frac{\flowBPd_e - \flowp[l]}{\flowdp[l]} = \frac{\flowBPd_e - x^{i+1}_e - \vec{u}_e^{\top} \vec{M}^l \vec{\delta}}{\flowdp[l]}
= \vec{m}_{l,e}^{\top} \, \vec{\delta}
$.
\end{proof}
As before, in every step $l$, the boundaries lying in front of the perturbed solution curve are the ones induced by edges $e \in E^*_l := \{ e \in E : 0 < \epsilonp[l] < \infty\}$. For $\delta > 0$ small enough, we have $E^*_0 = E^*$ (since $\epsilon(e) > 0$ for all $e \in E^*$) and we can use Lemma~\ref{lem:flowPertubationValues} to obtain $E^*_l = \{ e \in E^* : \vec{0} \lexle \vec{m}_{l,e} \lexle \vec{\infty} \}$. Finding the next boundary to cross thus means finding the lexicographically smallest vector in $\{ \vec{m}_{l,e} : e \in E^*_l \}$. As the next theorem shows, there is always a unique, lexicographic minimum in this set.

\begin{theorem} \label{thm:flowPerturbationUniqueMinimum}
For every $0 \leq l \leq \alpha$ there is a unique lexicographic minimum in $\{ \vec{m}_{l,e} : e \in E^*_l \}$.
\end{theorem}
\begin{proof}
If there are not at least two edges in $E^*_l$ then there is nothing to show. Otherwise, we need to show that for any two edges $e_1, e_2 \in E^*_l$ we have $\vec{m}_{l, e_1} \neq \vec{m}_{l, e_2}$. For $l=0$, the vectors $\vec{m}_{0,e}$ are multiples of unit vectors so the claim follows immediately.\\
For $l \geq 1$, we claim that 
\[
\vec{v}^{\top} \vec{M}^l = \vec{0} 
\text{ if and only if }
\vec{v} = \beta \vec{u}_{e^*_{l-1}}^{\top}
\]
for some $\beta \in \R$. This can be shown as in the proof of Theorem~\ref{thm:perturbationUniqueMinimum} by using Lemma~\ref{lem:flowsBoundaryCrossing} instead of Theorem~\ref{thm:boundaryCrossing}(\ref{thm:boundaryCrossing:Directions}). \\
Again, as in the proof of Theorem~\ref{thm:perturbationUniqueMinimum}, we can deduce with the help of Lemma~\ref{lem:flowsBoundaryCrossing} from the assumption $\vec{m}_{l, e_1} = \vec{m}_{l, e_2}$ that $\flowdp^{l-1}_{e^*_{l-1}} = 0$ which is a contradiction to the fact, that $e^*_{l-1}$ was the minimizer of $\epsilonp[l-1]$.
\end{proof}

}
\section{Multiple commodities}

In this section, we discuss the directed case with multiple commodities. While our approach admits multiple sources and sinks in general it is not suitable to model different commodities with different flows. We can, however, use a convex program to compute a direction vector $\Delta \vec X$ in a region.
Assume we have some multi-commodity Wardrop equilibrium flow $\vec X$ that is feasible for some demand encoded in the matrix $\vec Y$. This flow induces some total flow $\vec z = \vec X \, \vec{1}$. Furthermore, we have by Lemma~\ref{lem:characterisationWEPotentials} that there are potentials $\vec \pi_1, \vec \pi_2, \dotsc, \vec \pi_\numCommodities$ that satisfy \eqref{eq:lem:characterisationWEPotentials1} and \eqref{eq:lem:characterisationWEPotentials2}.

We then distinguish between three types of edges for every commodity $\commodity \in \setCommodities$. We call an edge $e = (v,w) \in E$ with $\mcFlow > 0$ \emph{used} by commodity $\commodity$, an edge $e = (v,w) \in E$ with $\vec \gamma_e^{\top} \vec \pi^\sCurveI_\commodity = l_e(\totalFlow_e)$ \emph{active} for commodity $\commodity$, and an edge $e = (v,w) \in E$ with $\smash{\vec \gamma_e^{\top} \vec \pi^\sCurveI_\commodity < l_e(\totalFlow_e)}$ \emph{inactive} for commodity $\commodity$. Every used edge is also active by Lemma~\ref{lem:characterisationWEPotentials}.

For notational convenience, let
$\indicUnused := \diag \big( \delta_{\commodity, e_1}^{\idUnused}, \dotsc, \delta_{\commodity, e_m}^{\idUnused} \big)$ 
for every commodity $\commodity \in \setCommodities$ where
$\delta_{\commodity, e}^{\idUnused} = 1$	 if $e$ is active but unused for commodity, and $\delta_{\commodity, e}^{\idUnused} = 0$ else. Further, let $\indicInactive := \diag \big( \delta_{\commodity, e_1}^{\idInactive}, \dotsc, \delta_{\commodity, e_m}^{\idInactive} \big)$ where $\delta_{\commodity, e}^{\idInactive} = 1$	if  $e$ is inactive for commodity~$  \commodity$ and $\delta_{\commodity, e}^{\idInactive} = 0$ else. 
We define the matrices $\vec A^+, \vec A^-$ as the diagonal matrices with the entries
$
\smash{a_e^+ := \lim_{\xi \downarrow \totalFlow_e^\sCurveI} \frac{l_e(\xi) - l_e( \totalFlow_e^\sCurveI )}{\xi - \totalFlow_e^\sCurveI}}
$
and
$
\smash{a_e^- := \lim_{\xi \uparrow \totalFlow_e^\sCurveI} \frac{l_e(\xi) - l_e( z_e^\sCurveI )}{\xi - \totalFlow_e^\sCurveI}}
$.
These matrices contain the slopes of the piecewise-linear cost function $l_e$ at $\totalFlow_e^i$. If $l_e$ is differentiable at $\totalFlow_e^\sCurveI$ (which is the case if $\totalFlow_e^\sCurveI$ is not a breakpoint of $l_e$) the values $a_e^+$ and $a_e^-$ coincide.
Consider the convex program 
\longversion{
\begin{align}
\min \quad &\frac{1}{2} \, \totalFlowVec^{\top} \vec A^+ \, \totalFlowVec^+ 
+
\frac{1}{2} \, \totalFlowVec^{\top} \vec A^- \, \totalFlowVec^- \notag
& 
\text{s.t. }\quad\quad \totalFlowVec &= \mcFlowVec \, \vec{1} \notag \\
&&\incidenceMatrix \cFlowVec &= \Delta \vec Y_\commodity \qquad \; \forall \commodity \in \setCommodities \label{eq:convexProgramMC} \\
&&\indicUnused \cFlowVec &\geq 0 \qquad \qquad \forall \commodity \in \setCommodities \notag \\
&&\indicInactive \cFlowVec &= 0 \qquad \qquad \forall \commodity \in \setCommodities \notag
\end{align}
}{
\begin{align}
\label{eq:convexProgramMC}
\begin{split}
\min \quad &\frac{1}{2} \, \totalFlowVec^{\top} \vec A^+ \, \totalFlowVec^+ 
+
\frac{1}{2} \, \totalFlowVec^{\top} \vec A^- \, \totalFlowVec^- \\
\text{s.t. } &\totalFlowVec = \mcFlowVec \, \vec{1}, \quad
\incidenceMatrix \cFlowVec = \Delta \vec Y_\commodity, \quad
\indicUnused \cFlowVec \geq 0, \quad
\indicInactive \cFlowVec = 0 \qquad \qquad \forall \commodity \in \setCommodities.
\end{split}
\end{align}
}
where $\totalFlowVec^+$ is the vector with the entries $\totalFlow_e^+ := \max \{ 0, z_e \}$ and $\totalFlowVec^-$ is the vector with the entries $\totalFlow_e^- := \min \{ 0, z_e \}$. Note that if no latency function is at a breakpoint in $\totalFlowVec^\sCurveI$ the matrices $\vec A^+$ and $\vec A^-$ coincide and the objective simplifies to $\frac{1}{2} \, \totalFlowVec^{\top} \vec A \, \totalFlowVec$. Furthermore, the objective function is convex and almost everywhere differentiable but in the points where $\totalFlow_e = 0$ for some $e \in E$.

Intuitively, this program computes a Wardrop equilibrium flow in the subgraph induced by the active edges since the last constraint ensures that there is no flow on the inactive edges. This equilibrium describes how the Wardrop equilibrium $\mcFlowVec$ changes if the demands change by $\Delta \vec Y$ and is therefore referred to as an \emph{augmenting equilibrium}. Furthermore, the program admits positive and negative flow on all used edges. A negative flow on some used edge can then be interpreted as a reduction of the flow on this edge.

\begin{theorem} \label{thm:convexProgramMC}
Let $\vec{X}$ be a Wardrop equilibrium for some demand $\vec Y$ and $(\hat{\vec X}, \hat{\vec z})$ an optimal solution of the convex program \eqref{eq:convexProgramMC}. Then there is $\epsilon > 0$ such that $\vec X (\xi) := \vec X + \xi \, \hat{\vec X}$
is a Wardrop equilibrium flow for every demand $\vec Y + \xi \, \Delta \vec Y$ with $0 \leq \xi \leq \epsilon$.
\end{theorem}
\begin{appendixproof}{Theorem~\ref{thm:convexProgramMC}}
Let $(\hat{\vec X}, \hat{\vec z})$ be an optimal solution. We denote by $\partial f$ the subdifferential of some function $f$. Then the KKT-conditions yield
\begin{align*}
\vec 0 &\in \partial_{\mcFlowVec, \totalFlowVec}  \frac{1}{2} \, \totalFlowVec^{\top} \vec A^+ \, \totalFlowVec^+ 
+
\partial_{\mcFlowVec, \totalFlowVec}  \frac{1}{2} \, \totalFlowVec^{\top} \vec A^- \, \totalFlowVec^-
+ \vec \lambda^{\top} \partial_{\mcFlowVec, \totalFlowVec} \vec (\mcFlowVec \vec 1 - \totalFlowVec) \\
&+ \sum_{\commodity \in \setCommodities} \vec \mu^{\top}_{\commodity} \partial_{\mcFlowVec, \totalFlowVec}  (\incidenceMatrix \cFlowVec - \Delta \vec Y_\commodity ) 
-\sum_{\commodity \in \setCommodities} \vec \nu_\commodity^{\top} \partial_{\mcFlowVec, \totalFlowVec} \indicUnused \cFlowVec
+ \sum_{\commodity \in \setCommodities} \vec \eta_\commodity^{\top} \partial_{\mcFlowVec, \totalFlowVec} \indicInactive \cFlowVec
\end{align*}
where $\vec \lambda, \vec \mu_{\commodity}, \vec \nu_{\commodity}, \vec \eta_{\commodity}$ are vectors of dual variables with $\vec \eta_\commodity \geq 0$ for all $\commodity \in \setCommodities$.
Note that all subdifferentials with respect to $\mcFlowVec$ only contain the gradient since the objective and all constraints are differentiable with respect to $\mcFlowVec$. The same holds for every subdifferential with respect to $z_e$ as long as $\totalFlow_e \neq 0$. Thus we get for all edges $e$ with $\hat{z}_e \neq 0$ that
\begin{align*}
\lambda_e &= a_e^+ \hat{z}_e \quad \text{ if } \hat{z}_e > 0\\ 
\text{ and } \lambda_e &= a_e^- \hat{z}_e \quad \text{ if } \hat{z}_e < 0
\end{align*}
and for all edges with $\hat{z}_e = 0$ that
\begin{align*}
\lambda_e &= \hat{a} \hat{z}_e = 0
\end{align*}
for some $\hat{a}$ between $a_e^-$ and $a_e^+$.
The subdifferential with respect to $\cFlowVec$ yields
\begin{align*}
\vec \lambda^{\top} + \vec \mu^{\top}_\commodity \incidenceMatrix + \nu^{\top}_\commodity \indicUnused + \eta_\commodity^{\top} \indicInactive = 0
\end{align*}
which implies for every edge $e =(v,w) \in E$ that
\begin{align*}
\lambda_e &= \mu_{\commodity, w} - \mu_{\commodity, v} \text{ if } e \text{ is used,}\\
\lambda_e &= \mu_{\commodity, w} - \mu_{\commodity, v} + \nu_{\commodity, e} \text{ if } e \text{ is active but unused,} \\
\lambda_e &= \mu_{\commodity, w} - \mu_{\commodity, v} - \eta_{\commodity, e} \text{ if } e \text{ is inactive.}
\end{align*}
Since $\nu_{\commodity, e} = 0$ if $\mcFlow \neq 0$ by the KKT-conditions, we have that
\begin{align*}
\lambda_e = \mu_{\commodity, w} - \mu_{\commodity, v} = \vec \gamma_e \vec \mu_\commodity
\end{align*}
for all active edges $e = (v,w) \in E$ with $\hat{x}_{e,j} \neq 0$ and
\begin{align*}
\lambda_e  \geq \mu_{\commodity, w} - \mu_{\commodity, v} = \vec \gamma_e \vec \mu_\commodity
\end{align*}
for all active edges $e = (v,w) \in E$ with $\hat{x}_{e,j} = 0$. So the flow $\hat{x}_{e,j}$ is a Wardrop equilibrium on the active edges with the potential $\vec \mu_\commodity$ for every commodity $\commodity$.

For $\xi \geq 0$ we define the multi commodity flow 
$
\vec X (\xi) := \vec X + \xi \, \hat{\vec X}
$
with the total flow
$
\vec z (\xi) = \vec z + \xi \, \hat{\vec z}
$
that is feasible for the demand $\vec Y (\xi) = \vec Y + \xi \, \Delta \vec Y$. Further, we define for every $\commodity \in \setCommodities$ the potential 
$
\vec \pi_\commodity (\xi) := \vec \pi_\commodity + \xi \, \vec \mu_\commodity.
$
Then we have for every $\commodity \in \setCommodities$ and every active edge $e \in E$ that
\begin{align*}
\vec \gamma_e^{\top} \vec \pi_\commodity (\xi) &= \vec \gamma_e^{\top} \big( \vec \pi_\commodity + \xi \vec \mu_\commodity \big) \\
&\leq a_e z_e + b_e + \xi \lambda_e \\
&= l_e (z_e + \xi \hat{z}_e) = l_e (z_e (\xi) )
.
\end{align*}
where $a_e = a_e^+$ if $\hat{z}_e \geq 0$ and $a_e = a_e^-$ if $\hat{z}_e < 0$.
For every active edges with $\hat{x}_{e,j} \neq 0$ we obtain
\begin{align*}
\vec \gamma_e^{\top} \vec \pi_\commodity (\xi) &= \vec \gamma_e^{\top} \big( \vec \pi_\commodity + \xi \vec \mu_\commodity \big) \\
&= a_e z_e + b_e + \xi \lambda_e \\
&= l_e (z_e + \xi \hat{z}_e) = l_e (z_e (\xi) )
.
\end{align*}

For all inactive edges and every commodity we have $\vec \gamma_e \vec \pi_\commodity < l_e (\totalFlow_e^0)$. For every commodity $j$ and every inactive edge $e$ we have $\vec \gamma_e^{\top} \vec \pi_j < l_e (z_e)$. Thus, by continuity of the cost functions $l_e$ we know that there is some $\epsilon_1 > 0$ such that
\begin{align*}
\vec \gamma_e^{\top} \vec \pi (\xi)_j &\leq l_e(z_e (\xi)) \quad \text{ for } \xi \leq \epsilon_1
\end{align*}
for all inactive edges and all commodities. Since $\hat{x}_{e,j}$ can only be negative if $x_{e,j}$ is positive there is an $\epsilon_2$ such that $x_{e,j} (\xi) \geq 0$ for all $0 \leq \xi \leq \epsilon_2$ Thus, by Lemma~\ref{lem:characterisationWEPotentials}, for every $0 \leq \xi \leq \epsilon := \min\{\epsilon_1, \epsilon_2\}$ the flow $\mcFlowVec^\xi$ is a Wardrop equilibrium for the demand $\vec Y^{\xi}$.
\end{appendixproof}

By Theorem~\ref{thm:convexProgramMC}, we can still use the same methodology for multi-commodity networks as for single-commodity networks except that the direction vectors $\Delta \vec X$ are computed in each region by a convex program.


\appendix

\includeappendixcollection{proof}{Missing Proofs}
\includeappendixcollection{figure}{Missing Figures}
\includeappendixcollection{example}{Missing Examples}

\longversiononly{
	\section{The Laplacian Matrix} \label{app:laplacian}

Given some directed graph, we consider the incidence matrix
\[
\vec{\Gamma} \in \R^{m,n}
\qquad \text{with} \qquad
\gamma_{e,v} =
\begin{cases}
1 & \text{if edge $e$ enters vertex $v$,} \\
-1 & \text{if edge $e$ leaves vertex $v$,} \\
0 & \text{otherwise.}
\end{cases}
\]
The matrix $\vec{L}' := \vec{\Gamma} \vec{\Gamma}^{\top}$ is called the \emph{(unweighted) Laplacian Matrix} of the Graph $G$. Given some edge weights $c_{e_1}, \dotsc, c_{e_m}$ we define the matrix $\vec{C} := \diag\big( c_{e_1}, \dotsc, c_{e_m} \big)$. Then the matrix $\vec{L} := \vec{\Gamma} \vec{C} \vec{\Gamma}^{\top}$ is called the  \emph{(weighted) Laplacian Matrix} of $G$.  Both the weighted and the unweighted Laplacians contain all essential information on the vertex-edge-relations (e.g. $\vec{L}' = \vec{\Delta} - \vec{A}$ where $\vec{\Delta}$ is the diagonal matrix containing the vertex degrees and $\vec{A}$ is the adjacency matrix of $G$). The matrices appear in many in many appliciation such as electrical networks (see \cite{seshu1961linear} for a reference) and random walks in graphs (as presented in \cite{doyle1984random}) are thus studied extensively (see, e.g., \cite{anderson1985eigenvalues, grone1991geometry, merris1994laplacian, mohar1991laplacian}).
The following theorem states some well-known, basic statements about the Laplacian matrix.

\begin{lemma} \label{lem:appendixLaplacianMatrix}
Given some non-negative edge weights $c_{e_1}, \dotsc, c_{e_m}$, the weighted Lapalacian matrix $\vec{L}$ has the following properties
\begin{compactenum}[(i)]
\item 
	The rank of $\vec L$ is $n- n_C$ where $n_C$ is the number of connected components when two vertices are considered connected if they are connected by an edge with $c_{e,t_e} \neq 0$.
\item 
	The matrix $\vec L$ is positive semi-definite.
\item 
	The row sum and column sum of $\vec L$ is zero for every row or column.
\item 
	If $G$ is connected then the reduced Laplacian matrix $\hat{\vec L}$ obtained from the matrix $\vec L$ by deleting the first row and column is positive definite and inverse-positive, i.e. $\hat{\vec{L}}_{\vec t}^{-1} \geq 0$.
\end{compactenum}
\end{lemma}
\begin{proof}
For statement $(i)$ see e.g. \cite[Theorem~1]{anderson1985eigenvalues}. Let $\vec{C}^*$ be the diagonal matrix containing the values $\sqrt{c_e}$. Then, for any vector $\vec{x} \in \R^n$, $\vec{x}^\top \vec{L} \, \vec{x} = (\vec{C}^* \vec{\Gamma}^\top \vec{x})^{\top} \, (\vec{C}^* \vec{\Gamma}^\top \vec{x}) = \Vert \vec{C}^* \vec{\Gamma} ^\top \vec{x} \Vert^2 \geq 0$. Thus $(ii)$ follows. Since $\vec{\Gamma}$ has zero row sum by definition, $(iii)$ follows directly. Finally, the first three statements imply that $\hat{\vec{L}}$ has full rank and is positive definite. The matrix $\vec{L}$ (and thus also $\hat{\vec{L}}$)  has only non-positive off-diagonal entries and all eigenvalues are non-negative. This means, that $\hat{\vec{L}}$ is a (non-singular) $M$-matrix and thus inverse-positive (see \cite{plemmons1977m} for a survey on properties of $M$-matrices). 
\end{proof}

\subsection{The (Pseudo-)Inverse of the Laplacian Matrix} \label{app:laplacian:inverse}

As explained in Section~\ref{sec:basic}, in every region $\vec{t}$, the Laplacian matrix $\vec{L} = \vec{\Gamma} \vec{C}_{\vec{t}} \vec{\Gamma}^{\top}$ maps any potential vector $\vec{\pi}$ to the excess vector $\vec{y} = \vec{L} \, \vec{\pi}$. Since we are interested in the opposite direction---we want a potential $\vec{\pi}$ for every demand vector $\lambda \cdot \Delta \vec{y}$---we need the inverse of $\vec{L}$. By Lemma~\ref{lem:LaplacianMatrix} we know that the rank is $n-1$ (if we assume that the Graph $G$ is connected) and the kernel of $\vec{L}$ is $\ker (\vec{L}) =  \vspan ( \vec{1} )$, where $\vec{1}$ is the all one vector. So it is not possible to find an inverse on $\R^n$. If, however, we restrict the matrix $\vec{L}$ to the subspace $\vec{R}^{n-1}$ by deleting the first row and column, we obtain the non-singular matrix $\hat{\vec{L}}^{-1}$.

For a more formal approach, we introduce the notion of a pseudo inverse.
\begin{definition}
For any matrix $\vec{A} \in \R^{n \times m}$ we say $\vec{X}$ is a \emph{pseudo inverse} of $\vec{A}$ if
\begin{equation} \label{eq:pseudoInverse}
\vec{A} \vec{X} \vec{A} = \vec{A}
\quad
\text{ and }
\quad
\vec{X} \vec{A} \vec{X} = \vec{X}
.
\end{equation}
If, further, it holds that 
\begin{equation} \label{eq:moorePenroseInverse}
(\vec{A} \vec{X} )^{\top} = \vec{A} \vec{X}
\qquad
\text{ and }
(\vec{X} \vec{A} )^{\top} = \vec{X} \vec{A}
\end{equation}
we call $\vec{X}$ the \emph{Moore-Penrose-Inverse} of $\vec{A}$ and write $\vec{A}^{\dagger} := \vec{X}$.
\end{definition}

As shown by Penrose in  \cite{penrose1955generalized}, the Moore-Penrose-Inverse exists and is unique. For  more general analysis on pseudo inverses, see e.g. \cite{ben2003generalized}. For the Moore-Penrose-Inverse of the Laplacian matrix, we in particular obtain
\[
\vec{L} \vec{L}^{\dagger} = \vec{L}^{\dagger} \vec{L} = \vec{Q}
\]
where $\vec{Q} := \vec{I}_n - \frac{1}{n} \, \vec{1} \, \vec{1}^{\top}$ is the projection on the subspace orthogonal to the kernel of $\vec{L}$. Further, the Moore-Penrose-Inverse can be obtained explicitly as
\begin{equation} \label{eq:laplaceMoorePenroseInverse}
\vec{L}^{\dagger} = \big(\vec{L} + \vec{1} \, \vec{1}^{\top} \big)^{-1} - \frac{1}{n^2} \, \vec{1} \, \vec{1}^{\top}
\end{equation}
in this case (see \cite[p. 79]{ben2003generalized}).
So $\vec{L}^{\dagger}$ maps excess vectors $\vec{y}$ to potential vectors $\vec{\pi}$ that are orthogonal to $\vec{1}$, i.e. to potential vectors with zero-sum. Although the Moore-Penrose-Inverse has nice properties (as uniqueness) and is commonly used for the inverse of the Lapalacian (e.g. in \cite{klein1993resistance}), we will use another pseudo inverse for two reasons: First, equation \eqref{eq:laplaceMoorePenroseInverse} makes it difficult to analyze properties of $\vec{L}^{\dagger}$. Second, the potentials obtained by using $\vec{L}^{\dagger}$ are zero sum but do not satisfy $\pi_s = 0$. This makes it harder to interpet these values as travel times from the sink to any vertex $v$.

This is why we use another, natural approach. We define the matrix
\[
\vec{L}^* := 
\begin{bmatrix}
0 & \vec{0}^{\top} \\
\vec{0} & \hat{\vec{L}}^{-1}
\end{bmatrix}
\] 
where $\vec{0}$ is the zero vector in $\R^{n-1}$. It is easy to verify that $\vec{L}^*$ is a pseudo inverse in the sense  of the conditions \eqref{eq:pseudoInverse}. In particular, $\vec{L}^* \vec{L} = ( \vec{L} \vec{L}^* )^{\top} = \tilde{\vec{Q}}$ where $\tilde{\vec{Q}} := \vec{I}_n - \vec{1} \vec{u}_1^{\top}$ is the projection on the subspace $\{ \vec{v} \in \R^n : v_1 = 0 \}$ and $\vec{u}_1$ is the first unit vector. So, $\vec{L}^*$ maps excess vectors $\vec{y}$ to potentials $\vec{\pi}$ with $\pi_1 = 0$. 

Overall, $\vec{L}^*$ is a easy to obtain pseudo inverse of $\vec{L}$ that further maps to a space of potentials that are fixed to $\pi_1 = 0$. All relevant information is kept in the matrix $\hat{\vec{L}}$, thus it is sufficient to store and govern this matrix.

\subsection{Effective Resistances/Effective Travel Times}

In this subsection, we state some well-known properties of the effective resistance.
\begin{definition}
For every pair of vertices $v,w \in V$, define the vector $\vec{\gamma}_{v,w} \in \R^n$ as the vector with $\gamma_{v,w,u} = 1$ if $u = w$, $\gamma_{v,w,u} = -1$ if $u=v$, and $\gamma_{v,w,u} = 0$ otherwise. Thus, $\vec{\gamma}_{v,w}$ is the incidence vector of a (possibly non existent) edge between the vertices $v$ and $w$. We then define by
\[
\Omega_{v,w} := \vec{\gamma}_{v,w}^{\top} \, \vec{L}^* \, \vec{\gamma}_{v,w}
\]
the \emph{effective resistance} or \emph{effective travel time} between the vertices $v$ and $w$.
\end{definition}
Note that the effective travel time can also be computed as $\Omega_{v,w} = \hat{\vec{\gamma}}_{v,w}^{\top} \hat{\vec{L}}^{-1} \hat{\vec{\gamma}}_{v,w} = \vec{\gamma}_{v,w}^{\top} \, \vec{L}^{\dagger} \, \vec{\gamma}_{v,w}$. The quantity $\Omega_{v,w}$ measures the potential difference $\pi_w - \pi_v$ when we send a unit flow from $v$ to $w$ and can thus be interpreted as the travel time between $v$ and $w$ if a unit flow is routed between these two vertices. In electrical networks, this quantity is known as the \emph{effective resistance}.

We can state the following properties of the effective travel times.
\begin{theorem} \label{thm:effectiveTravelTime}
The effective travel time satisfies the following properties.
\begin{compactenum}[(i)]
\item \label{thm:effectiveTravelTime:metric}
	The effective travel time is a metric on $G$, i.e. for all $u,v,w \in V$ we have
	\begin{align*}
	\Omega_{v,w} &\geq 0 \text{ and } \Omega_{v,w} = 0 \Leftrightarrow v=w, \\
	\Omega_{v,w} &= \Omega_{w,v}, \\
	\Omega_{v,w} &\leq \Omega_{v,u} + \Omega_{u,w}.
	\end{align*}
\item \label{thm:effectiveTravelTime:distance}
	For every $v$-$w$ path $P$ in $G$, we have $\Omega_{v,w} \leq \sum_{e \in P} a_e$ with equality if and only if there is a unique path between $v$ and $w$. In particular, for any edge $e = (v,w) \in E$, we have $\Omega_{v,w} \leq a_e$.
\item \label{thm:effectiveTravelTime:monotonicity}
	The effective travel times are monotone in the (slopes of the) cost functions $l_e$. More precisely
	\begin{align*}
	\frac{\partial}{\partial a_e} \Omega_{v,w} &\geq 0 \quad \text{for all }v,w \in V \text{ and } e \in E.
	\end{align*}
	In particular, if $e$ is not contained in any $v$-$w$ path, then $\frac{\partial}{\partial a_e} \Omega_{v,w} = 
0$.
\end{compactenum}
\end{theorem}
\begin{proof}
See \cite{klein1993resistance}.
\end{proof}
The first two statements of Theorem~\ref{thm:effectiveTravelTime} show that the effective travel times are a distance measure on the graph $G$, where the distance $\Omega_{v,w}$ is always less or equal then the length of a shortest path with respect to the edge lengths $a_e$.

The third part of Theorem~\ref{thm:effectiveTravelTime} is also known as \emph{Rayleighs Monotonicity Law} (see for example \cite{doyle1984random}). This implies that increasing the slope of some cost function increases the overall travel time. In particular, removing edges (by setting $a_e = \infty$) increases the travel time for all users. This might seem as a contradiction to the Braess paradox where removing an edge actually improves the travel times for all users. To resolve this inconsistency, note that the effective travel times are travel times with respect to the strictly linear cost functions $l_e(x) = a_e x$ while for the Braess paradox we need edges with cost functions $l_e(x) = a_e x + b_e$ with $b_e > 0$. These offsets (or in general non-linearity) and the discontinuity of the cost functions at $x = 0$ (to model the direction of the edges) are necessary to observe the effects of the Braess paradox.

However, the results of Theorem~\ref{thm:effectiveTravelTime} apply to the potential $\Delta \hat{\vec{\pi}} = \hat{\vec{L}}_{\vec{t}}^{-1} \Delta \hat{\vec{y}}$ and the induced flow $\Delta \vec{x} = \vec{C}_{\vec{t}} \vec{\Gamma}^{\top} \Delta \vec{\pi}$. So the change of the Wardrop equilibrium in every region $R_{\vec{t}}$ is not affected by the Braess paradox---increasing a slope of some cost function (or removing edges) will always result in an increase in $\Delta \pi_t - \Delta \pi_s$ and thus increase the rate at which the travel time increases. 
}

\section{Nested Braess Networks}
\label{app:braess}

Let $\vec x$ be a directed Wardrop equilibrium for some demand rate $r \geq 0$. Then we call the set $S(\vec x) := \{ e \in E : x_e > 0 \}$ the support set of the flow. Flows with two distinct support sets must be induced by potentials in different regions $R_{\vec t^1}, R_{\vec t^2}$ in the potential space. So, the number of different support sets is a lower bound to the number of function parts of the Wardrop equilibrium functions and thus a lower bound for the number of iterations that are needed for their computation.
 
We construct a family of directed single-commodity networks $(\mathcal{B}^{(j)})_{j \in \N}$ such that every network has $2j+2$ vertices and $4j+1$ edges but admits $2^{j+1}$ different support sets.

For $j \in \N$ we define a set of nodes $V_j := \{s = v_0, v_1, \dotsc, v_{2j}, t = v_{2j+1}\}$ and edge sets
\begin{align*}
E^{(1)}_j &:= \left\{ (v_i, v_{i+1}) : i \in \{0,\dotsc,2j\}, i \neq j  \right\} \\
E^{(2)}_j &:= \left\{ (v_{i}, v_{2j-i}) : i \in \{0,\dotsc,j-1\} \right\}  \\
E^{(3)}_j &:= \left\{ (v_{i+1}, v_{2j+1-i}) : i \in \{0,\dotsc,j-1\} \right\}  \\
E_j &:= E^{(1)}_j \cup E^{(2)}_j \cup E^{(3)}_j \cup \{ (v_j, v_{j+1}) \} 
\end{align*}
and call the graph $G^{B,j}$ the $j$-th nested Braess graph. Further, we define coefficients
\begin{align*}
a^j_e &=
\begin{cases}
1 &\text{if } e \in E_j^{(1)} \\
0 &\text{else}
\end{cases} \quad
b^j_e =
\begin{cases}
10^{j-1-i} &\text{if } e = (v_{i}, v_{2j-i}) \in E^{(2)}_j, i \in \{0, ..., j-1\} \\
10^{i} &\text{if } e = (v_{i+1}, v_{2j+1-i}) \in E^{(3)}_j, i \in \{0, ..., j-1\} \\
0 &\text{else}
\end{cases}
\end{align*}
of the cost functions $l^j_e(x) = a^j_e x_e + b^j_e$ and one commodity with $s=v_0$ and $t=v_{2l+1}$ and a demand rate of $q \geq 0$. Then the network $\mathcal{B}^{(j)}$ consisting of the graph $G^{B,j}$ with the cost functions
$l^j$ is called the $j$-th nested Braess network.

\begin{figure}[t]
\begin{center}
\begin{subfigure}[b]{.35\textwidth}
\begin{center}
\begin{tikzpicture}
\useasboundingbox (0,-2.25) rectangle (4,2.25);

\draw (0,0) node[solid] (1) {} node[left] {$s$};
\draw (2,1) node[solid] (2) {} node[above] {$v_2$};
\draw (2,-1) node[solid] (3) {} node[below] {$v_1$};
\draw (4,0) node[solid] (4) {} node[right] {$t$};

\draw[->] (1) -- (2) node[midway, above] {$1$};
\draw[->] (1) -- (3) node[midway, below] {$x$};
\draw[->] (3) -- (2) node[midway, left] {$0$};
\draw[->] (2) -- (4) node[midway, above] {$x$};
\draw[->] (3) -- (4) node[midway, below] {$1$};
\end{tikzpicture}
\caption{The Braess network $\mathcal{B}^{(1)}$.}
\end{center}
\end{subfigure}
\begin{subfigure}[b]{.63\textwidth}
\begin{center}
\begin{tikzpicture}
\useasboundingbox (0,-2.25) rectangle (8,2.25);

\draw (0,0) node[solid] (1)  {} node[left] {$s$};
\draw (4,2) node[solid] (2) {} node[above] {$v_6$};
\draw (4,-2) node[solid] (3) {} node[below] {$v_1$};
\draw (8,0) node[solid] (4) {} node[right] {$t$};

\draw (2,0) node[solid] (5) {} node[left] {$v_2$};
\draw (6,0) node[solid] (6) {} node[right] {$v_5$};

\draw (4,-.75) node[solid] (7) {} node[below] {$v_3$};
\draw (4,.75) node[solid] (8) {} node[above] {$v_4$};

\draw[->] (1) -- (2) node[midway, above left] {$100$};
\draw[->] (1) -- (3) node[midway, below left] {$x$};

\draw[->] (2) -- (4) node[midway, above right] {$x$};
\draw[->] (3) -- (4) node[midway, below right] {$100$};

\draw[->] (3) -- (5) node[midway, left] {$x$};
\draw[->] (3) -- (6) node[midway, right] {$10$};

\draw[->] (5) -- (2) node[midway, left] {$10$};
\draw[->] (6) -- (2) node[midway, right] {$x$};

\draw[->] (5) -- (7) node[midway, below] {$x$};
\draw[->] (5) -- (8) node[midway, above] {$1$};
\draw[->] (7) -- (8) node[midway, left] {$0$};
\draw[->] (7) -- (6) node[midway, below] {$1$};
\draw[->] (8) -- (6) node[midway, above] {$x$};
\end{tikzpicture}
\caption{The third nested Braess network $\mathcal{B}^{(3)}$.}
\end{center}
\end{subfigure}
\end{center}
\caption{Two nested Braess networks for $j=1$ and $j=3$ with cost functions on the edges.}
\label{fig:Braess}
\end{figure}
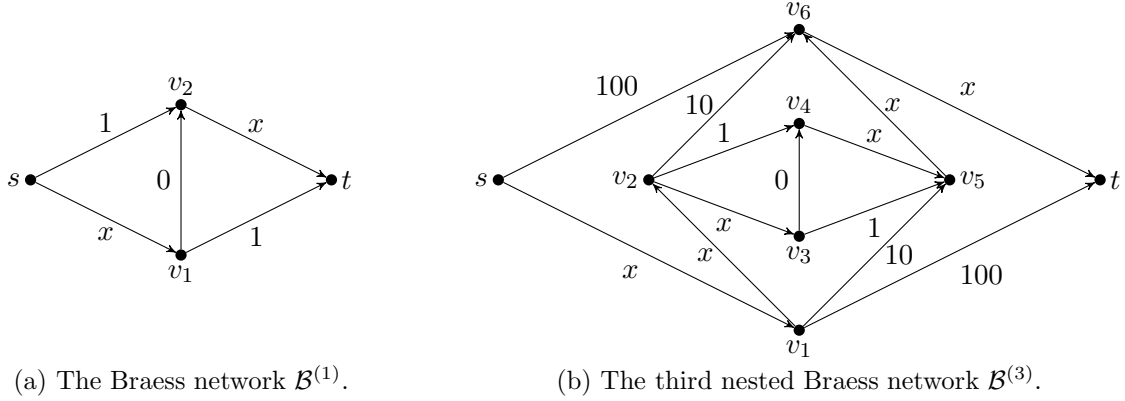

For $j=1$ the network $\mathcal{B}^{(1)}$ is just the classical Braess network as depicted in figure \ref{fig:Braess}(a). Figure \ref{fig:Braess}(b) shows the third nested Braess network. 
It is easy to see that for every $j \geq 2$ and every $p \in \{1, \dotsc, j-1\}$ the subnetwork of $\mathcal{B}^{(j)}$ induced by the node set 
\[
V^{(j)}_{p} := \left\{ v_i : i \in \{p, \dotsc, 2j+1 - p\} \right\}
\]
is again a nested Braess network, namely $\mathcal{B}^{(j-p)}$.

\begin{lemma} \label{lem:nestedBraessHighDemand}
Let $j \in \N$ be fixed. Then for every demand rate $q \geq 2 \cdot 10^{j-1}$ we have that 
\[
x_e (q):=
\begin{cases}
\nicefrac{q}{2} &\text{if } e \in \{ (v_0, v_1), (v_0, v_{2l}), (v_1, v_{2l+1}), (v_{2l}, v_{2l+1}) \} \\
0 &\text{else}
\end{cases}
\]
is a Wardrop equilibrium flow in $\mathcal{B}^{(j)}$ and the travel time from $s$ to $t$ is
\begin{equation} \label{eq:nestedBraessJointLatency}
\pi_t (q) =  10^{j-1} + \frac{q}{2}
.
\end{equation}
This means for sufficiently high demand, all flow is only routed over the outer edges and none of the inner components is used.
\end{lemma}
\begin{proof}[Proof of Lemma~\ref{lem:nestedBraessHighDemand}]
We see that the outer path $P_1 =\{ (v_0, v_1), (v_1, v_{2l+1}) \}$ and the outer path $P_2 =\{ (v_0, v_{2l}), (v_{2l}, v_{2l+1}) \}$ both have the same cost $l_{P_i} (x(q)) = 10^{j-1} + \nicefrac{q}{2}$. By the definition of the network, any other path $Q \in \mathcal{P}$ must contain the edges $(v_0, v_1)$ and $(v_{2l}, v_{2l+1})$. This we have for these paths $Q$ that
\begin{align*}
l_Q (x(q)) &\geq l_{(v_0, v_1)} \left( x_{(v_0, v_1)} (r) \right) + l_{(v_{2l}, v_{2l+1})} \left( x_{(v_{2l}, v_{2l+1})} (q) \right) \\
	&= \frac{q}{2} + \frac{q}{2} \geq 10^{j-1} + \frac{r}{2} = l_{P_i} (x(r))
\end{align*}
and thus $x_e (q)$ is a Wardrop equilibrium with $\pi_t (q) =  10^{j-1} + \nicefrac{q}{2}$.
\end{proof}

\begin{lemma} \label{lem:nestedBraessUsedSets}
Let $j \in \N$ be fixed. Then there are $2^{j+1}$ different demand rates $q_1, \dotsc, q_{2^{j+1}} \in [0, 3 \cdot 10^{j-1})$ such that all support sets $S^{q_1}, \dotsc, S^{q_{2^{j+1}}}$ of the Wardrop equilibrium flows are different.
\end{lemma}
\begin{proof}[Proof of Lemma~\ref{lem:nestedBraessUsedSets}]
We prove this by induction on $j$.

For $j=1$ this is easy to verify. For $q_1 = 0$ no edges are used, i.e. $S^{q_1} = \emptyset$. For $q_2 \in (0, 1]$ only the inner path is used, i.e. $S^{q_2} = \{ (v_0, v_1), (v_1, v_2), (v_2, v_3) \}$. For $q_3 \in (1, 2]$ all edges are used, i.e. $S^{q_2} = E$ and finally for $q_4 > 2$ all edges but the inner edge is used, i.e. $S^{q_4} = E \setminus \{(v_1, v_2)\}$.

Now let $j \geq 2$. Then there are essentially three paths in the network: The outer paths
$P_1 =\{ (v_0, v_1), (v_1, v_{2j+1}) \}$
and
$P_2 =\{ (v_0, v_{2j}), (v_{2j}, v_{2j+1}) \}$,
and the path
$P_3 = \{ (v_0, v_1), e_{\text{in}}, (v_{2j}, v_{2j+1}) \}$ where $e_{\text{in}}$ is a placeholder for the $j-1$-th nested Braess network.

The latency function of the inner network is increasing, thus for $q \leq 3 \cdot 10^{j-2}$ we have
\[
l_{e_{\text{in}}} (x_{e_{\text{in}}}) \leq l_{e_{\text{in}}} (q) \leq l_{e_{\text{in}}} (3 \cdot 10^{j-2}) \stackrel{\eqref{eq:nestedBraessJointLatency}}{=} 10^{j-2} + \frac{3 \cdot 10^{j-2}}{2} = \frac{5 \cdot 10^{j-2}}{2}
\]
and $l_{v_0,v_1} ( x_{v_0,v_1} ) \leq q < \frac{6 \cdot 10^{j-2}}{2}$ and $l_{v_{2l}, v_{2l+1}} ( x_{v_{2l}, v_{2l+1}} ) \leq q \leq \frac{6 \cdot 10^{j-2}}{2}$ and thus every path using the inner component has latency of at most
\[
2 \cdot \frac{6 \cdot 10^{j-2}}{2} +  \frac{5 \cdot 10^{j-2}}{2} = \frac{17 \cdot 10^{j-2}}{2} < 10^{j-1} = l_{v_0, v_{2j}} ( x_{v_0, v_{2j}} ) = l_{v_, v_{2j+1}} (x_{v_1, v_{2j+1}})
.
\]
This means for every demand rate $q \leq 3 \cdot 10^{j-2}$ all demand is routed through the inner network over the path $\{ (v_0, v_1), e_{\text{in}}, (v_{2j}, v_{2j+1}) \}$.

By induction hypothesis, there are $2^{j}$ different sets of used edges $\emptyset = \tilde{S}_1, \tilde{S}_2, \dotsc, \tilde{S}_{2^j}$ for different demand rates $\tilde{q}_1, \dotsc, \tilde{q}_{2^j} \in [0, 3 \cdot 10^{j-2})$ in the $j-1$-th nested Braess network which is embedded in $e_{\text{in}}$. 
Thus in the whole $j$-th nested Braess network we have the set of used edges $S_1 = \emptyset$ for $q=0$ and $S_i := \tilde{S}_i \cup \{ (v_0, v_1), (v_{2l}, v_{2l+1}) \}$ for the demand rates $\tilde{q}_i$ for $i=2,3,\dotsc,2^{q}$.

Lemma~\ref{lem:nestedBraessHighDemand} yields that for any demand rate $q \geq 2 \cdot 10^{j-1}$, only the outer edges are used. In particular there is no flow through the inner components, i.e. $x_{e_{\text{in}}} = 0$, at demand rate $q = 2 \cdot 10^{j-1}$. We have seen before that for $q = 3 \cdot 10^{j-2}$ all flow is routed through the inner component, i.e. $x_{e_{\text{in}}} = 3 \cdot 10^{j-2}$. Since we know that all flow functions are continuous, there are demand rates $\hat{q}_1, \dotsc, \hat{q}_{2^j} \geq 3 \cdot 10^{j-2}$ such that $x_{e_{\text{in}}} ( \hat{r}_i ) = \tilde{q}_i$ and thus inducing the sets of used edges $\tilde{S}_i$ in the inner component and $S_{2^j+i} := \tilde{S}_i \cup \{ (v_0, v_1), (v_0, v_{2l}), (v_1, v_{2l+1}), (v_{2l}, v_{2l+1}) \}$ in the whole network.

So overall we have $2^{j}$ different sets of used edges for demand rates $q < 3 \cdot 10^{j-2}$ and $2^j$ different sets of used edges for demand rates $q \geq 3 \cdot 10^{j-2}$.
\end{proof}

Note, that the previous arguments also hold true if we only admit strictly increasing cost functions. Then the same construction can be made, but all constant cost functions can be replaced by cost functions with a slope of $\epsilon$.

\section{Computing direction vectors via convex programming}
\label{app:convex-program}

We here briefly discuss another approach of computing the desired directions $\Delta \vec \pi$ and $\Delta \vec x$ with convex programming.

Assume $\vec x^{\tilde{q}}$ is a Wardrop equilibrium for some demand rate $\tilde{q} \geq 0$ with some potential $\vec \pi^{\tilde{q}}$. Then define the values
$a_{e}^+ := \lim_{x \downarrow x_e} \frac{l_e(x) - l_e( x^{\tilde{q}}_e )}{x - x^{\tilde{q}}_e}$ and $a_{e}^- := \lim_{x \uparrow x_e} \frac{l_e(x) - l_e( x^{\tilde{q}}_e )}{x - x^{\tilde{q}}_e}$. So $a_e^+$ is the slope of the piece-wise linear part of $l_e$ above of $x^{\tilde{q}}_e$ a d $a_e^-$ is the slope of the piece-wise linear part of $l_e$ below of $x^{\tilde{q}}_e$. If $x^{\tilde{q}}_e$ is no breakpoint, then $a_e^+$ and $a_e^-$ coincide. Let $\vec A^+ := \diag(a_{e_1}^+, \dotsc, a_{e_m}^+)$ and $\vec A^- := \diag(a_{e_1}^-, \dotsc, a_{e_m}^-)$. Then we consider the following convex program.
\begin{align} 
\min \; \frac{1}{2} (\vec x^+)^{\top} \vec A^+ \vec x^+ &+ \frac{1}{2} (\vec x^-)^{\top} \vec A^- \vec x^- \notag \\
\text{s.t.} \qquad \qquad \quad
\vec \Gamma \, \vec x &= \Delta \vec y \notag \\
\vec x &= \vec x^+ + \vec x^- \label{eq:convexProgramSC} \\
\vec x^+ &\geq 0 \notag \\
\vec x^- &\leq 0 \notag
\end{align}

\begin{theorem} \label{thm:convexProgramSC}
Let $\vec x^{\tilde{q}}$ be a Wardrop equilibrium for some demand rate $\tilde{q} \geq 0$ and let $(\hat{\vec x}, \hat{\vec x}^+, \hat{\vec x}^-)$ be the optimal solution of the convex program~\eqref{eq:convexProgramSC}. Then there is $\epsilon > 0$ such that
$\vec{x} (q) := \vec x^{\tilde{q}} + \hat{\vec x} \, (q - \tilde{q})$
is a Wardrop equilibrium flow for every $\tilde{q} \leq q \leq \tilde{q} + \epsilon$.
\end{theorem}

\begin{proof}[Proof of Theorem~\ref{thm:convexProgramSC}]
Let $(\hat{\vec x}, \hat{\vec x}^+, \hat{\vec x}^-)$ be an optimal solution of the convex program. Then the Karush-Kuhn-Tucker conditions imply the existence of vectors $\vec \lambda \in \R^n$ and $\vec \mu, \vec \nu, \vec \eta \in \R^m$ with $\vec \nu, \vec \eta \geq 0$ such that
\begin{align*}
\nabla \frac{1}{2} (\hat{\vec x}^+)^{\top} \vec A^+ \hat{\vec x}^+ &+ \frac{1}{2} (\hat{\vec x}^-)^{\top} \vec A^- \hat{\vec x}^- 
+ \vec \lambda^{\top} \, \nabla (\Delta \vec y - \vec \Gamma \, \hat{\vec x}) \\
&+ \vec \mu^{\top} \, \nabla (\hat{\vec x} - \hat{\vec x}^+ - \hat{\vec x}^-)
- \vec \nu^{\top} \, \hat{\vec x}^+
+ \vec \eta^{\top} \, \hat{\vec x}^-
 = 0
.
\end{align*}
This implies 
\begin{align*}
\mu_e - \vec \lambda^{\top} \, \vec \gamma_e  &= 0 \\
a_e^+ \hat{x}^+_e - \mu_e - \nu_e &= 0 \\
a_e^- \hat{x}^-_e - \mu_e + \eta_e &= 0
\end{align*}
for every edge $e \in E$. In particular, if $\hat{x}_e^+ > 0$, we have $\nu_e = 0$ and thus $a_e^+ \hat{x}^+_e = \vec \lambda^{\top} \, \vec \gamma_e$ and, if $\hat{x}_e^- < 0$, we have $\eta_e = 0$ and thus $a_e^- \hat{x}^-_e = \vec \lambda^{\top} \, \vec \gamma_e$.
This also implies, since $a_e^+, a_e^- > 0$, that $\hat{x}^+_e$ and $\hat{x}^-_e$ can not both be non-zero. This means that either $a_e^+ \hat{x}_e = \vec \lambda^{\top} \, \vec \gamma_e$ (if $\hat{x}_e > 0$) or $a_e^- \hat{x}_e = \vec \lambda^{\top} \, \vec \gamma_e$ (if $\hat{x}_e < 0$) holds. 
If $\hat{x}_e^{+} = \hat{x}_e^{-} = 0$, then $\mu_e = 0$ and thus $a_e^+ \hat{x}_e = a_e^- \hat{x}_e = 0 = \vec \lambda^{\top} \gamma_e$.

For $q \geq \tilde{q}$ we define a flow $\vec x (q)$ and a potential $\vec \pi (q)$ by
\begin{align*}
\vec x (q) &:= \vec x^{\tilde{q}} + \hat{\vec x} \, (q - \tilde{q}) \\
\vec \pi (q) &:= \vec \pi^{\tilde{q}} + \vec \lambda \, (q - \tilde{q}) 
\end{align*}
where $\vec \pi^{\tilde{q}}$ is the (shortest path) potential of the flow $\vec x^{\tilde{q}}$.
By definition, $\vec x (q)$ is a flow that satisfies the demand rate $q$. Further, there is $\epsilon > 0$ such that for every edge $e \in E$ we have $l_e(x) = a_e^+ x + b_e^+$ for all $x_e^{\tilde{q}} \leq x \leq x_e^{\tilde{q}} + \epsilon$ and $l_e(x) = a_e^- x + b_e^-$ for all $x_e^{\tilde{q}} - \epsilon \leq x \leq x_e^{\tilde{q}}$.

Let $\tilde{q} \leq q \leq \tilde{q} + \epsilon$ and assume that $\hat{x}_e \geq 0$. Then $x_e (q) \geq x^{\tilde{q}}_e$ and thus
\begin{align*}
l_e ( x_e (q) ) &= a^+_e \, x_e (q) + b^+_e \\
&= a^+_e x^{\tilde{q}}_e + b^+_e + a^+_e \hat{x}_e \, (q-\tilde{q}) \\
&= l_e (x^{\tilde{q}}_e) + \vec \gamma_e^{\top} \, \vec \lambda \, (q-\tilde{q}) \\ 
&= \vec \gamma_e^{\top} \vec \pi^{\tilde{q}} + \vec \gamma_e^{\top} \, \vec \lambda \, (q-\tilde{q})
= \vec \pi (q) .
\end{align*}
The same can be shown in the case $\hat{x}_e < 0$ with $a_e^{-}$. Hence, $\vec x (q)$ is a Wardrop equilibrium by Lemma~\ref{lem:characterisationWEPotentials} with the potential $\vec{\pi} (q)$. The solution $\hat{\vec x}$ of the convex program is thus the direction vector $\Delta \vec x$ in the flow space while the dual variable $\vec{\lambda}$ corresponds to the direction vector $\Delta \vec \pi$ in the potential space.
\end{proof}

This shows that we can compute the direction $\Delta \vec x = \hat{\vec x}$ by solving the quadratic program above.

\clearpage

\section{Further Examples}
\label{app:examples}

This section contains some concrete examples illustrating the basic algorithm, the handling of degeneracy with the lexicographic rule, and ambiguous regions.

\subsection{An Example for the basic algorithm with continuous costs}
\begin{example} \label{ex:simpleUndirected}
We want to compute the undirected Wardrop equilibrium flow functions in the graph depicted in Figure~\ref{fig:simpleUndirectedExample} with the following cost functions:
\begin{align*}
l_{e_1} (x) =
\begin{cases}
x & \text{if } x < 1, \\
2x - 1 & \text{if } x \geq 1,
\end{cases}
\;\;
l_{e_2} (x) =
\begin{cases}
x & \text{if } x < 2, \\
2x - 2 & \text{if } x \geq 2,
\end{cases}
\;\;
l_{e_3} (x) =
\begin{cases}
2x & \text{if } x < 2, \\
x + 2 & \text{if } x \geq 2.
\end{cases}
\end{align*}
Since $\pi_s = 0$ we can represent the whole potential space in a two-dimensional space (see figure~\ref{fig:simpleUndirectedExample}(b)).
We start with the potential $\vec \pi^0 = (0,0,0)^{\top}$ which lies in the region $R_{(1,1,1)^{\top}}$. We then get
\begin{align*}
\vec C_{(1,1,1)^{\top}} =
\begin{bmatrix}
1 & 0 & 0 \\
0 & 1 & 0 \\
0 & 0 & \nicefrac{1}{2}
\end{bmatrix}
\end{align*}
and can thus compute
\begin{align*}
\vec L_{(1,1,1)^{\top}} =
\begin{bmatrix}
\nicefrac{3}{2} & -1 & -\nicefrac{1}{2} \\
-1 & 2 & -1 \\
-\nicefrac{1}{2} & -1 & \nicefrac{3}{2}
\end{bmatrix}
\quad
\hat{\vec L}_{(1,1,1)^{\top}}^{-1} =
\begin{bmatrix}
2 & -1 \\
-1 & \nicefrac{3}{2}
\end{bmatrix}^{-1}
=
\begin{bmatrix}
\nicefrac{3}{4} & \nicefrac{1}{2} \\
\nicefrac{1}{2} & 1
\end{bmatrix}
.
\end{align*}
From the last column of $\hat{\vec L}_{(1,1,1)^{\top}}^{-1}$ we obtain $\Delta \vec  \pi^0 = (0, \nicefrac{1}{2}, 1)^{\top}$ and can compute
\begin{align*}
\epsilon(e) =
\begin{cases}
\frac{1 - 0}{\nicefrac{1}{2} - 0} = 2 & \text{for } e = e_1, \\
\frac{2 - 0}{1-\nicefrac{1}{2}} = 4 & \text{for } e = e_2, \\
\frac{4 - 0}{1 - 0} = 4  & \text{for } e = e_3.
\end{cases}
\end{align*}
Thus $\epsilon = 2$ and for $\sCurveB^1 = 0 + 2$ the solution curve hits a unique boundary induced by the edge $e_1$ and enters the region $R_{(2,1,1)^{\top}}$ in the potential point $\vec{\pi}^1 = \vec{\pi}^0 + \epsilon \Delta \vec{\pi}^0 = (0,1,2)^{\top}$.
The matrix $\vec C$ changes to
\begin{align*}
\vec C_{(2,1,1)^{\top}} =
\begin{bmatrix}
\nicefrac{1}{2} & 0 & 0 \\
0 & 1 & 0 \\
0 & 0 & \nicefrac{1}{2}
\end{bmatrix}
\end{align*}
and we get $\Delta \vec \pi^1 = (0, \nicefrac{4}{5}, \nicefrac{6}{5})^{\top}$. The minimal epsilon $\epsilon = \nicefrac{5}{3} = \epsilon(e_3)$ is attained for edge $e_3$. So at $\sCurveB^2 = 2 + \nicefrac{5}{3} = \nicefrac{11}{3}$ the solution curve enters the region $R_{(2,1,2)^{\top}}$ and we get $\Delta \vec \pi^2 = (0, \nicefrac{1}{2}, \nicefrac{3}{4})^{\top}$ and $\epsilon=\nicefrac{4}{3}$. Finally, at $\sCurveB^3 = \nicefrac{15}{3}$ the solution curve enters the region $R_{(2,2,2)^{\top}}$ with $\Delta \vec \pi^3 = (0, \nicefrac{2}{5}, \nicefrac{4}{5})^{\top}$. In this region the values $\epsilon (e)$ are infinite for all edges $e$ and the algorithm stops.
\begin{figure}[t]
\centering
\begin{subfigure}[b]{.4\textwidth}
\begin{center}
\begin{tikzpicture}
\draw (0,0) node[solid] (s) {} node [below left] {$s$};
\draw (2,1) node[solid] (v) {} node [above] {$v$};
\draw (4,0) node[solid] (t) {} node [below right] {$t$};

\draw[->]	(s) edge node[midway, above left] {$e_1$} (v) 
			(v) edge node[midway, above right] {$e_2$} (t) 
			(s) edge node[midway, below] {$e_3$} (t);
\end{tikzpicture}
\caption{Graph.}
\end{center}
\end{subfigure}
\hspace{.02\textwidth}
\begin{subfigure}[b]{.55\textwidth}
\begin{center}
\begin{tikzpicture}[scale=.65]
\draw[thin, dashed] (1,-.5) -- (1,6.5)
			(-1,4) -- (6,4)
			(-1,1) -- (4.5,6.5);
\path[fill=color1,fill opacity=0.25] (-1,-0.5) -- (1,-0.5) -- (1,3) -- (-1,1) -- (-1,0);
\node[color=color1!30!black] at (-1,.5) {$R_{(1,1,1)^{\top}}$};
\path[fill=color3,fill opacity=0.25] (-1,1) -- (1,3) -- (1,4) -- (-1,4) -- (-1,1);
\node[color=color3!30!black] at (-1,3) {$R_{(1,2,1)^{\top}}$};
\path[fill=color4,fill opacity=0.25] (-1,4) -- (1,4) -- (1,6.5) -- (-1,6.5) -- (-1,4);
\node[color=color4!30!black] at (-1,5.5) {$R_{(1,2,2)^{\top}}$};
\path[fill=color2,fill opacity=0.5] (1,-0.5) -- (6,-0.5) -- (6,4) -- (2,4) -- (1,3) -- (1,-0.5);
\node[color=color2!30!black] at (4.5,2) {$R_{(2,1,1)^{\top}}$};
\path[fill=color6,fill opacity=0.25] (2,4) -- (6,4) -- (6,6.5) -- (4.5,6.5) -- (2,4);
\node[color=color6!30!black] at (5,5.5) {$R_{(2,1,2)^{\top}}$};
\path[fill=color5,fill opacity=0.25] (1,4) -- (2,4) -- (4.5,6.5) -- (1,6.5) -- (1,4);
\node[color=color5!30!black] at (2.25,5.5) {$R_{(2,2,2)^{\top}}$};
\path[fill=color8,fill opacity=0.25] (1,3) -- (2,4) -- (1,4) -- (1,3);
\node[color=color8!30!black] at (1.75,7.35) {$R_{(2,2,1)^{\top}}$};
\draw[ultra thin] (1.15,7) -- (1.15,3.5);
\draw (1,2) node[solid] (cross1) {};
\draw (7/3,4) node[solid] (cross2) {};
\draw (3,5) node[solid] (cross3) {};
\draw[thick]	(0,0)--(1,2)--(7/3,4)--(3,5)--(3.75,6.5);
\draw[thick, ->]
	(-1,0) -- (6.4,0) node[anchor=west] {$\pi_v$};
\draw[thick, ->]
	(0,-0.5) -- (0,6.9) node[anchor=south] {$\pi_t$};
\end{tikzpicture}
\end{center}
\caption{The Solution curve in the potential space.}
\end{subfigure}
\caption{(a) Graph and (b) solution curve in the potential space of Example~\ref{ex:simpleUndirected}. The graph in (a) is undirected, arrows indicate edge orientation. The solid line in (b) shows the solution curve in the potential space which is piece-wise linear with breakpoints at the boundary crossings shown as dots; dashed lines correspond to boundaries.}
\label{fig:simpleUndirectedExample}
\end{figure}
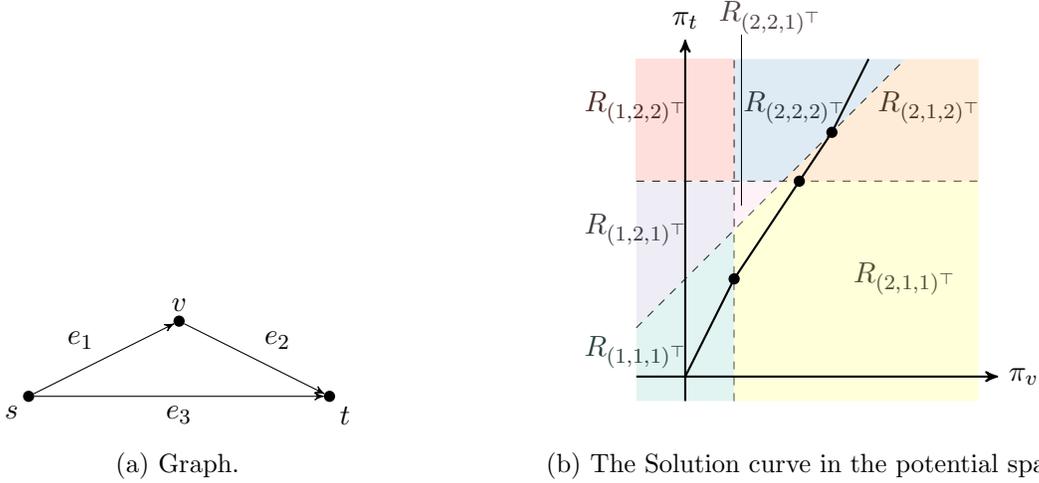
\begin{table}[t]
\setlength\tabcolsep{8pt}
\begin{center}
\begin{tabular}{c|c|c|c|c|c|c|c|c}
$i$ & $\lambda^i$ & $\vec{t}$ & $\vec{\pi}^i$ & $\vec{x}^i$ & $\vec{c}$ & $\vec{L}_{\vec{t}}$ & $\hat{\vec{L}}_{\vec{t}}^{-1}$ & $\epsilon$  \\
\hline
\rule{0pt}{2.5em}%
	$0$ & 	$0$		&	$\begin{bmatrix} 1 \\ 1 \\ 1 \end{bmatrix}$ &
$\begin{bmatrix} 0 \\ 0 \\ 0 \end{bmatrix}$ & $\begin{bmatrix} 0 \\ 0 \\ 0 \end{bmatrix}$ &
$\begin{bmatrix} 1 \\ 1 \\ \nicefrac{1}{2} \end{bmatrix}$ &
$
\begin{bmatrix}
\nicefrac{3}{2} & -1 & -\nicefrac{1}{2} \\
-1 & 2 & -1 \\
-\nicefrac{1}{2} & -1 & \nicefrac{3}{2}
\end{bmatrix}
$ &
$\begin{bmatrix} \nicefrac{3}{4} & \textcolor{blue}{\nicefrac{1}{2}} \\ \nicefrac{1}{2} & \textcolor{blue}{1} \end{bmatrix}$ &
$\begin{bmatrix} \textcolor{blue}{2} \\ 4 \\ 4 \end{bmatrix}$ \\[1.5em]
\hline
\rule{0pt}{2.5em}%
	$1$ & 	$2$		&	$\begin{bmatrix} 2 \\ 1 \\ 1 \end{bmatrix}$ &
$\begin{bmatrix} 0 \\ 1 \\ 2 \end{bmatrix}$ & $\begin{bmatrix} 1 \\ 1 \\ 1 \end{bmatrix}$ &
$\begin{bmatrix} \nicefrac{1}{2} \\ 1 \\ \nicefrac{1}{2} \end{bmatrix}$ &
$
\begin{bmatrix}
1 & -\nicefrac{1}{2} & -\nicefrac{1}{2} \\
-\nicefrac{1}{2} & \nicefrac{3}{2} & -1 \\
-\nicefrac{1}{2} & -1 & \nicefrac{3}{2}
\end{bmatrix}
$ &
$\begin{bmatrix} \nicefrac{6}{5} & \textcolor{blue}{\nicefrac{4}{5}} \\ \nicefrac{4}{5} & \textcolor{blue}{\nicefrac{6}{5}} \end{bmatrix}$ &
$\begin{bmatrix} \infty \\ \nicefrac{5}{2} \\  \textcolor{blue}{\nicefrac{5}{3}} \end{bmatrix}$ \\[1.5em]
\hline
\rule{0pt}{2.5em}%
	$2$ & 	$\nicefrac{11}{3}$		&	$\begin{bmatrix} 2 \\ 1 \\ 2 \end{bmatrix}$ &
$\begin{bmatrix} 0 \\ \nicefrac{7}{3} \\ 4 \end{bmatrix}$ & $\begin{bmatrix} \nicefrac{5}{3} \\ \nicefrac{5}{3} \\ 2 \end{bmatrix}$ &
$\begin{bmatrix} \nicefrac{1}{2} \\ 1 \\ 1 \end{bmatrix}$ &
$
\begin{bmatrix}
\nicefrac{3}{2} & -\nicefrac{1}{2} & -1 \\
-\nicefrac{1}{2}  & \nicefrac{3}{2} & -1 \\
-1 & -1 & 2
\end{bmatrix}
$ &
$\begin{bmatrix} 1 & \textcolor{blue}{\nicefrac{1}{2}} \\ \nicefrac{1}{2} & \textcolor{blue}{\nicefrac{3}{4}} \end{bmatrix}$ &
$\begin{bmatrix} \infty \\ \textcolor{blue}{\nicefrac{4}{3}} \\ \infty \end{bmatrix}$ \\[1.5em]
\hline
\rule{0pt}{2.5em}%
	$3$ & 	$5$		&	$\begin{bmatrix} 2 \\ 2 \\ 2 \end{bmatrix}$ &
$\begin{bmatrix} 0 \\ 3 \\ 5 \end{bmatrix}$ & $\begin{bmatrix} 2 \\ 2 \\ 3 \end{bmatrix}$ &
$\begin{bmatrix} \nicefrac{1}{2}  \\ \nicefrac{1}{2}  \\ 1 \end{bmatrix}$ &
$
\begin{bmatrix}
 \nicefrac{3}{2} & -\nicefrac{1}{2} & -1 \\
 -\nicefrac{1}{2} & 1 & -\nicefrac{1}{2} \\
 -1 & -\nicefrac{1}{2} & \nicefrac{3}{2} 
\end{bmatrix}
$ &
$\begin{bmatrix} \nicefrac{6}{5} & \textcolor{blue}{\nicefrac{2}{5}} \\ \nicefrac{2}{5} & \textcolor{blue}{\nicefrac{4}{5}} \end{bmatrix}$ &
$\begin{bmatrix} \infty \\ \infty \\ \infty \end{bmatrix}$
\end{tabular} 
\caption{The iterations  of the algorithm in Example~\ref{ex:simpleUndirected}. The highlighted values in the $\hat{\vec{L}}^{-1}$-matrix correspond to the potentials $\Delta \vec{\pi}$. 
}
\label{tab:ex:simpleUndirected}
\end{center}
\end{table}
Thus, we get the the potential function
\begin{align*}
\vec \pi(q) =
\begin{cases}
(0, \nicefrac{1}{2}, 1)^{\top} \, (q - 0) + (0,0,0)^{\top}  & \text{if } 0 \leq q < 2, \\
(0, \nicefrac{4}{5}, \nicefrac{6}{5})^{\top} \, (q - 2) + (0,1,2)^{\top} & \text{if } 2 \leq q < \nicefrac{11}{3}, \\
(0, \nicefrac{1}{2}, \nicefrac{3}{4})^{\top} \, (q - \nicefrac{11}{3}) + (0,\nicefrac{7}{3},4)^{\top} & \text{if } \nicefrac{11}{3}  \leq q < 5, \\
(0, \nicefrac{2}{5}, \nicefrac{4}{5})^{\top} \, (q - 5) + (0,3,5)^{\top} & \text{if } 5 \leq q .
\end{cases}
\end{align*}
and with this we can compute the Wardrop equilibrium flow functions
\begin{align*}
\vec{x} (q) =
\begin{cases}
(\nicefrac{1}{2}, \nicefrac{1}{2}, \nicefrac{1}{2})^{\top} \, (q - 0) + (0,0,0)^{\top} & \text{if } 0 \leq q < 2, \\
(\nicefrac{2}{5}, \nicefrac{2}{5}, \nicefrac{3}{5})^{\top} \, (q - 2) + (1,1,1)^{\top} & \text{if } 2 \leq q < \nicefrac{11}{3}, \\
(\nicefrac{1}{4}, \nicefrac{1}{4}, \nicefrac{3}{4})^{\top} \, (q - \nicefrac{11}{3}) + (\nicefrac{5}{3}, \nicefrac{5}{3}, 2)^{\top} & \text{if } \nicefrac{11}{3}  \leq q < 5, \\
(\nicefrac{1}{5}, \nicefrac{1}{5}, \nicefrac{4}{5})^{\top} \, (q - 5) + (2,2,3)^{\top} & \text{if } 5 \leq q .
\end{cases}
\end{align*}
All iterations of the algorithm are summarized in Table~\ref{tab:ex:simpleUndirected}.
\end{example}

\subsection{An Example for the lexicographic rule}

\begin{example} \label{ex:lexicographicRule}
Consider again the graph from Figure~\ref{fig:simpleUndirectedExample}. We use the costs
\[
l_{e_1} (x) =
\begin{cases}
x			&\text{if } x < 1, \\
5x - 4 		&\text{if } x \geq 1,
\end{cases}
\qquad
l_{e_2} (x) =
\begin{cases}
x			&\text{if } x < 1, \\
7x - 6 		&\text{if } x \geq 1,
\end{cases}
\qquad
l_{e_3} (x) =
\begin{cases}
x			&\text{if } x < 2, \\
12x - 22	&\text{if } x \geq 2.
\end{cases}
\]
\begin{figure}[t]
\centering
\begin{center}
\begin{tikzpicture}
\newcommand{\xMin}{-1.5} \newcommand{\xMax}{5}
\newcommand{\yMin}{-0.75} \newcommand{\yMax}{5.2}
\path[fill=color1,fill opacity=0.25] (\xMin,\yMin) -- ({max(\xMin,\yMin-1)}, {max(\xMin+1,\yMin)}) -- (1,2) -- (1,\yMin);
\path[fill=color2,fill opacity=0.25] (1, \yMin) -- (1,2) -- (\xMax, 2) -- (\xMax, \yMin);
\path[fill=color3,fill opacity=0.25] (1,2) -- ({min(\xMax,\yMax-1)},{min(\xMax+1,\yMax)}) -- (\xMax, \yMax) -- (\xMax, 2);
\path[fill=color4,fill opacity=0.25] (1,2) -- ({min(\xMax,\yMax-1)},{min(\xMax+1,\yMax)}) -- (\xMax, \yMax) -- (1, \yMax);
\path[fill=color5,fill opacity=0.25] (1,2) -- (1, \yMax) -- (\xMin, \yMax) -- (\xMin, 2);
\path[fill=color6,fill opacity=0.25] (1,2) -- ({max(\xMin,\yMin-1)}, {max(\xMin+1,\yMin)}) -- (\xMin, \yMin) -- (\xMin, 2);

\foreach \x in {0,...,4} {
    \foreach \y in {0,...,\x} {
    	\pgfmathsetmacro{\xpos}{-1.25 + \x/2} \pgfmathsetmacro{\ypos}{-0.5 + \y/2}
    	\pgfmathsetmacro{\vx}{1/3} \pgfmathsetmacro{\vy}{2/3}
    	\draw[color1, ultra thin, ->, opacity=0.5] (\xpos-\vx/4, \ypos-\vy/4) -- (\xpos+\vx/4, \ypos+\vy/4);
    }
}
\foreach \x in {0,...,7} {
    \foreach \y in {0,...,4} {
    	\pgfmathsetmacro{\xpos}{1.25 + \x/2} \pgfmathsetmacro{\ypos}{-0.5 + \y/2}
    	\pgfmathsetmacro{\vx}{5/7} \pgfmathsetmacro{\vy}{6/7}
    	\draw[color2!95!black, ultra thin, ->, opacity=0.8] (\xpos-\vx/4, \ypos-\vy/4) -- (\xpos+\vx/4, \ypos+\vy/4);
    }
}
\foreach \x in {0,...,6} {
	\pgfmathsetmacro{\yend}{min(\x,5)}
    \foreach \y in {0,...,\yend} {
    	\pgfmathsetmacro{\xpos}{1.75 + \x/2} \pgfmathsetmacro{\ypos}{2.25 + \y/2}
    	\pgfmathsetmacro{\vx}{2/3*5} \pgfmathsetmacro{\vy}{2/3*6}
    	\draw[color3, ultra thin, ->, opacity=0.5] (\xpos-\vx/16, \ypos-\vy/16) -- (\xpos+\vx/16, \ypos+\vy/16);
    }
}
\foreach \x in {0,...,4} {
    \foreach \y in {0,...,\x} {
    	\pgfmathsetmacro{\xpos}{3.25 - \x/2} \pgfmathsetmacro{\ypos}{4.75 - \y/2}
    	\pgfmathsetmacro{\vx}{5/2} \pgfmathsetmacro{\vy}{12/2}
    	\draw[color4, ultra thin, ->, opacity=0.5] (\xpos-\vx/16, \ypos-\vy/16) -- (\xpos+\vx/16, \ypos+\vy/16);
    }
}
\foreach \x in {0,...,4} {
    \foreach \y in {0,...,5} {
    	\pgfmathsetmacro{\xpos}{-1.25 + \x/2} \pgfmathsetmacro{\ypos}{2.25 + \y/2}
    	\pgfmathsetmacro{\vx}{3/5} \pgfmathsetmacro{\vy}{24/5}
    	\draw[color5, ultra thin, ->, opacity=0.5] (\xpos-\vx/16, \ypos-\vy/16) -- (\xpos+\vx/16, \ypos+\vy/16);
    }
}
\foreach \x in {0,...,3} {
    \foreach \y in {0,...,\x} {
    	\pgfmathsetmacro{\xpos}{0.25 - \x/2} \pgfmathsetmacro{\ypos}{1.75 - \y/2}
    	\pgfmathsetmacro{\vx}{1/9} \pgfmathsetmacro{\vy}{8/9}
    	\draw[color6, ultra thin, ->, opacity=0.5] (\xpos-\vx/4, \ypos-\vy/4) -- (\xpos+\vx/4, \ypos+\vy/4);
    }
}

\draw[thin, dashed] (1,\yMin) -- (1,\yMax);
\draw[thin, dashed] ({max(\xMin,\yMin-1)}, {max(\xMin+1,\yMin)}) -- ({min(\xMax,\yMax-1)},{min(\xMax+1,\yMax)});
\draw[thin, dashed] (\xMin,2) -- (\xMax,2);

\draw[thick, ->] (\xMin, 0) -- (\xMax, 0) node[right] {$\pi_v$};
\foreach \x in {1,...,4} {
	\draw[thick] (\x, 0.1) -- (\x, -0.1) node[below] {$\x$};
}
\draw[thick, ->] (0, \yMin) -- (0, \yMax) node[above] {$\pi_t$};
\foreach \x in {1,...,4} {
	\draw[thick] (0.1,\x) -- (-0.1, \x) node[left] {$\x$};
}

\node[anchor=south west, color1!80!black] at (\xMin+0.2, \yMin) {$R_{(1,1,1)}$};
\node[anchor=north east, color2!80!black] at (\xMax, 2) {$R_{(2,1,1)}$};
\node[anchor=south east, color3!80!black] at (\xMax,2) {$R_{(2,1,2)}$};
\node[anchor=east, color4!80!black] at ({min(\xMax,\yMax-1)-.25},{min(\xMax+1,\yMax)-.25}) {$R_{(2,2,2)}$};
\node[anchor=north west, color5!80!black] at (\xMin, \yMax) {$R_{(1,2,2)}$};
\node[anchor=north west, color6!80!black] at (\xMin, 2) {$R_{(1,2,1)}$};

\draw (0,0) node[solid] {};
\draw (1,2) node[solid] {};
\draw[thick] (0,0) -- (1,2) -- ({min(\xMax, 5/12*(\yMax+2/5))}, {min(12/5*\xMax+2/5, \yMax)});

\draw[thick, darkgreen] (9/16,3/4) -- (1,13/8) -- (21/16,2) -- (23/8, 31/8) -- ({min(\xMax, 5/12*(\yMax+121/40)}, {min(12/5*\xMax-121/40, \yMax)});
\draw[thick, red] (1/9,1/3) -- (8/9,17/9) -- (65/72,2) -- (1, 25/9) -- ({min(\xMax, 5/12*(\yMax-17/45)}, {min(12/5*\xMax+17/45, \yMax)});
\draw[darkgreen]
	node[solid] at (9/16,3/4) {}
	node[solid] at (1,13/8) {}
	node[solid] at (21/16,2) {}
	node[solid] at (23/8, 31/8) {};
\draw[red]
	node[solid] at (1/9,1/3) {}
	node[solid] at (8/9,17/9) {}
	node[solid] at (65/72,2) {}
	node[solid] at (1, 25/9) {};

\end{tikzpicture}
\end{center}
\caption{The potential space from Example~\ref{ex:lexicographicRule}. The black line is the solution curve, the green and the red line indicate two possible courses of perturbed solution curves (red for $\delta = \nicefrac{1}{3}$, green for $\delta = \nicefrac{3}{4}$). The light arrows in the regions indicate the potential directions (not to scale). The red region $R_{(2,2,2)}$ is the unique region where the potential direction is directed away from all intersecting boundaries.}
\label{fig:ex:lexicographicRule:potentialSpace}
\end{figure}
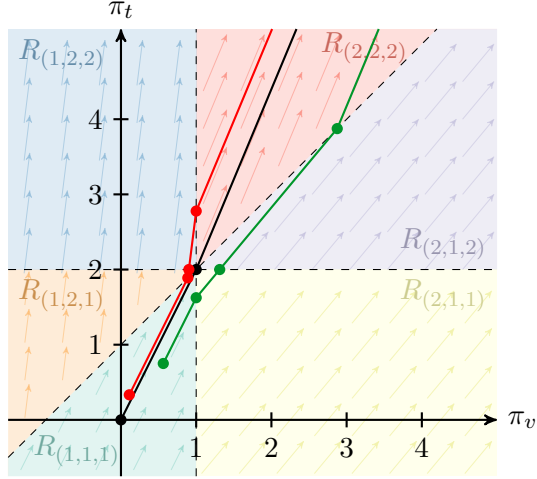
These cost functions induce three boundaries and six regions in the potential space as depicted in Figure~\ref{fig:ex:lexicographicRule:potentialSpace}. The solution curve (black line in Figure~\ref{fig:ex:lexicographicRule:potentialSpace}) starts in $\vec{\pi}^0 = (0,0,0)^\top \in R_{(1,1,1)}$. The direction vector in this region is $\Delta \vec{\pi}^0 = (0, \nicefrac{1}{3}, \nicefrac{2}{3})$ and, thus, at $\lambda = 3$ the solution curve hits the degenerate point $\vec{\pi}^1 = (0,1,2)^\top$. The only region where the direction $\Delta \vec{\pi}$ is directed away from all boundaries is region $R_{(2,2,2)}$---this is the region where the solution curve can proceed.

If, for some $\delta > 0$, the point $\vec{\pi}^0$ is perturbed by the vector $\vec{\delta} = (0, \delta^2, \delta)$ this leads to new, perturbed solution curves. Figure~\ref{fig:ex:lexicographicRule:potentialSpace} shows two possible perturbed solution curves, one for $\delta = \nicefrac{3}{4}$ (green) and one for $\delta = \nicefrac{1}{3}$ (red). Note, that for $\delta = 1/2$, the perturbed solution curve would start in $(0, \nicefrac{1}{4}, \nicefrac{1}{2})$ and thus follow the normal solution curve into the degenerate point. For all $\delta > \nicefrac{1}{2}$, the perturbed solution behaves essentially as the green solution curve in Figure~\ref{fig:ex:lexicographicRule:potentialSpace} -- for all $\delta < \nicefrac{1}{2}$ the perturbed solution curve follows essentially the course of the red solution curve in Figure~\ref{fig:ex:lexicographicRule:potentialSpace}.

Hence, the lexicographic rule should retrace the regions traversed by the red perturbed solution curve. Table~\ref{tab:ex:lexicographicRule:mVectors} lists all matrices $\vec{M}, \tilde{\vec{M}}$ and vectors $\vec{m}$ used for the lexicographic rule.

\begin{table}[t]
\setlength\tabcolsep{8pt}
\begin{center}
\begin{tabular}{m{1cm}|c|c|c|c|c|c}
Step $l$ & $\potentialVp[l]$ & $\tilde{\vec{M}}^l$ & $\vec{M}^l$ & $\vec{m}_{l, e_1}$ & $\vec{m}_{l, e_2}$ & $\vec{m}_{l, e_3}$ \\
\hline 
\centering 0 {\footnotesize $R_{(1,1,1)}$}
& 
\rule{0pt}{2.5em}%
$ \displaystyle
\begin{bmatrix}
0 \\ \nicefrac{1}{3} \\ \nicefrac{2}{3}
\end{bmatrix}
$
&
$ \displaystyle
\begin{bmatrix}
1 & 0 & 0 \\
0 & 1 & 0 \\
0 & 0 & 1
\end{bmatrix}
$
&
$ \displaystyle
\begin{bmatrix}
1 & 0 & 0 \\
0 & 1 & 0 \\
0 & 0 & 1
\end{bmatrix}
$
&
\textcolor{blue}{
$ \displaystyle
\begin{bmatrix}
3 \\ -3 \\ 0
\end{bmatrix}
$
}
&
\textcolor{blue}{
$ \displaystyle
\begin{bmatrix}
0 \\ 3 \\ -3
\end{bmatrix}
$
}
&
\textcolor{blue}{
$ \displaystyle
\begin{bmatrix}
\nicefrac{3}{2} \\ 0 \\ - \nicefrac{3}{2}
\end{bmatrix}
$
}
\\[1.5em]
\hline  
\centering 1 {\footnotesize $R_{(1,2,1)}$}
& 
\rule{0pt}{2.5em}%
$ \displaystyle
\begin{bmatrix}
0 \\ \nicefrac{1}{9} \\ \nicefrac{8}{9}
\end{bmatrix}
$
&
$ \displaystyle
\begin{bmatrix}
1 & 0 & 0 \\
0 & 2 & -1 \\
0 & 2 & -1
\end{bmatrix}
$
&
$ \displaystyle
\begin{bmatrix}
1 & 0 & 0 \\
0 & 2 & -1 \\
0 & 2 & -1
\end{bmatrix}
$
&
\textcolor{blue}{
$ \displaystyle
\begin{bmatrix}
9 \\ -18 \\ 9
\end{bmatrix}
$
}
&
$ \displaystyle
\begin{bmatrix}
0 \\ 0 \\ 0
\end{bmatrix}
$
&
\textcolor{blue}{
$ \displaystyle
\begin{bmatrix}
\nicefrac{9}{8} \\ -\nicefrac{18}{8} \\ \nicefrac{9}{8}
\end{bmatrix}
$
}
\\[1.5em]
\hline  
\centering 2 {\footnotesize $R_{(1,2,2)}$}
& 
\rule{0pt}{2.5em}%
$ \displaystyle
\begin{bmatrix}
0 \\ \nicefrac{1}{9} \\ \nicefrac{8}{9}
\end{bmatrix}
$
&
$ \displaystyle
\begin{bmatrix}
1 & 0 & 0 \\
\nicefrac{1}{8} & 1 & -\nicefrac{1}{8} \\
1 & 0 & 0
\end{bmatrix}
$
&
$ \displaystyle
\begin{bmatrix}
1 & 0 & 0 \\
\nicefrac{1}{8} & \nicefrac{7}{4} & -\nicefrac{7}{8} \\
1 & 0 & 0
\end{bmatrix}
$
&
\textcolor{blue}{
$ \displaystyle
\begin{bmatrix}
\nicefrac{63}{9} \\ - \nicefrac{63}{4} \\ \nicefrac{63}{8}
\end{bmatrix}
$
}
&
$ \displaystyle
\begin{bmatrix}
- \nicefrac{9}{8} \\ \nicefrac{9}{4} \\ -\nicefrac{9}{8}
\end{bmatrix}
$
&
$ \displaystyle
\begin{bmatrix}
0 \\ 0 \\ 0
\end{bmatrix}
$
\\[1.5em]
\hline  
\centering 3 {\footnotesize $R_{(2,2,2)}$}
& 
\rule{0pt}{2.5em}%
$ \displaystyle
\begin{bmatrix}
0 \\ \nicefrac{5}{2} \\ 6
\end{bmatrix}
$
&
$ \displaystyle
\begin{bmatrix}
1 & 0 & 0 \\
1 & 0 & 0 \\
8 & -8 & 1 
\end{bmatrix}
$
&
$ \displaystyle
\begin{bmatrix}
1 & 0 & 0 \\
1 & 0 & 0 \\
8 & -14 & 7 
\end{bmatrix}
$
&
$ \displaystyle
\begin{bmatrix}
0 \\ 0 \\ 0
\end{bmatrix}
$

&
$ \displaystyle
\begin{bmatrix}
- 2 \\ 4 \\ -2
\end{bmatrix}
$
&
$ \displaystyle
\begin{bmatrix}
- \nicefrac{7}{6} \\ \nicefrac{14}{6} \\ -\nicefrac{7}{6}
\end{bmatrix}
$
\end{tabular} 
\caption{The Matrices $\vec{M}, \tilde{\vec{M}}$ and the vectors $\vec{m}$ used in the lexicographic rule when bypassing the degenerate point in Example~\ref{ex:lexicographicRule}.}
\label{tab:ex:lexicographicRule:mVectors}
\end{center}
\end{table}

In the $0$-th step, the matrices $\vec{M}^0, \tilde{\vec{M}^0}$ are initialized as the identity matrix. Recall that the $\vec{m}$-vectors can be computed by
\[
\vec{m}_{l,e}^{\top} := - \frac{1}{\incidenceV_{ e }^{\top} \potentialDp[l]} \, \incidenceV_{ e }^{\top} \, \vec{M}^l
\quad \text{for } 0 \leq l \leq \alpha
.
\]
This means, the vector $\vec{m}_{l,e}^{\top}$ is obtained by taking the differences of the columns in $\vec{M}^l$ that correspond to the terminals of the edge $e$ (in the $0$-th step this is just the incidence vector $\gamma_e$) and scale this vector by the negative reciprocal of the potential-direction-difference $ \frac{1}{\incidenceV_{ e }^{\top} \potentialDp[l]} $ on that edge. For example, consider the $0$-th step and the edge $e_2 = (v,t)$. The vector $\vec{m}_{0, e_2}$ is the incidence vector $\gamma_{e_2}$ scaled by $-3$ which is the negative reciprocal of $\Delta \perturbed{\pi}_t -  \Delta \perturbed{\pi}_v = \nicefrac{2}{3} - \nicefrac{1}{3} = \nicefrac{1}{3}$. \\
In the $0$-th step, all vectors $\vec{m}_{0, e}$ are considered for finding the lexicographic minimum. In this example, the minimum is attained for $e^*_0 = e_2$. (Recall that we consider lexicographic order \emph{from bottom component to top component}!) \\
This means, for $\delta$ small enough, the perturbed solution curve crosses the boundary induced by $e_2$ first and moves into the region $R_{(1,2,1)}$.

In the next step, we compute the new matrix
\[
\tilde{\vec{M}}^1 := \vec{I}_n - \frac{1}{\incidenceV_{ e^*_{0} }^{\top} \potentialDp[0]} \potentialDp[0] \, \incidenceV_{ e^*_{0} }^{\top}
=
\begin{bmatrix}
1 & 0 & 0 \\
0 & 1 & 0 \\
0 & 0 & 1
\end{bmatrix}
-
\frac{1}{\nicefrac{2}{3} - \nicefrac{1}{3}}
\begin{bmatrix}
0 & 0 & 0 \\
0 & -\nicefrac{1}{3} & \nicefrac{1}{3} \\
0 & -\nicefrac{2}{3} & \nicefrac{2}{3}
\end{bmatrix}
=
\begin{bmatrix}
1 & 0 & 0 \\
0 & 2 & -1 \\
0 & 2 & -1
\end{bmatrix}
.
\]
Using this matrix and the direction vector $\potentialDp[1] = (0, \nicefrac{1}{9}, \nicefrac{8}{9})^\top$ we obtain the vectors $\vec{m}_{1, e}$ as before. In this (and all further iterations), we only consider lexicographic positive vectors. In particular, the vector $\vec{m}_{1, e_2}$ is not lexicographic positive since the perturbed solution curve has just crossed the boundary induced by this edge. In step $1$, the edge $e_3$ is the minimizer, we thus move on to the region $R_{(1,2,2)}$ in the next step.

In the second step, the only lexicographic positive vector is the vector $\vec{m}_{e_1}$, thus we proceed with the region $R_{(2,2,2)}$ in the next step.

Finally, in the third step, we obtain no lexicographic positive $\vec{m}$-vector and stop. The solution curve can proceed in this region since the potential direction in this region is directed away from all intersecting boundaries.
\end{example}

\subsection{An Example for an ambiguous region}

\begin{example} \label{ex:ambiguousRegion}
Consider again the graph from Figure~\ref{fig:simpleUndirectedExample}. We this time use the costs $l_{e_1}(x) = x$,
\begin{align*}
l_{e_2}(x) =
\begin{cases}
- \infty 	& \text{if } x < 0, \\
x			& \text{if } 0 \leq x < 1, \\
x + 2		& \text{if } 1 \leq x < 2, \\
\infty 		& \text{if } x > 2, 
\end{cases}
\qquad
l_{e_3}(x) =
\begin{cases}
2x			& \text{if } x \leq \nicefrac{3}{2}, \\
2x + 2		& \text{if } x > \nicefrac{3}{2}.
\end{cases}
\end{align*}
Note that the first function part of $l_{e_2}$ models a direction of edge $e_2$ and the last function part a maximal capacity of $2$. The inverse function are $l_{e_1}^{-1} (v) = v$ and 
\begin{align*}
l_{e_2}^{-1}(v) =
\begin{cases}
0		 	& \text{if } v < 0, \\
v			& \text{if } 0 \leq v < 1, \\
1			& \text{if } 1 \leq v < 3, \\
v-2 		& \text{if } 3 \leq v < 4, \\
2			& \text{if } v \geq 4,
\end{cases}
\qquad
l_{e_3}^{-1}(v) =
\begin{cases}
\nicefrac{1}{2} \, v 		& \text{if } v < 3, \\
\nicefrac{3}{2}				& \text{if } 3 \leq v < 5, \\
\nicefrac{1}{2} \, (v - 2)	& \text{if } v \geq 5.
\end{cases}
\end{align*}
inducing the regions depicted in figure \ref{fig:ambiguousRegionExample}. 
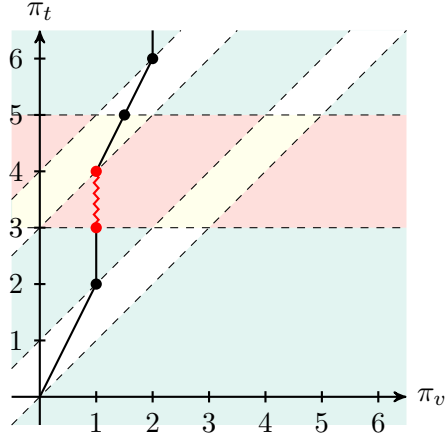
\begin{figure}[t]
\centering
\begin{tikzpicture}[scale=.75]
\newcommand{\xMin}{-.5} \newcommand{\xMax}{6.5}
\newcommand{\yMin}{-.5} \newcommand{\yMax}{6.5}

\fill[fill=color1,fill opacity=0.25] ({max(\xMin,\yMin-1)},{max(\xMin+1,\yMin}) -- (2,3)--(0,3)--({max(\xMin,\yMin-3)},{max(\xMin+3,\yMin}) -- (\xMin, \yMin);
\fill[fill=color1,fill opacity=0.25] (2,5) -- (4,5) -- ({min(\xMax,\yMax-1)},{min(\xMax+1,\yMax}) -- (\xMax, \yMax) -- ({min(\xMax,\yMax-3)},{min(\xMax+3,\yMax});
\fill[fill=color1,fill opacity=0.25] ({max(\xMin,\yMin-4)},{max(\xMin+4,\yMin}) -- ({max(-1,\xMin)},3) -- (\xMin, 3);
\fill[fill=color1,fill opacity=0.25]  (\xMin, \yMax) -- (\xMin,5) -- (1,5) -- ({min(\xMax,\yMax-4)},{min(\xMax+4,\yMax});
\fill[fill=color1,fill opacity=0.25] (\xMax, \yMin) -- ({max(\xMin,\yMin)},{max(\xMin,\yMin)}) -- (3,3) -- (\xMax, 3);
\fill[fill=color1,fill opacity=0.25] (5,5) -- ({min(\xMax, \yMax)}, {min(\xMax, \yMax)}) -- (\xMax, \yMax)--(\xMax, 5);
\fill[fill=color2,fill opacity=0.25] (1,5) -- (2,5) -- (0,3) -- (\xMin, 3) -- (\xMin, {max(3, \xMin+4)});
\fill[fill=color2,fill opacity=0.25] (2,3) -- (3,3) -- ({min(\xMax, 5)},{min(\xMax, 5}) -- ({min(\xMax,5)},5) -- (4,5);
\fill[fill=color4,fill opacity=0.25] (0,3) -- (2,5)--(4,5)--(2,3);
\fill[fill=color4,fill opacity=0.25] (\xMin,4 + \xMin) -- (1,5)--(\xMin,5);
\fill[fill=color4,fill opacity=0.25] (3,3) -- (5,5)--(\xMax,5)--(\xMax,3);

\draw[thin, dashed] 
			({max(\xMin,\yMin)},{max(\xMin,\yMin)}) -- ({min(\xMax,\yMax)}, {min(\xMax,\yMax)})
			({max(\xMin,\yMin-1)},{max(\xMin+1,\yMin}) -- ({min(\xMax,\yMax-1)},{min(\xMax+1,\yMax})
			({max(\xMin,\yMin-3)},{max(\xMin+3,\yMin}) -- ({min(\xMax,\yMax-3)},{min(\xMax+3,\yMax})
			({max(\xMin,\yMin-4)},{max(\xMin+4,\yMin}) -- ({min(\xMax,\yMax-4)},{min(\xMax+4,\yMax});
\draw[thin, dashed]
			(\xMin,3) -- (\xMax, 3)
			(\xMin,5) -- (\xMax, 5);

\draw[thick]	(0,0)--(1,2)--(1,3);
\draw[thick, decoration={zigzag,segment length=4,amplitude=.9}, decorate, red]	(1,3) -- (1,4);
\draw[thick]	(1,4)--(3/2,5)--(2,6)--(2,\yMax);

\draw (1,2) node[solid] (cross1) {};
\draw[red] (1,3) node[solid] (cross2) {};
\draw[red] (1,4) node[solid] (cross3) {};
\draw (3/2,5) node[solid] (cross4) {};
\draw (2,6) node[solid] (cross5) {};

\draw[thick, ->]
	(\xMin,0) -- (\xMax,0) node[anchor=west] {$\pi_v$};
\foreach \x in {1,2,...,6} {
	\draw[thick] (\x,0.1)--(\x,-0.1) node[below] {$\x$};
}	
\draw[thick, ->]
	(0,\yMin) -- (0,\yMax) node[anchor=south] {$\pi_t$};
\foreach \y in {1,2,...,6} {
	\draw[thick] (0.1,\y)--(-0.1, \y) node[left] {$\y$};
}	
\end{tikzpicture}
\caption{The potential space of example \ref{ex:ambiguousRegion}. The green regions are regions with $c_{e_2, t} = 0$, the yellow regions are regions with $c_{e_3, t} = 0$, and the red regions are ambiguous regions. Moving the potential along the zigzag line in the region $R_{(1,3,2)^{\top}}$ does not change the induced flow.}
\label{fig:ambiguousRegionExample}
\end{figure}
The ambiguous regions are highlighted red.

\begin{table}[t]
\setlength\tabcolsep{8pt}
\begin{center}
\begin{tabular}{c|c|c|c|c|c|c}
$i$ & $\lambda^i$ & $\vec{\pi}^i$ & $\vec{x}^i$ & Conductivity values $c_e$ & $\hat{\vec{L}}$ & $ \Delta \vec{\pi}$ \\
\hline
$0$ & $0$ & $\begin{bmatrix} 0 \\ 0 \\ 0 \end{bmatrix}$ & $\begin{bmatrix} 0 \\ 0 \\ 0 \end{bmatrix}$ &
\begin{tikzpicture}
\useasboundingbox (0,0) rectangle (1.5,1);
\node[solid] (s) at (0,0) {};
\node[solid] (v) at (0.75,0.5) {};
\node[solid] (t) at (1.5,0) {};
\draw[->]
	(s) edge node[midway, above left] {$1$} (v) 
	(v) edge node[midway, above right] {$1$} (t) 
	(s) edge node[midway, below] {$\nicefrac{1}{2}$} (t);
\end{tikzpicture}
&
$\begin{bmatrix} 2 & -1 \\ -1 & \nicefrac{3}{2} \end{bmatrix}$
&
$\begin{bmatrix} 0 \\ \nicefrac{1}{2} \\ 1 \end{bmatrix}$ \\[1.5em]
\hline
$1$ & $2$ & $\begin{bmatrix} 0 \\ 1 \\ 2 \end{bmatrix}$ & $\begin{bmatrix} 1 \\ 1 \\ 1 \end{bmatrix}$ &
\begin{tikzpicture}
\useasboundingbox (0,0) rectangle (1.5,1);
\node[solid] (s) at (0,0) {};
\node[solid] (v) at (0.75,0.5) {};
\node[solid] (t) at (1.5,0) {};
\draw[->]
	(s) edge node[midway, above left] {$1$} (v)
	(s) edge node[midway, below] {$\nicefrac{1}{2}$} (t);
\draw[gray, ->]
	(v) edge node[midway, above right] {$0$} (t) ;
\end{tikzpicture}
&
$\begin{bmatrix} 1 & 0 \\ 0 & \nicefrac{1}{2} \end{bmatrix}$
&
$\begin{bmatrix} 0 \\ 0 \\ 2 \end{bmatrix}$ \\[1.5em]
\hline
$2$ & $\nicefrac{5}{2}$ & $\begin{bmatrix} 0 \\ 1 \\ 3 \end{bmatrix}$ & $\begin{bmatrix} 1 \\ 1 \\ \nicefrac{3}{2} \end{bmatrix}$ &
\begin{tikzpicture}
\useasboundingbox (0,0) rectangle (1.5,1);
\node[solid] (s) at (0,0) {};
\node[solid] (v) at (0.75,0.5) {};
\node[solid] (t) at (1.5,0) {};
\draw[->]
	(s) -- (v) node[midway, above left] {$1$};
\draw[color4, ->]
	(s) edge node[midway, below] {$0$} (t)
	(v) edge node[midway, above right] {$0$} (t);
\end{tikzpicture}
&
$\begin{bmatrix} 1 & 0 \\ 0 & 0 \end{bmatrix}$
&
$\begin{bmatrix} 0 \\ 0 \\ 1 \end{bmatrix}$ \\[1.5em]
\hline
$3$ & $\nicefrac{5}{2}$ & $\begin{bmatrix} 0 \\ 1 \\ 4 \end{bmatrix}$ & $\begin{bmatrix} 1 \\ 1 \\ \nicefrac{3}{2} \end{bmatrix}$ &
\begin{tikzpicture}
\useasboundingbox (0,0) rectangle (1.5,1);
\node[solid] (s) at (0,0) {};
\node[solid] (v) at (0.75,0.5) {};
\node[solid] (t) at (1.5,0) {};
\draw[->]
	(s) edge node[midway, above left] {$1$} (v)
	(v) edge node[midway, above right] {$1$} (t);
\draw[gray, ->]
	(s) edge node[midway, below] {$0$} (t);
\end{tikzpicture}
&
$\begin{bmatrix} 2 & -1 \\ -1 & 1 \end{bmatrix}$
&
$\begin{bmatrix} 0 \\ 1 \\ 2 \end{bmatrix}$ \\[1.5em]
\hline
$4$ & $3$ & $\begin{bmatrix} 0 \\ \nicefrac{3}{2} \\ 5 \end{bmatrix}$ & $\begin{bmatrix} \nicefrac{3}{2} \\ \nicefrac{3}{2} \\ \nicefrac{3}{2} \end{bmatrix}$ &
\begin{tikzpicture}
\useasboundingbox (0,0) rectangle (1.5,1);
\node[solid] (s) at (0,0) {};
\node[solid] (v) at (0.75,0.5) {};
\node[solid] (t) at (1.5,0) {};
\draw[->]
	(s) edge node[midway, above left] {$1$} (v)
	(s) edge node[midway, below] {$\nicefrac{1}{2}$} (t)
	(v) edge node[midway, above right] {$1$} (t);
\end{tikzpicture}
&
$\begin{bmatrix} 2 & -1 \\ -1 & \nicefrac{3}{2} \end{bmatrix}$
&
$\begin{bmatrix} 0 \\ \nicefrac{1}{2} \\ 1 \end{bmatrix}$ \\[1.5em]
\hline
$6$ & $4$ & $\begin{bmatrix} 0 \\ 2 \\ 6 \end{bmatrix}$ & $\begin{bmatrix} 2 \\ 2 \\ 2 \end{bmatrix}$ &
\begin{tikzpicture}
\useasboundingbox (0,0) rectangle (1.5,1);
\node[solid] (s) at (0,0) {};
\node[solid] (v) at (0.75,0.5) {};
\node[solid] (t) at (1.5,0) {};
\draw[->]
	(s) edge node[midway, above left] {$1$} (v)
	(s) edge node[midway, below] {$\nicefrac{1}{2}$} (t);
\draw[gray, ->]
	(v) -- (t) node[midway, above right] {$0$};
\end{tikzpicture}
&
$\begin{bmatrix} 1 & 0 \\ 0 & \nicefrac{1}{2} \end{bmatrix}$
&
$\begin{bmatrix} 0 \\ 0 \\ 2 \end{bmatrix}$ 
\end{tabular}
\caption{The iteration of the algorithm in Example~\ref{ex:ambiguousRegion}. In iteration $i=2$ the reduced Laplacian matrix $\hat{L}$ is not invertible since the solution curve is in the ambiguous region $R_{(1,3,2)^{\top}}$---the edges with conductivity value $c_e = 0$ form an $s$-$t$-cut (marked red).}
\label{tab:ex:ambiguousRegion}
\end{center}
\end{table}

Table~\ref{tab:ex:ambiguousRegion} shows all iterations of the algorithm. In iteration $i=2$ the algorithm enters the region $R_{(1,3,2)^{\top}}$. Since the conductivity values $c_{e_2} = c_{e_3} = 0$ are zero in this region, the graph decomposes into the the two actively connected components $\mathcal{C}_1 = \{s,v\}$ and $\mathcal{C}_2 = \{t\}$. The matrix $\hat{\vec{L}}_{(1,3,2)^{\top}}$ is not invertible.
We use the direction $\Delta \vec{\pi}$ from Lemma~\ref{lem:enteringAmbiguousRegions} with $\Delta \pi_v = 1$ if $v \in \mathcal{C}_2$ and $\Delta \pi_v = 0$ if $v \in \mathcal{C}_1$, i.e. $\Delta \vec{\pi} = (0,0,1)^{\top}$.
Following this direction $\Delta \vec{\pi}$ the flow does not change until the solution curve reaches the potential $\vec \pi^2 = (0,1,4)^{\top}$ where the solution curve leaves the region. So the solution curve may jump along the red zigzag line in Figure~\ref{fig:ambiguousRegionExample} without changing the flow.
\end{example}

\clearpage

\section{Pseudocodes}
\label{app:pseudo}

\subsection{Algorithm without Degeneracy Handling}

\begin{algorithm}
\label{Basic algorithm}
\SetKwInOut{Input}{input}
\SetKwInOut{Output}{output}
\SetKw{True}{true}
\SetKw{Choose}{choose}
\SetKw{Return}{return}
\SetKw{And}{and}

\Input{$G = (V,E)$, $s_1 = 1$, $t_1 = n$}
\Output{solution curve}
\tikzmk{A}
$i := 1$ \;
$\hat{\vec \pi}^1 := \vec 0$ \;
$\vec t^i := \vec t(\pi^0)$ \;
$\Delta \hat{\vec y}  :=  \unitVec_n \in \R^{n-1}$ \;
\tikzmk{B}
\nameboxit{color1}{initialization}
\While{\True}{
	\tikzmk{A}$\Delta \hat{\vec \pi} := \hat{\vec L}_{\vec t^i}^{-1} \Delta \hat{\vec y}$\;
	$\Delta \vec \pi := (\pi_v)_{v \in V} \text{ with }
	\Delta \pi_v = \begin{cases}
	0 &\text{ if $v = s$,}\\
	\Delta \hat{\pi}_v &\text{ if $v \neq s$.}\\
	\end{cases}$ \;
	\tikzmk{B}
	\nameboxit{color2}{direction computation}
  	\tikzmk{A}
  	\ForEach{$e = (v,w) \in E$}{
		$\epsilon(e) :=
		\begin{cases}
		\frac{\sigma_{t^i_e+1} - (\pi_w^i - \pi_v^i)}{\Delta \pi_w - \Delta \pi_v}  & \text { if } \Delta \pi_w - \Delta \pi_v > 0,\\
		\frac{\sigma_{t^i_e\phantom{+1}} - (\pi_w^i - \pi_v^i)}{\Delta \pi_w - \Delta \pi_v} & \text{ if } \Delta \pi_w - \Delta \pi_v < 0,\\
		\infty & \text{ else.}
		\end{cases}$\;
	}
	$\epsilon := \min_{e \in E} \epsilon(e)$\;
	\eIf{$\epsilon = \infty$}{
		\Return \;
	}{
		\Choose $e^* \in E$ with $\epsilon(e^*) = \epsilon$\;
	}
 	\tikzmk{B}
 	\nameboxit{color3}{boundary selection}
 	\tikzmk{A}
 	$\vec{\pi}^{i+1} := \vec{\pi}^i + \epsilon \Delta \vec{\pi}$\;
	$\vec{t}^{i+1}_e := \vec{t}^i + \sgn( \vec{\gamma}_{e^*} \Delta\vec{\pi}^i ) \, \vec{u}_{e^*}$\;
	\textbf{compute $\hat{\vec L}_{t^{i+1}}^{-1}$}\;
	$i := i+1$\;
	\tikzmk{B}
	\nameboxit{color4}{boundary crossing}
}
\end{algorithm}

\clearpage

\subsection{Lexicographic Rule}

\begin{algorithm}[H]
\label{alg:degenerateAlgorithm}
\SetKwInOut{Input}{input}
\SetKwInOut{Output}{output}
\SetKw{True}{true}
\SetKw{Choose}{choose}
\SetKw{Compute}{compute}
\SetKw{Update}{update}
\SetKw{Return}{return}
\SetKw{And}{and}
\SetKw{Break}{break}
\SetKw{Not}{not}

  	\tikzmk{A}
  	\ForEach{$e = (v,w) \in E$}{
		$\epsilon(e) :=
		\begin{cases}
		\frac{\sigma_{t^i_e+1} - (\pi_w^i - \pi_v^i)}{\Delta \pi_w - \Delta \pi_v}  & \text { if } \Delta \pi_w - \Delta \pi_v > 0,\\
		\frac{\sigma_{t^i_e\phantom{+1}} - (\pi_w^i - \pi_v^i)}{\Delta \pi_w - \Delta \pi_v} & \text{ if } \Delta \pi_w - \Delta \pi_v < 0,\\
		\infty & \text{ else.}
		\end{cases}$\;
	}
	$\epsilon := \min_{e \in E} \epsilon(e)$\;
	$E^* := \argmax_{e \in E} \epsilon(e)$\;
	\tikzmk{B}
	\uIf{$\epsilon = \infty$}{
		\Return \;
	}
	\nameboxit{color3}{boundary selection}
	\uElseIf{$|E^*|=1$}{
		\tikzmk{A}
		\Choose $e^* \in E^*$ \;
		$\vec{\pi}^{i+1} := \vec{\pi}^i + \epsilon \Delta \vec{\pi}$\;
		$\vec{t}^{i+1}_e := \vec{t}^i + \sgn( \vec{\gamma}_{e^*} \Delta\vec{\pi}^i ) \, \vec{u}_{e^*}$\;
		\Compute $\hat{\vec L}_{t^{i+1}}^{-1}$\;
		\tikzmk{B}		
		\nameboxit{color4}{unique boundary crossing}
	}
	\Else{
		\tikzmk{A}
		$\vec{M} := \vec{I}_n$\;
		$l := 0$\;
		\Compute $\vec{m}^{\top}_{e} := - \frac{1}{\vec \gamma_{e^*}^{\top} \Delta \vec \pi} \vec \gamma_{e^*}^{\top} \vec{M}$ for all $e \in E^*$\;
		\tikzmk{B}
 		\nameboxit{color3}{boundary selection}
		\While{\Not $\vec{m}_e \lexleq \vec 0$ for all $e \in E^*$}{
			\tikzmk{A}
			\uIf{$l = 0$}{
				\Choose $e^* \in E^*$ with $\vec{m}_e$ lex. minimal\;
			}
			\Else{
				\Choose $e^* \in \{e \in E^* : \vec{0} \lexle \vec{m}_e \lexle \vec{\infty} \}$ with $\vec{m}_e$ lex. minimal\;
			}
			\tikzmk{B}
 			\nameboxit{color3}{boundary selection}
 			\tikzmk{A}
			$\vec{t}^{i}_e := \vec{t}^i + \sgn( \vec{\gamma}_{e^*} \Delta\vec{\pi}^i ) \, \vec{u}_{e^*}$\;
			\Update $\hat{\vec L}_{t^{i}}^{-1}$\;
			\tikzmk{B}
			\nameboxit{color4}{boundary crossing}
			\tikzmk{A}
			\vspace{-0.3cm}
			
			$\Delta \hat{\vec \pi} := \hat{\vec L}_{\vec t^i}^{-1} \Delta \hat{\vec y}$\;
			$\Delta \vec \pi := (\pi_v)_{v \in V} \text{ with }
			\Delta \pi_v = \begin{cases}
			0 &\text{ if $v = s$,}\\
			\Delta \hat{\pi}_v &\text{ if $v \neq s$.}\\
			\end{cases}$ \;
			\tikzmk{B}
			\nameboxit{color2}{direction computation}
			\tikzmk{A}
			$\vec M := ( \vec I_{n} - \frac{1}{\vec \gamma_{e^*}^{\top} \, \Delta \vec \pi} \, \Delta \vec \pi \, \vec \gamma_{e^*}^{\top} ) \, \vec M$\;
			\Compute $\vec{m}^{\top}_{e} := - \frac{1}{\vec \gamma_{e^*}^{\top} \Delta \vec \pi} \vec \gamma_{e^*}^{\top} \vec{M}$ for all $e \in E^*$\;
			$l := l + 1$\;
			\tikzmk{B}
 			\nameboxit{color3}{boundary selection}
		}
		\tikzmk{A}
		$\vec{\pi}^{i+1} := \vec{\pi}^i + \epsilon \Delta \vec{\pi}$\;
		$\vec{t}^{i+1} = \vec{t}$\;
		$\hat{\vec L}_{t^{i+1}}^{-1} := \hat{\vec L}_{t^{i}}^{-1}$\;
		\tikzmk{B}
		\nameboxit{color4}{boundary crossing}
	}
 	$i := i+1$\;
\end{algorithm}


\bibliography{bibliography} 


\end{document}